\documentclass[conference,10pt]{IEEEtran}

\usepackage{bm}
\usepackage{cite}
\usepackage{amsmath}
\usepackage[caption=false,font=footnotesize]{subfig}
\usepackage{url}
\usepackage{mathtools}
\usepackage[utf8]{inputenc}
\usepackage[T1]{fontenc}
\usepackage{hyperref}
\usepackage{bussproofs}
\usepackage{tikz}
\usepackage{tikz-cd}
\usepackage{amsthm}
\usepackage{amssymb}
\usepackage{stmaryrd}
\usepackage{esvect}

\DeclarePairedDelimiterX\braket[2]{\langle}{\rangle}{#1 \delimsize\vert #2}

\newenvironment{bprooftree}
  {\leavevmode\hbox\bgroup}
  {\DisplayProof\egroup}

\newcommand{\Ob}{\text{Ob}}%

\newcommand{\id}{\mathrm{id}}
\newcommand{\Id}{\mathrm{Id}}

\newcommand{\op}{\ensuremath{\mathrm{op}}}
\newcommand{\up}{\ensuremath{{\uparrow\,}}}
\newcommand{\down}{\ensuremath{\mathop{\downarrow}}}
\newcommand{\Da}{\mathord{\mbox{\makebox[0pt][l]{\raisebox{-.4ex}
                           {$\downarrow$}}$\downarrow$}}\,}

\newcommand{\AAA}{\ensuremath{\mathbf{A}}}
\newcommand{\BB}{\ensuremath{\mathbf{B}}}

\newcommand{\CC}{\ensuremath{\mathbf{C}}}

\newcommand{\DCPO}{\ensuremath{\mathbf{DCPO}}}
\newcommand{\dcpo}{\DCPO}
\newcommand{\dcpobs}{\ensuremath{\dcpo_{\perp!}}}

\newcommand{\RR}{\ensuremath{\mathbf{R}}}

\newcommand{\KL}{\ensuremath{\dcpo}_{\MM}}

\newcommand{\DOM}{\ensuremath{\mathbf{DOM}}}

\newcommand{\TD}{\ensuremath{\mathbf{TD}}}
\newcommand{\PD}{\ensuremath{\mathbf{PD}}}
\newcommand{\PDe}{\ensuremath{\PD_e}}
\newcommand{\FF}{\ensuremath{\mathcal{F}}}

\newcommand{\JJ}{\ensuremath{\mathcal{J}}}

\newcommand{\MM}{\ensuremath{\mathcal{M}}}
\newcommand{\VV}{\ensuremath{\mathcal{V}}}
\newcommand{\UU}{\ensuremath{\mathcal{U}}}

\newcommand{\LL}{\ensuremath{\mathcal{L}}}
\newcommand{\TTT}{\ensuremath{\mathcal{T}}}
\newcommand{\SSS}{\ensuremath{\mathcal{S}}}

\newcommand{\sfold}{\ensuremath{\mathrm{fold}}}

\newcommand{\sunfold}{\ensuremath{\mathrm{unfold}}}

\newcommand{\probto}[1]{\ensuremath{\xrightarrow{#1}}}

\newcommand{\Halt}{\mathrm{Halt}}

\newcommand{\Val}{\mathrm{Val}}
\newcommand{\Prog}{\mathrm{Prog}}
\newcommand{\ValRel}{\mathrm{ValRel}}
\newcommand{\ValReld}{\ValRel(X,A,e)}

\newcommand{\Paths}{\mathrm{Paths}}
\newcommand{\TPaths}{\mathrm{TPaths}}

\newcommand{\orp}[3]{#1\ \mathtt{or}_{#2}\ #3}

\newcommand{\FOLD}[1]{\mathbb{I}^{#1}}
\newcommand{\UNFOLD}[1]{\mathbb{E}^{#1}}

\newcommand{\lrb}[1]{{\llbracket #1 \rrbracket}}

\newcommand{\elrb}[2]{{\lVert #1 \rVert^{#2}}}
\newcommand{\elrbc}[1]{{\elrb{#1}{\vec C}}}
\newcommand{\elrbs}[1]{{\lVert #1 \rVert}}

\newcommand{\tleq}{\ensuremath{\vartriangleleft}}
\newcommand{\tleqd}{\ensuremath{\tleq_{X,A}^e}}
\newcommand{\qtleqd}{\ensuremath{\ \overline{\tleq_{X,A}^e}\ } }
\newcommand{\tleqone}{\ensuremath{\tleq_{X_1,A_1}^{e_1}}}

\newcommand{\tleqtwo}{\ensuremath{\tleq_{X_2,A_2}^{e_2}}}
\newcommand{\qtleqtwo}{\ensuremath{\ \overline{\tleq_{X_2,A_2}^{e_2}}\ }}

\newcommand{\btleq}{\ensuremath{\blacktriangleleft}}

\newcommand{\ol}[1]{\ensuremath{\ \overline{#1}\ }}

\newcommand{\defeq}{\stackrel{\textrm{{\scriptsize def}}}{=}}
\newcommand{\eqdef}{\defeq}

\newcommand{\naturalto}{\ensuremath{\Rightarrow}}

\newcommand{\ktimes}{%
  \mathbin{\vbox{\offinterlineskip
    \mathsurround=0pt
    \ialign{\hfil##\hfil\cr
      \normalfont\scalebox{1}{.}\cr
      \noalign{\kern-.05ex}
      $\times$\cr}
}}%
}

\newcommand{\kplus}{%
  \mathbin{\vbox{\offinterlineskip
    \mathsurround=0pt
    \ialign{\hfil##\hfil\cr
      \normalfont\scalebox{1}{.}\cr
      \noalign{\kern+.2ex}
      $+$\cr}
}}%
}

\newcommand{\kstar}{%
  \mathbin{\vbox{\offinterlineskip
    \mathsurround=0pt
    \ialign{\hfil##\hfil\cr
      \normalfont\scalebox{1}{.}\cr
      \noalign{\kern.1ex}
      $\star$\cr}
}}%
}

\newcommand{\kto}{%
  \mathbin{\vbox{\offinterlineskip
    \mathsurround=0pt
    \ialign{\hfil##\hfil\cr
      \normalfont\scalebox{2}{.}\cr
      \noalign{\kern-.5ex}
      $\to$\cr}
}}%
}

\newcommand{\ktwo}{%
  \mathbin{\vbox{\offinterlineskip
    \mathsurround=0pt
    \ialign{\hfil##\hfil\cr
      \normalfont\scalebox{1}{.}\cr
      \noalign{\kern.1ex}
      $\to$\cr}
}}%
}

\newcommand{\kcirc}{%
  \ensuremath{
    \mathbin{\raisebox{0.9pt}{\scalebox{0.7}{$\varodot$}}}
  }%
}%

\newcommand{\kid}{%
\ensuremath{\textbf{id}}
}

\makeatletter
\newsavebox{\@brx}
\newcommand{\llangle}[1][]{\savebox{\@brx}{\(\m@th{#1\langle}\)}%
\mathopen{\copy\@brx\kern-0.5\wd\@brx\usebox{\@brx}}}
\newcommand{\rrangle}[1][]{\savebox{\@brx}{\(\m@th{#1\rangle}\)}%
\mathclose{\copy\@brx\kern-0.5\wd\@brx\usebox{\@brx}}}
\makeatother

\newcommand{\pto}{\ensuremath{\rightharpoonup}}

\makeatletter
\def\moverlay{\mathpalette\mov@rlay}
\def\mov@rlay#1#2{\leavevmode\vtop{%
   \baselineskip\z@skip \lineskiplimit-\maxdimen
   \ialign{\hfil$\m@th#1##$\hfil\cr#2\crcr}}}
\newcommand{\charfusion}[3][\mathord]{
    #1{\ifx#1\mathop\vphantom{#2}\fi
        \mathpalette\mov@rlay{#2\cr#3}
      }
    \ifx#1\mathop\expandafter\displaylimits\fi}
\makeatother

\theoremstyle{plain}
\newtheorem{theorem}{Theorem}
\newtheorem{lemma}[theorem]{Lemma}
\newtheorem{proposition}[theorem]{Proposition}

\newtheorem{corollary}[theorem]{Corollary}

\theoremstyle{definition}
\newtheorem{definition}[theorem]{Definition}
\newtheorem{assumption}[theorem]{Assumption}
\newtheorem{example}[theorem]{Example}
\newtheorem{notation}[theorem]{Notation}

\theoremstyle{remark}
\newtheorem{remark}[theorem]{Remark}

\usetikzlibrary{decorations.pathmorphing}
\usetikzlibrary{decorations.markings}
\usetikzlibrary{decorations.pathreplacing}
\usetikzlibrary{arrows}
\usetikzlibrary{shapes}

\pgfdeclarelayer{edgelayer}
\pgfdeclarelayer{nodelayer}
\pgfsetlayers{edgelayer,nodelayer,main}

\tikzstyle{braceedge}=[decorate,decoration={brace,amplitude=10pt}]
\tikzstyle{square box}=[rectangle,fill=white,draw=black,minimum height=6mm,minimum width=6mm,yshift=0.7mm]
\tikzstyle{wire label}=[font=\footnotesize, auto,swap]

\tikzstyle{none}=[inner sep=0pt]
\tikzstyle{empty}=[rectangle,fill=none,draw=none]
\tikzstyle{scaled}=[rectangle,fill=none,draw=none, font=\small]

\tikzstyle{to}=[->,draw=black]
\tikzstyle{naturalto}=[-{Implies},double distance=1.5pt]
\tikzstyle{hook}=[right hook->, draw=black]

\tikzstyle{equal-arrow}=[double equal sign distance]

\tikzstyle{every picture}=[baseline=-0.25em]

\tikzset{every picture/.append
 style=baseline={([yshift=-.5ex]current bounding box.center)}}

\newcommand{\stikz}[2][1]{\scalebox{#1}{
\InputIfFileExists{#2}{}{\input{./tikz/#2}}
}}
\newcommand{\cstikz}[2][1]{\begin{center}\stikz[#1]{#2}\end{center}}

\begin{document}
\onecolumn
\title{Commutative Monads for Probabilistic Programming Languages}
\author{
  \IEEEauthorblockN{
    Xiaodong Jia\IEEEauthorrefmark{1}\IEEEauthorrefmark{3},
    Bert Lindenhovius\IEEEauthorrefmark{2},
    Michael Mislove\IEEEauthorrefmark{3} and
    Vladimir Zamdzhiev\IEEEauthorrefmark{4}
  }
  \IEEEauthorblockA{\IEEEauthorrefmark{1}
    School of Mathematics, Hunan University, Changsha, 410082, China
  }
  \IEEEauthorblockA{\IEEEauthorrefmark{2}
    Department of Knowledge-Based Mathematical Systems, Johannes Kepler Universität, Linz, Austria
  }  
  \IEEEauthorblockA{\IEEEauthorrefmark{3}
    Department of Computer Science, Tulane University, New Orleans, LA, USA
  }
  \IEEEauthorblockA{\IEEEauthorrefmark{4}
    Universit\'e de Lorraine, CNRS, Inria, LORIA, F 54000 Nancy, France
  }
}
\maketitle

\begin{abstract}
A long-standing open problem in the semantics of programming languages supporting
probabilistic choice is to find a commutative monad for probability on the
category DCPO. In this paper we present three such monads and a general
construction for finding even more. We show how to use these monads to provide
a sound and adequate denotational semantics for the Probabilistic FixPoint
Calculus (PFPC) -- a call-by-value simply-typed lambda calculus with
mixed-variance recursive types, term recursion and probabilistic choice. We
also show that in the special case where we consider continuous dcpo's, then
all three monads coincide with the valuations monad of Jones and we fully
characterise the induced Eilenberg-Moore categories by showing that they are
all isomorphic to the category of continuous Kegelspitzen of Keimel and
Plotkin. 
\end{abstract}

\section{Introduction}
\label{sec:intro}
Probabilistic methods now are  a staple of computation. The initial discovery
of randomized algorithms~\cite{rabin} was quickly followed by the definition of
Probabilistic Turing machines and related complexity classes~\cite{gill}.
There followed advances in a number of areas,
including, e.g., process calculi, probabilistic model checking and
verification~\cite{Baeten,LarsenSkou,Morgan96}, right through to the recent
development of  statistical probabilistic programming languages
(cf.~\cite{PPS,statonetal,vakaretal}),  not to mention the crucial role
probability plays in quantum programming languages~\cite{quant-semantics,cho}. 

\emph{Domain theory}, a staple of denotational semantics, has struggled to keep
up with these advances.  Domain theory encompasses two broad classes of
objects: \emph{directed complete partial orders (dcpo's)}, based on an
order-theoretic view of computation, and the smaller class of
\emph{(continuous) domains}, those dcpo's that also come equipped with a notion
of approximation.
However, adding probabilistic choice to the domain-theoretic approach has been
a challenge. The canonical model of (sub)probability measures in domain theory
is the family of \emph{valuations} -- certain maps from the lattice of open
subsets of a dcpo to the unit interval. It is well-known that these valuations
form a monad $\VV$ on  $\mathbf{DCPO}$ (the category of dcpo's and
Scott-continuous functions) and on $\mathbf{DOM}$ (the full
subcategory of $\dcpo$ consisting of domains) \cite{JonesP89,jones90}. 

In fact, the monad $\VV$ on $\mathbf{DOM}$ is commutative~\cite{jones90}, which
is important for two reasons: (1) its commutativity is
equivalent to the Fubini Theorem~\cite{jones90}, a cornerstone of integration
theory and (2) computationally, commutativity of a monad together with adequacy can
be used to establish contextual equivalences for effectful programs.
However, in order to do so, one typically needs a Cartesian closed category for the semantic model,
and $\mathbf{DOM}$ is not closed; in fact,  despite repeated attempts, 
it remains unknown whether there is \emph{any}  Cartesian closed category of
\emph{domains} on which $\VV$ is an endofunctor; this is the well-known
\emph{Jung-Tix Problem}~\cite{jungtix}.  On the other hand, it also is unknown
if the monad $\VV$ is commutative on the larger Cartesian closed category
$\DCPO$. In this paper, we offer a solution to this conundrum. 

\subsection{Our contributions}
We use \emph{topological methods} to construct a commutative valuations monad $\MM$ on $\DCPO$, 
as follows: it is straightforward to show the family $\SSS D$ of \emph{simple valuations on $D$} can be equipped with the structure of a commutative monad,
but $\SSS D$ is not a dcpo, in general. So, we complete $\SSS D$ by taking the smallest subdcpo $\MM D\subseteq \VV D$ that contains $\SSS D$.
This defines the object-mapping of a monad $\MM$ on $\DCPO$.
The unit, multiplication and strength of the monad $\MM$ at $D$ are given by the restrictions of the same operations of $\VV$ to $\MM D.$
Topological arguments then imply that $\MM$ is  a commutative valuations monad on $\dcpo$.

In fact, there are several completions of $\SSS D$ that give rise to commutative valuations monads on $\DCPO$. These completions are determined by so-called \texttt{K}-categories, introduced by Keimel and Lawson~\cite{KeimelL09}. This observation allows us to define two additional commutative valuations monads, $\mathcal W$ and $\mathcal P$, on $\DCPO$ simply by specifying their corresponding  \texttt{K}-categories. Finally, while we have identified three such \texttt{K}-categories, there likely are more that meet our requirements, each of which would define yet another commutative monad of valuations on $\DCPO$ containing $\SSS$.

With this background, we now summarise our main results.

\paragraph*{\textbf{Commutative monads}} A \texttt{K}-category is a full subcategory of the category $\mathbf{T_0}$ of $T_0$-spaces satisfying properties that imply it determines a \emph{completion} of each $T_0$-space among the objects of the \texttt{K}-category. For example, each \texttt{K}-category defines a completion of a poset endowed with its Scott topology, among the dcpo's in the \texttt{K}-category. In particular, each \texttt{K}-category determines a completion of the family $\SSS D$ when considered as a subset of $\VV D$, for each dcpo $D$. 

By specifying an additional constraint on \texttt{K}-categories, we can show the corresponding completions 
of $\SSS$ define commutative monads on $\DCPO$. 
We identify three commutative monads concretely: $\MM$, $\mathcal W$ and
$\mathcal P$, corresponding to the \texttt{K}-categories of d-spaces, that of
well-filtered spaces and that of sober spaces, respectively (see Theorem \ref{theorem:M is commutative} and Theorem \ref{theorem:kcofsiscommu}). As part of our
construction, we also prove the most general Fubini Theorem for dcpo's yet
available (see~Theorem~\ref{theorem:VPFubini}).

\paragraph*{\textbf{Eilenberg-Moore Algebras}}
All three of $\MM, \mathcal{W}$ and $\mathcal P$ restrict to monads on $\mathbf{DOM}$, where they coincide with $\VV$.
We characterize their Eilenberg-Moore categories over $\mathbf{DOM}$ by showing they are isomorphic to the category of continuous Kegelspitzen and Scott-continuous linear maps~\cite{keimelplotkin17}; 
this corrects an error in~\cite{jones90} (see~Remark~\ref{rem:kegelspitz} below).

On the larger category $\DCPO$, we show the Eilenberg-Moore algebras of our
monads $\mathcal M, W$ and $\mathcal P$ are Kegelspitzen
(see~Subsection~\ref{sub:M-V-relationship}). It is unknown if every Kegelspitze
is an $\mathcal M$-algebra, but we believe this to be the case.

\paragraph*{\textbf{Semantics}} 
We consider the \emph{Probabilistic FixPoint Calculus} ($\mathrm{PFPC}$) -- a
call-by-value simply-typed lambda calculus with mixed-variance recursive types,
term recursion and probabilistic choice (see Section \ref{sec:syntax}). We show
that each of the Kleisli categories of our three commutative monads is a sound
and computationally adequate model of $\mathtt{PFPC}$ (see Section
\ref{sec:semantics}). Moreover, we show that adequacy holds in a strong sense
(Theorem \ref{thm:strong-adequacy}), i.e., the interpretation of each
term is a (potentially infinite) convex sum of the values it reduces to.

\subsection{Related work}

The first dcpo model for probabilistic choice was given in~\cite{saheb}, but
this preceded Moggi's seminal work using Kleisli categories to model
computational effects~\cite{moggi-monads}. The work closest to ours is Jones' thesis \cite{jones90} (see also \cite{JonesP89}),
which considers the same language $\mathtt{PFPC}$,
but with a slightly different syntax. This work is based on an early version of
$\mathtt{FPC}$, and uses the Kleisli category of $\VV$ over $\mathbf{DCPO}$ as
the semantic model. While soundness and adequacy theorems are included in
\cite{jones90}, the proof of adequacy does not identify a
semantic space on which $\VV$ is commutative, instead offering arguments based
on the commutativity of $\SSS$, and on realizing  the valuations needed to
interpret the language as directed suprema of simple valuations. Our semantic
results improve those of Jones, because the commutativity of our monads
together with adequacy allows us to establish a larger class of contextual
equivalences.

Another related paper is \cite{stochastic-meet-pcf}, where the authors describe
a different construction for a commutative monad for probability.
The construction in \cite{stochastic-meet-pcf} is based on functional-analytic techniques similar to those
in~\cite{tix,Gou-csl07}, whereas ours is based on the topological
and categorical methods in~\cite{jiamis}.
Furthermore, the two constructions  yield distinct monads.
With our construction, we identify three probabilistic commutative
monads, study the structure of the induced Eilenberg-Moore and
Kleisli categories and then prove semantic results such as soundness and adequacy
for $\mathtt{PFPC}$. The work in \cite{stochastic-meet-pcf} constructs yet another
commutative monad that is used to study a different language (a real
$\mathtt{PCF}$-like language with sampling and conditioning) with a semantics
that reflects a concern for implementability and computability.

Other related work includes ~\cite{ehrhardtasson}, where the authors use
\emph{probabilistic coherence spaces} to provide a fully abstract model of a
probabilistic call-by-push-value language with recursive types. This work
builds on previous work \cite{PPCF} which describes a fully abstract model of
probabilistic $\mathtt{PCF}$ also based on probabilistic coherence spaces.
Recently, \emph{quasi-Borel spaces} were introduced in \cite{quasi-borel} and
they were later used to provide a sound and adequate model of $\mathtt{SFPC}$
(a statistical probabilistic programming language with recursive types,
sampling and conditioning) in \cite{vakaretal}. Compared to probabilistic
coherence spaces and quasi-Borel spaces, our methods are based on the
traditional domain-theoretic approach and its well-established connections to probability
theory~\cite{edalatdomint,edalatweak,alvman}; we hope to exploit these connections in future work.

The paper~\cite{rennelaprobFPC} uses Kegelsptizen to provide a sound and
adequate model for probabilistic $\mathtt{PCF}$. The author then discusses a
possible interpretation of a version of linear $\mathtt{PFPC}$ without contraction or a
!-modality (which means the system is strongly normalising), but \cite{rennelaprobFPC} does not
state any soundness, nor adequacy results for it.

\section{Syntax and Operational Semantics}
\label{sec:syntax}

In this section we describe the syntax and operational semantics of our language.
The language we consider is the \emph{Probabilistic FixPoint Calculus} ($\mathtt{PFPC}$).
The presentation we choose for $\mathtt{PFPC}$ is exactly the same as $\mathtt{FPC}$
\cite{fpc-syntax,fiore-thesis,fiore-plotkin} together with the addition of one
extra term ($M\ \mathtt{or}_p\ N$) for probabilistic choice. The same language is
also considered by Jones \cite{jones90}, but with a slightly different syntax.

\begin{figure*}
{\centering
\begin{tabular}{l l l l l l l l}
  Type Variables           & $X,Y$ &                & \multicolumn{4}{l}{ Term Variables \qquad $x,y$     }    \\ 
  Type Contexts            & $\Theta$               & ::= & $\cdot\ |\ \Theta, X$ & \multicolumn{4}{l}{  } \\
	Types                    & $A, B$                 & ::= & \multicolumn{5}{l}{$X$  | $A+B$ | $A \times B$ | $A \to B$ | $\mu X.A$} \\
  Term Contexts            & $\Gamma$               & ::= & $\cdot\ |\ \Gamma, x : A $ & \multicolumn{4}{l}{} \\
  Terms                    & $M,N$                  & ::= & \multicolumn{5}{l}{ $x\ |\ (M,N)\ |\ \pi_1 M\ |\ \pi_2 M\ |\ \mathtt{in}_1 M\ |\ \mathtt{in}_2 M\ |\  (\mathtt{case}\ M\ \mathtt{of}\ \mathtt{in}_1 x \Rightarrow N_1\ |\ \mathtt{in}_2 y \Rightarrow N_2)\ | $ } \\
                           &                        &     & \multicolumn{5}{l}{ $\lambda x. M\ |\ MN   \ |\ \mathtt{fold}\ M\ |\ \mathtt{unfold}\ M\ |\  M\ \mathtt{or}_p\ N $} \\
  Values                   & $V,W$                  & ::= & \multicolumn{5}{l}{ $x\  |\ (V,W)\ |\ \mathtt{in}_1 V\ |\ \mathtt{in}_2 V\ |\ \mathtt{fold}\ V\ |\ \lambda x. M  $ } 
\end{tabular}
\caption{Grammars for types, contexts and terms.}
\label{fig:syntax-grammars}
\[
    \begin{bprooftree}
    \AxiomC{$\Theta \vdash $}
    \UnaryInfC{$\Theta \vdash \Theta_i$}
    \end{bprooftree}
    \quad
    \begin{bprooftree}
    \AxiomC{$\Theta \vdash A$}
    \AxiomC{$\Theta \vdash B$}
    \RightLabel{$\star \in \{+, \times, \to \}$}
    \BinaryInfC{$\Theta \vdash A \star B$}
    \end{bprooftree}
    \quad
    \begin{bprooftree}
    \AxiomC{$\Theta, X \vdash A$}
    \UnaryInfC{$\Theta \vdash \mu X. A$}
    \end{bprooftree}
\]
\caption{Formation rules for types.}
\label{fig:type-syntax}
\[
  \begin{bprooftree}
  \AxiomC{\phantom{$\Gamma \vdash $}}
  \UnaryInfC{$\Gamma, x:A \vdash x: A$}
  \end{bprooftree}
  \begin{bprooftree}
  \def\ScoreOverhang{0.5pt}
  \AxiomC{$ \Gamma \vdash M : A$}
  \AxiomC{$ \Gamma \vdash N : A$}
  \BinaryInfC{$ \Gamma \vdash M\ \mathtt{or}_p\ N : A$}
  \end{bprooftree}
  \begin{bprooftree}
  \def\ScoreOverhang{0.5pt}
  \AxiomC{$ \Gamma \vdash M : A$}
  \AxiomC{$ \Gamma \vdash N : B$}
  \BinaryInfC{$ \Gamma \vdash (M, N) : A \times B$}
  \end{bprooftree}
  \ 
  \begin{bprooftree}
  \def\ScoreOverhang{0.5pt}
  \AxiomC{$ \Gamma \vdash M : A_1 \times A_2$}
  \RightLabel{$i \in \{1,2\}$}
  \UnaryInfC{$ \Gamma \vdash \pi_i M : A_i$}
  \end{bprooftree}
\]

\[
  \begin{bprooftree}
  \def\ScoreOverhang{0.5pt}
  \AxiomC{$ \Gamma \vdash M : A_i$}
  \RightLabel{$i \in \{1,2\}$}
  \UnaryInfC{$ \Gamma \vdash \mathtt{in}_{i} M : A_1+A_2$}
  \end{bprooftree}
  \begin{bprooftree}
  \def\ScoreOverhang{0.5pt}
  \AxiomC{$ \Gamma \vdash M : A_1+A_2$}
  \AxiomC{$ \Gamma, x : A_1 \vdash N_1 : B$}
  \AxiomC{$ \Gamma, y : A_2 \vdash N_2 : B$}
  \TrinaryInfC{$ \Gamma \vdash ( \mathtt{case}\ M\ \mathtt{of}\ \mathtt{in}_1 x \Rightarrow N_1\ |\ \mathtt{in}_2 y \Rightarrow N_2 ) : B$}
  \end{bprooftree}
\]

\[
  \begin{bprooftree}
  \def\ScoreOverhang{0.5pt}
  \AxiomC{$ \Gamma, x: A \vdash M : B$}
  \UnaryInfC{$ \Gamma \vdash \lambda x^A . M : A \to B$}
  \end{bprooftree}
  \begin{bprooftree}
  \def\ScoreOverhang{0.5pt}
  \AxiomC{$ \Gamma \vdash M : A \to B$}
  \AxiomC{$ \Gamma \vdash N : A$}
  \BinaryInfC{$ \Gamma \vdash MN : B$}
  \end{bprooftree}
  \begin{bprooftree}
  \def\ScoreOverhang{0.5pt}
  \AxiomC{$ \Gamma \vdash M : A[\mu X. A / X]$}
  \UnaryInfC{$ \Gamma \vdash \mathtt{fold}\ M: \mu X. A$}
  \end{bprooftree}
  \ 
  \begin{bprooftree}
  \def\ScoreOverhang{0.5pt}
  \AxiomC{$ \Gamma \vdash M : \mu X. A$}
  \UnaryInfC{$ \Gamma \vdash \mathtt{unfold}\ M : A[\mu X. A / X]$}
  \end{bprooftree}
\]
\caption{Formation rules for terms. All types are assumed to be closed and well-formed.}
\label{fig:term-syntax}}
\end{figure*}


\begin{figure*}
{
\begin{minipage}{0.5\textwidth}
\begin{align*}
\pi_1 (V,W) \probto 1 V \qquad \qquad \pi_2 (V,W) &\probto 1 W \\
(\mathtt{case}\ \mathtt{in}_1 V\ \mathtt{of}\ \mathtt{in}_1 x \Rightarrow N_1\ |\ \mathtt{in}_2 y \Rightarrow N_2) &\probto 1 N_1[V/x] \\
(\mathtt{case}\ \mathtt{in}_2 V\ \mathtt{of}\ \mathtt{in}_1 x \Rightarrow N_1\ |\ \mathtt{in}_2 y \Rightarrow N_2) &\probto 1 N_2[V/y] \\
\mathtt{unfold}\ \mathtt{fold}\ V &\probto 1 V \\
(\lambda x. M)V &\probto 1 M[V/x]  \\
M\ \mathtt{or}_p\ N \probto p M \qquad \qquad M\ \mathtt{or}_p\ N &\probto{1-p} N 
\end{align*}
\end{minipage}
\begin{minipage}{0.5\textwidth}
\begin{align*}
E & ::= [\cdot]\ |\ (E, M)\ |\ (V, E)\ |\ \pi_i E\ |\  EM\ |\ VE\ | \\
  & \qquad \mathtt{in}_i E\ |\ (\mathtt{case}\ E\ \mathtt{of}\ \mathtt{in}_1 x \Rightarrow N_1\ |\ \mathtt{in}_2 y \Rightarrow N_2)\ | \\  
  & \qquad \mathtt{fold}\ E\  |\ \mathtt{unfold}\ E \\
\end{align*}
\[
  \begin{bprooftree}
  \def\ScoreOverhang{0.5pt}
  \AxiomC{$M \probto{p} M' $}
  \UnaryInfC{$E[M] \probto{p} E[M'] $}
  \end{bprooftree}
\]
\end{minipage}
\caption{Reduction rules for $\mathtt{PFPC}$. The grammar for $E$ defines our call-by-value evaluation contexts.}
\label{fig:operational}}
\end{figure*}

\subsection{The Types of $\mathtt{PFPC}$}

Recursive types in $\mathtt{PFPC}$ are formed in the same way as in $\mathtt{FPC}$. We use
$X,Y$ to range over \emph{type variables} and we use $\Theta$ to range over
\emph{type contexts}. A type context $\Theta = X_1, \ldots, X_n$ is
\emph{well-formed}, written $\Theta \vdash$, if all type variables within it
are distinct. We use $A,B$ to range over the \emph{types} of our language
which are defined in Figure \ref{fig:syntax-grammars}. We write
$\Theta \vdash A$ to indicate that type $A$ is \emph{well-formed} in type
context $\Theta$ whenever the judgement is derivable via the rules in Figure
\ref{fig:type-syntax}. A type $A$ is \emph{closed} when $\cdot \vdash A$.
We remark that there are no restrictions on the admissible logical polarities
of our type expressions, even when forming recursive types.

\begin{example}
\label{ex:basic-types}
Some important (closed) types may be defined in the following way. The \emph{empty type}
is defined as $0 \eqdef \mu X.X$ and the \emph{unit type} as $1 \eqdef 0 \to 0.$
We may also define:
\begin{itemize}
\item \emph{Booleans} as $\mathtt{Bool} \eqdef 1 + 1$;
\item \emph{Natural numbers} as $\mathtt{Nat} \eqdef \mu X. 1 + X$;
\item \emph{Lists of type} $A$ as $\mathtt{List}(A) \eqdef \mu X. 1 + A \times X$;
\item \emph{Streams of type} $A$ as $\mathtt{Stream}(A) \eqdef \mu X. 1 \to A \times X$;
\end{itemize}
and many others.
\end{example}

\subsection{The Terms of $\mathtt{PFPC}$}

We now explain the syntax we use for terms. When forming terms and term contexts, we implicitly assume that all types within are closed and well-formed.
We use $x,y$ to range over \emph{term variables} and we use $\Gamma$ to range over \emph{term contexts}. A (well-formed) term context $\Gamma = x_1 : A_1, \ldots, x_n : A_n$ is a list of (distinct) variables with their types.
The terms (ranged over by $M, N$) and the values (ranged over by $V,W$) of $\mathtt{PFPC}$ are specified in Figure \ref{fig:syntax-grammars} and their formation rules in Figure \ref{fig:term-syntax}. They are completely standard.
In Figure \ref{fig:term-syntax}, the notation $A[\mu X. A / X]$ indicates type substitution which is defined in the standard way. The term $\orp M p N$ represents probabilistic choice.
A term $M$ of type $A$ is \emph{closed} when $\cdot \vdash M : A$ and in this case we also simply write $M : A.$

\begin{example}
Important closed values in $\mathtt{PFPC}$ include:
the \emph{unit value} $() \eqdef \lambda x^0. x : 1;$ 
the \emph{false} and \emph{true} values given by $\mathtt{ff} \eqdef \mathtt{in}_1 () : \mathtt{Bool}$ and $\mathtt{tt} \eqdef \mathtt{in}_2 () : \mathtt{Bool};$ the \emph{zero natural number} $\mathtt{zero} \eqdef \mathtt{fold}\ \mathtt{in}_1 () : \mathtt{Nat}$ and the \emph{successor function}
$\mathtt{succ}\ \eqdef \lambda n^{\mathtt{Nat}}. \mathtt{fold}\ \mathtt{in}_2 n : \mathtt{Nat} \to \mathtt{Nat};$ among many others.
\end{example}

\subsection{The Reduction Rules of $\mathtt{PFPC}$}

To describe execution of programs in $\mathtt{PFPC}$, we use a small-step call-by-value operational semantics which is described in Figure \ref{fig:operational}.
The reduction relation $M \probto{p} N$ should be understood as specifying that term $M$ reduces to term $N$ with probability $p \in [0,1]$ in exactly one step. Our reduction rules are simply the standard rules for small-step reduction in
$\mathtt{FPC}$ \cite[{\S 20}]{harper-book}
and small-step reduction for probabilistic choice \cite{lambda-calculus-probabilistic-computation}.
Of course, it is well-known this system is type-safe.

\begin{theorem}
If $\Gamma \vdash M \colon A$ and $M \probto{p} N,$ then $\Gamma \vdash N \colon A.$
In this situation, if $p < 1$, then there exists a term $N'$, such that $M \probto{1-p} N'.$
Furthermore, if $\cdot \vdash M \colon A$, then either $M$ is a value or there exists $N,$ such that $M \probto{p} N$ for some $p \in [0,1]$. 
\end{theorem}

\begin{assumption}
Throughout the rest of the paper, we implicitly assume that all types, terms and contexts are well-formed.
\end{assumption}

\subsection{Recursion and Asymptotic Behaviour of Reduction}
\label{sub:probability-reduction}

It is well-known that type recursion in $\mathtt{FPC}$ induces term recursion \cite{fiore-thesis,fpc-syntax,harper-book} and the same is true for $\mathtt{PFPC}$.
This allows us to \emph{derive} the call-by-value fixpoint operator
\[\cdot \vdash \mathtt{fix}_{A \to B} \colon \left( \left( A \to B \right) \to A \to B \right) \to A \to B \]
at any function type $A \to B$ (see \cite{fpc-syntax} and \cite[\S 8]{fiore-thesis} for more details).
Using $\mathtt{fix}_{A \to B}$, we may write recursive functions.

\begin{example}
\label{ex:infinite-coin-toss}
Consider the following program:
\begin{align*}
\mathtt{coins} &\eqdef \mathtt{fix}_{1 \to 1} \lambda f^{1 \to 1}. \lambda x^1.\ \mathtt{case} ( \mathtt{ff}\ \mathtt{or}_{0.5}\ \mathtt{tt}  )\ \mathtt{of} \\
& \mathtt{in_1}z \Rightarrow ()\ |\ \mathtt{in_2}z \Rightarrow fx .
\end{align*}
It follows $\cdot \vdash \mathtt{coins} : 1 \to 1$. Evaluating at $()$ shows that $\mathtt{coins}()$ performs a fair coin toss and depending on the outcome, either terminates to $()$ or repeats the process again.
We see that there is no upper bound on the number of coin tosses this program would perform. On the other hand, it is easy to see that the probability $\mathtt{coins}()$ terminates to $()$ is precisely
$\sum_{i=1}^\infty 2^{-i} = 1.$
\end{example}

The above simple example shows that a rigorous operational analysis of $\mathtt{PFPC}$ has to consider the \emph{asymptotic behaviour} of terms under reduction. 
We do this by showing how to determine the probability that a term reduces to a value in any number of steps. We will later see that this is crucial for proving our adequacy result (Theorem \ref{thm:strong-adequacy}).

We may determine the overall probability that a term $M$ reduces to a value $V$ in the same way as in \cite{quant-semantics}.
The \emph{probability weight} of a reduction path $\pi = \left( M_1  \probto{p_1} \cdots \probto{p_n} M_n \right)$ is $P(\pi) \eqdef \prod_{i=1}^n p_i.$
The probability that term $M$ reduces to the value $V$ in \emph{at most} $n$ steps is
\[ P(M \probto{}_{\leq n} V) \defeq \sum_{\pi \in \Paths_{\leq n}(M,V) } P(\pi) , \]
where $\Paths_{\leq n}(M,V)$ is the set of all reduction paths from $M$ to $V$ of length at most $n$.
The probability that term $M$ reduces to value $V$ (in \emph{any finite number} of steps) is $P(M \probto{}_{*} V) \defeq \sup_i P(M \probto{}_{\leq i} V) . $

Finally, the probability that term $M$ \emph{terminates} is denoted $\Halt(M)$ and it is determined in the following way:
\begin{align}
\Val(M) &\defeq \{ V \ |\ V \text{ is a value and } P(M \probto{}_* V) > 0 \} \label{eq:val(M)} \\
\Halt(M) &\eqdef \sum_{V \in \Val(M)} P(M \probto{}_* V). \label{eq:halt} 
\end{align}
Note that the sum in \eqref{eq:halt} is countably infinite, in general.

\section{Commutative Monads for Probability}
\label{sec:probabilistic}

In this section we present a novel and general construction for probabilistic
commutative monads on $\dcpo$ and we use it to identify three such monads.

\subsection{Domain-theoretic and Topological Preliminaries}
\label{sub:domain-preliminaries}

A nonempty subset $A$ of a partially ordered set (\emph{poset}) $D$ is 
\emph{directed} if each pair of elements in $A$ has an upper bound in $A$. 
A \emph{directed-complete partial order,} (\emph{dcpo}, for short) is a poset 
in which every directed subset $A$ has a supremum $\sup A$. For example, 
the unit interval $[0, 1]$ is a dcpo in the usual ordering. A function $f \colon D \to E$ between two (posets) dcpo's
is \emph{Scott-continuous} if it is monotone and preserves (existing) suprema
of directed subsets. 

The category $\dcpo$  of dcpo's and Scott-continuous functions is complete, cocomplete and 
cartesian closed~\cite{abramskyjung:domaintheory}. We denote with $A_1 \times A_2$ ($A_1$ +
$A_2$) the categorical (co)product of the dcpo's $A_1$ and $A_2$ and with $\pi_1,
\pi_2$ ($\emph{in}_1, \emph{in}_2$) the associated (co)projections. We denote with
$\varnothing$ and $1$ the initial and terminal objects of $\dcpo$; these are the empty dcpo and the singleton dcpo, respectively. $\dcpo$ is Cartesian closed, where the internal hom of $A$ and $B$ is $[A \to B]$, 
the Scott-continuous functions $f: A \to B$ ordered pointwise.

The category $\dcpobs$ 
 of \emph{pointed} dcpo's and \emph{strict} Scott-continuous functions also is important. $\dcpobs$ is symmetric monoidal closed when equipped with the smash product and strict Scott-continuous function space, and it is also complete and cocomplete \cite{abramskyjung:domaintheory}.

The \emph{Scott topology}  $\sigma D$ on a dcpo $D$ consists of 
the upper subsets $U = \up U = \{x\in D\mid (\exists u\in U)\, u\leq x\}$ 
that are \emph{inaccessible by directed suprema:} i.e., if  $A\subseteq D$ is directed and 
$\sup A\in U$, then $A\cap U\not=\emptyset$. The space $(D, \sigma D)$ is also written as $\Sigma D$. Scott-continuous functions between dcpo's $D$ and $E$ are exactly the continuous 
functions between $\Sigma D$ and $\Sigma E$~\cite[Proposition II-2.1]{gierzetal:domains}. We always equip $[0,1]$ with the Scott topology unless stated otherwise.

A subset $B$ of a dcpo $D$ is  a \emph{sub-dcpo}
if  every directed subset $A\subseteq B$ satisfies $\sup_D A\in B$. 
In this case, $B$  is a dcpo in the induced order from $D$. 
The \emph{d-topology} on $D$ is the  topology whose closed subsets consist of sub-dcpo's of~$D$.
Open (closed) sets in the d-topology will be called \emph{d-open (d-closed)}. The \emph{d-closure} of $C\subseteq D$ 
is the topological closure of $C$ with respect to the d-topology on $D$, which is the 
intersection of all sub-dcpo's of $D$ containing $C$. 

The family of open sets of a topological space $X$, denoted $\mathcal O X$, is a complete lattice in the inclusion order. 
The \emph{specialization order}~$\leq_{X}$ on $X$ is defined as  $x\leq_{X} y$ if and only if $x$ is in the closure of $\{y\}$, for $x, y\in X$. We write $\Omega X$ to denote $X$ equipped with the specialization order. It is well-known that $X$ is $T_{0}$ if and only if $\Omega X$ is a poset. A subset of $X$ is called \emph{saturated} if it is an upper set in $\Omega X$. 
A space $X$ is called a \emph{d-space} or a \emph{monotone-convergence space} if $\Omega X$ is a dcpo and each open set of $X$ is Scott open in $\Omega X$. As an example, $\Sigma D$ is always a d-space for each dcpo $D$. The full subcategory of $\mathbf{T_{0}}$ consisting of d-spaces is denoted by $\mathbf D$.
There is a functor $\Sigma\colon \DCPO\to \mathbf D$ that assigns the space $\Sigma D$ to each dcpo $D$, and the map $f\colon \Sigma D \to \Sigma E$ to the Scott-continuous map $f\colon D\to E$. Dually, the functor $\Omega\colon {\mathbf D}\to \DCPO$ assigns $\Omega X$ to each d-space $X$ and the map $f\colon \Omega X\to \Omega Y$ to each continuous map $f\colon X\to Y$. In fact, $\Sigma \dashv \Omega$, i.e., $\Sigma$ is left adjoint to $\Omega$~\cite{HoGJX18}.

A $T_{0}$ space $X$ is called \emph{sober} if every nonempty closed irreducible subset of $X$ is the closure of some (unique) singleton set, where  $A\subseteq X$ is  \emph{irreducible} if $A\subseteq B\cup C$ with $B$ and $C$ nonempty closed subsets implies $A\subseteq B$ or $A\subseteq C$. The category of sober spaces and continuous functions is denoted by $\mathbf{SOB}$. Sober spaces are d-spaces, hence $\mathbf{SOB} \subseteq \mathbf D$ \cite{KeimelL09}.

\subsection{A Commutative Monad for Probability}
\label{sub:monad}

To begin, a \emph{subprobability valuation} on a topological space $X$ is a Scott-continuous function $\nu \colon \mathcal O X \to [0,1]$ that is 
strict ($\nu(\emptyset) = 0$),  and modular $(\nu(U) + \nu(V) = \nu(U\cup V) + \nu(U\cap V) )$.  
The set of  subprobability valuations on $X$ is denoted by $\VV X$.
The \emph{stochastic order} on $\VV X$ is defined pointwise: $\nu_{1}\leq \nu_{2}$ if and only if $\nu_{1}(U)\leq \nu_{2}(U)$ for all $U\in \mathcal OX$.
$\VV X$ is a pointed dcpo in the stochastic order, with least element given by the constantly zero valuation~${\mathbf 0}_{X}$
and where the supremum of a directed family $\{\nu_{i}\}_{i\in I}$ is $\sup_{i\in I}\nu_i \defeq\lambda U. \sup_{i\in I} \nu_{i}(U)$.

The canonical examples of subprobability valuations are the \emph{Dirac valuations} $\delta_{x}$
for $x\in X$, defined by $\delta_{x}(U) = 1$ if $x\in U$ and $\delta_{x} (U) =0$ otherwise. 
$\VV X$ enjoys a convex structure: if $\nu_{i}\in \VV X$ and $r_{i}\geq 0, $ with $\sum_{i=1}^{n}r_{i}\leq 1$, then the convex sum $\sum_{i=1}^{n}r_{i}\nu_{i} \defeq \lambda U. \sum_{i=1}^{n}r_{i}\nu_{i}(U)$ also is in $\VV X$. 
The \emph{simple valuations} on $D$ are those of the form $\sum_{i=1}^{n}r_{i}\delta_{x_{i}}$, where $x_{i}\in X$, $r_{i}>0, i= 1, \ldots, n$ and $\sum_{i=1}^{n}r_{i}\leq 1$. 
The set of simple valuations on $X$ is denoted by $\SSS X$. Clearly, $\mathcal SX\subseteq \VV X$. Unlike $\VV X$, $\SSS X$ is not directed-complete in the stochastic order in general.

Given $\nu\in \VV X$ and $f\colon X\to [0,1]$ continuous,  we can define the 
integral of $f$ against $\nu$ by the Choquet formula
\[
\int_{x\in X} f(x) d \nu \defeq \int_{0}^{1} \nu(f^{-1}((t, 1]))dt,
\]
where the right side  is a Riemann integral of the bounded antitone function $\lambda t. \nu(f^{-1}((t, 1]))$. If no confusion occurs, we simply write $\int_{x\in X} f(x) d\nu$ as $\int fd\nu$. Basic properties of this integral can be found in~\cite{jones90}. Here we note that the map $\nu\mapsto \int fd\nu \colon \VV X\to [0, 1]$, for a fixed $f$, is Scott-continuous, and 
\begin{equation}
\label{eq:nestedintwithsim}
\int fd \sum_{i}^{n}r_{i}\delta_{x_{i}} = \sum_{i=1}^{n}r_{i}f(x_{i}) 
\end{equation}
for $\sum_{i=1}^{n}r_{i}\delta_{x_{i}}\in \VV X$.

For a dcpo~$D$, $\VV D$ is defined as $\VV(D, \sigma D)$.
Using Manes' description of monads (Kleisli triples)~\cite{manes76}, Jones proved in her PhD thesis~\cite{jones90} that $\VV$ is a monad on $\DCPO$:
\begin{itemize}
\item The \emph{unit} of $\VV$ at $D$ is $\eta^{\VV}_{D}\colon D\to \VV D\colon x\mapsto \delta_{x}$. 
\item The \emph{Kleisli extension} $f^{\dagger}$ of a Scott-continuous map $f\colon D\to \VV E$ maps $\nu\in \VV D$ to 
$f^{\dagger}(\nu)\in \VV E$   by\\[1ex]
 \centerline{$f^{\dagger}(\nu) \defeq \lambda U \in \sigma E. \int_{x\in D}f(x)(U)d\nu.$}
\end{itemize}
Then the \emph{multiplication} $\mu^{\VV}_{D}\colon \VV\VV D\to \VV D$ is given by $\id_{\VV D}^{\dagger}$; it maps $\varpi \in \VV\VV D$ to  $\lambda U \in \sigma D. \int_{\nu\in \VV D}\nu(U)d\varpi\in \VV D$. 
Thus, $\VV$ defines an endofunctor on $\DCPO$ that sends a dcpo $D$ to $\VV D$, and a Scott-continuous map $h\colon D\to E$ to $\VV(h) \defeq(\eta_{E}\circ h)^{\dagger}$; concretely, $\VV(h)$ maps $\nu\in \VV D$ to $\lambda U \in\sigma E. \nu(h^{-1}(U))$. 

Jones~\cite{jones90} also showed that $\VV$ is a strong monad over $\DCPO$: its strength at $(D, E)$ is given by 
\[\tau^{\VV}_{DE}\colon D \times \VV E \to \VV(D\times E) \colon (x, \nu) \mapsto \lambda U. \int_{y\in E} \chi_{U}(x, y)d\nu,\]
where $\chi_{U}$ is the characteristic function of $U\in \sigma(D\times E)$. 
Whether $\VV$ is a commutative monad on $\DCPO$ has remained an open problem for decades. Proving this to be true requires showing the following Fubini-type equation holds: 
\begin{equation}\label{eqn:Fub}
\int_{x\in D}\int_{y\in E}\chi_{U}(x, y)d\xi d\nu = \int_{y\in E}\int_{x\in D}\chi_{U}(x, y)d\nu d\xi, 
\end{equation}
for dcpo's $D$ and $E$, for $U\in \sigma(D\times E)$ and for $\nu\in \VV D, \xi\in \VV E$~\cite[Section~6]{JonesP89}.
The difficulty lies in the well-known fact that a Scott open set $U\in \sigma(D\times E)$ might not be open in the product topology $\sigma D\times \sigma E$ in general~\cite[Exercise II-4.26]{gierzetal:domains}.  

However, if either $\nu$ or $\xi$ is a simple valuation,   then Equation~\eqref{eqn:Fub} 
holds. For example, if  $\nu=\sum_{i=1}^{n}r_{i}\delta_{x_{i}}\in \mathcal SD$, then by \eqref{eq:nestedintwithsim} both sides of~\eqref{eqn:Fub} are equal to $\sum_{i=1}^{n}r_{i}\int_{y\in E}\chi_{U}(x_{i}, y)d\xi$. The Scott continuity of the integral in~$\nu$ then implies Equation~\eqref{eqn:Fub} holds for valuations that are directed suprema of simple valuations. 
This is why, for example, $\VV$ is a commutative monad on the category of domains and Scott-continuous maps
, as we now explain. 

If $D$ is a dcpo and $x, y\in D$, we say $x$ is \emph{way-below} $y$ (in symbols, $x\ll y$)  if and only 
if for every directed set $A$ with $y\leq \sup A$, there is some $a\in A$ such
that $x\leq a$. We write $\Da y = \{x\in D\mid x\ll y\}$. 
A \emph{basis} for a dcpo $D$ is subset $B$ satisfying $\Da x\cap B$ is directed and $x = \sup \Da x\cap B$, for each $x\in D$. $D$ is  \emph{continuous} if it has a basis.
 Continuous dcpo's are also called \emph{domains}, and the category of domains
and Scott-continuous maps is denoted by $\DOM$.

Applying the reasoning  above about simple valuations, we obtain a commutative monad of valuations on $\DCPO$ by restricting
to a suitable completion of $\SSS D$ inside $\VV D$.  There are several possibilities (cf.~\cite{jiamis}), and we choose the smallest and simplest  -- the d-closure of $\SSS D$ in $\VV D$. 

\begin{definition}
For each dcpo $D$, we define $\MM D$ to be the intersection of all sub-dcpo's of $\VV D$ that contain $\SSS D$. \footnote{The same definition applies in the case of topological spaces.}
\end{definition}

Since $\VV D$ itself is a  dcpo containing $\SSS D$, it is immediate from the definition of sub-dcpo's that $\MM D$ is a well-defined dcpo in the stochastic order with  $\SSS D\subseteq \MM D\subseteq \VV D$. Analogous to $\VV D$, $\MM D$ also enjoys a convex structure.

\begin{lemma}
\label{lemma:MDisconvexclosed}
For $\nu_{i}\in \MM D$ and $r_{i}\geq 0, i=1,\ldots, n$ with $\sum_{i=1}^{n}r_{i}\leq 1$, the convex sum $\sum_{i=1}^{n}r_{i}\nu_{i}$ is still in $\MM D$. 
\end{lemma}
\begin{proof}
In Appendix~\ref{section:commu}.
\end{proof}

For the proofs of the following results, we repeatedly use the fact that Scott-continuous maps between dcpo's $D$ and $E$ are \emph{d-continuous}, i.e.,  
continuous when $D$ and $E$ are equipped with the d-topology~\cite[Lemma 5]{zhao10}.

\begin{theorem}
\label{theorem:M is commutative}
$\MM$ is a commutative monad on $\DCPO$. 
\end{theorem}
\begin{proof}
We sketch the key steps in showing $\MM$ is commutative:
\paragraph*{\textbf{Unit}}
The unit of $\MM$ at $D$ is $\eta^{\MM}_{D}\colon D\to \MM D\colon x\mapsto \delta_{x}$, the co-restriction of $\eta^{\VV}_{D}$ to $\MM D$. Obviously, it is  a well-defined Scott-continuous map. 

\paragraph*{\textbf{Extension}}
Since a Scott-continuous map $f\colon D\to \MM E$ is also Scott-continuous from $D$ to $\VV E$, the Kleisli extension $f^{\ddagger}\colon \MM D\to \MM E$ of $f$ can be defined as the restriction and co-restriction of $f^{\dagger}\colon \VV D\to \VV E$ to $\MM D$ and $\MM E$, respectively. The validity of this definition requires $f^{\dagger}(\MM D)\subseteq \MM E$, which boils down to $f^{\dagger}(\SSS D)\subseteq \MM E$ by d-continuity of $f^{\dagger}$, since $f^{\dagger}$ is Scott-continuous. Hence we only need to check that $f^{\dagger}( \sum_{i=1}^{n}r_{i}\delta_{x_{i}}) \in \MM E$ for each $ \sum_{i=1}^{n}r_{i}\delta_{x_{i}} \in \SSS D$. However,  $f^{\dagger}( \sum_{i=1}^{n}r_{i}\delta_{x_{i}}) = \sum_{i=1}^{n} r_{i}f(x_{i})$, which is indeed in $\MM E$ by Lemma~\ref{lemma:MDisconvexclosed}.  

\paragraph*{\textbf{Strength}} The strength $\tau^{\MM}_{DE}$ of $\MM$ at $(D, E)$ is given by $\tau^{\VV}_{DE}$ restricted to $D\times \MM E$ and co-restricted to $\MM(D\times E)$. This is well-defined provided that $\tau^{\VV}_{DE}$ maps $D\times \MM E$ into $\MM(D\times E)$. Again, we only need to prove that  $\tau^{\VV}_{DE}$ maps $D\times \SSS E$ into $\MM(D\times E)$ and conclude the proof with the d-continuity of $\tau^{\VV}_{DE}$ in its second component. Towards this end, we pick $(a , \sum_{i=1}^{n}r_{i}\delta_{y_{i}}) \in  D\times \SSS E$, and see 
\begin{align*}
\tau^{\VV}_{DE}(a, \sum_{i=1}^{n}r_{i}\delta_{y_{i}})  &= \lambda U. \int \chi_{U}(a, y) d \sum_{i=1}^{n}r_{i}\delta_{y_{i}} \\
											    &\stackrel{\eqref{eq:nestedintwithsim}}{=} \lambda U.  \sum_{i=1}^{n}r_{i}\chi_{U}(a, y_{i}) &\\
										&	   = 
											    \lambda U. \sum_{i=1}^{n} r_{i} \delta_{(a, y_{i})}(U) 
											    \defeq \sum_{i=1}^{n} r_{i} \delta_{(a, y_{i})}
\end{align*}
is indeed in $\MM(D\times E)$. 

With $f^{\ddagger}$ and $\tau^{\MM}$ well-defined, the same arguments used to prove $(\VV, \eta^{\VV}, \_^{\dagger}, \tau^{\VV})$ is a strong monad in~\cite{jones90} prove $(\MM, \eta^{\MM}, \_^{\ddagger}, \tau^{\MM})$ is a strong monad on $\DCPO$.

\paragraph*{\textbf{Commutativity}} Finally, we show $\MM$ is commutative by proving the Equation~\eqref{eqn:Fub} holds for any dcpo's $D$ and $E$ and $\nu\in \MM D, \xi\in \MM E$. As commented above, this holds if $\nu$ is simple, and then the Scott-continuity of the integral in the $\nu$-component implies Equation~\eqref{eqn:Fub} also holds for directed suprema of simple valuations, directed suprema of directed suprema of simple valuations and so forth, transfinitely. But these are exactly the valuations $\MM D$. 

Formally, we consider for each fixed $\xi \in \MM E$ (even for $\xi \in \VV E$) the functions
\[
F\colon \nu\mapsto \int_{x\in D}\int_{y\in E}\chi_{U}(x, y)d\xi d\nu \colon \MM D\to [0, 1]
\]
and 
\[
G\colon \nu\mapsto  \int_{y\in E}\int_{x\in D}\chi_{U}(x, y)d\nu d\xi \colon \MM D\to [0, 1]. 
\]
Note that both $F$ and $G$ are Scott-continuous functions hence d-continuous, and they are equal on $\SSS D$ by Equation~\eqref{eq:nestedintwithsim}. Since $[0, 1]$ is Hausdorff in the d-topology, $F$ and $G$ are then equal on the d-closure of $\SSS D$ which is, by construction, $\MM D$. 
\end{proof}

\begin{remark}
\label{remark:multi}
The \emph{multiplication} $\mu^{\MM}_{D}$ of $\MM$ at $D$ is given by $(\id_{\MM D})^{\ddagger}$. Concretely, $\mu^{\MM}_{D}$ maps each valuation $\varpi \in \MM(\MM D)$ to $\lambda U\in \sigma D. \int_{\nu\in \MM D} \nu(U) d \varpi$. In particular,  $\mu^{\MM}_{D}$ maps each simple valuation $\sum_{i=1}^{n}r_{i}\delta_{\nu_{i}}\in \MM(\MM D)$ to $\sum_{i=1}^{n}r_{i}{\nu_{i}}$, where $\nu_{i}\in \MM D, i=1,\ldots, n$, and $\sum_{i=1}^{n}r_{i}\leq 1$. 
\end{remark}

\begin{remark}
\label{remark:tensorofvaluations}
The \emph{double strength} of $\MM$ at $(D, E)$ is given by the Scott-continuous map  $(\nu, \xi) \mapsto \nu\otimes \xi \colon \MM(D) \times \MM(E) \to \MM(D\times E)$, where $\nu \otimes \xi$ is defined as  $\lambda U\in \sigma(D\times E). \int_{y\in E}\int_{x\in D} \chi_{U}(x, y) d\nu d\xi$.
\end{remark}

\begin{remark}
\label{remark:sameval}
We note that  $\MM D$ is the first example of a commutative valuations monad on $\dcpo$ that contains the simple valuations. And, since every valuation on a domain $D$ is a directed supremum of simple valuations~\cite[Theorem 5.2]{jones90},  it follows that $\MM = \VV$ on the category $\DOM$.
\end{remark}

\subsection{Dcpo-completion versus $\mathbf D$-completion}
Recall that a \emph{dcpo-completion} of a poset $P$ is a pair $(D, e)$, where $D$ is a dcpo and $e\colon P\to D$ is an injective Scott-continuous map, such that for any dcpo $E$ and Scott-continuous map $f\colon P\to E$, there exists a unique Scott-continuous map $f'\colon D\to E$ satisfying $f = f'\circ e$. The dcpo-completion of posets always exists~\cite[Theorem 1]{zhao10}. 

As we have seen, for each dcpo~$D$, $\MM D$ is  the smallest sub-dcpo in $\VV D$ containing $\SSS D$, one may wonder whether $\MM D$, together with the inclusion map from $\SSS D$ into $\MM D$,  is a dcpo-completion of $\SSS D$. The answer is ``no''  in general. The reason is that the inclusion of $\SSS D$ into $\MM D$ may not be Scott-continuous, even when $D$ is a domain (see \cite[Section 6]{jiamis}). The construction $\MM D$ is actually more in a topological flavour, as we now explain. For simplicity, we assume all spaces considered in the sequel are in $\mathbf {T_{0}}$, the category of $T_{0}$ spaces and continuous maps. 

\begin{definition}
Let $X$ be a topological space. The \emph{weak topology} on $\VV X$ is generated by the sets  $$[U>r] \defeq \{\nu \in \VV X \ | \ \nu(U)>r \},$$ which form a subbasis, where $U$ is  open in $X$ and $r\in [0, 1]$. 
\end{definition}
\begin{remark}
\label{remark:fgeqr}
For each continuous map $f\colon X\to [0, 1]$ and $r\in [0, 1]$, the set $[f>r] \defeq \{\nu \in \VV X \ | \ \int fd\nu >r \}$ is open in the weak topology. 
\end{remark}

We use $\VV_{w}X$ to denote the space $\VV X$ equipped with the weak topology.  We will use the fact that $\VV_{w} X$ is a sober space, which follows from~\cite[Proposition 5.1]{heckmann96}. It is easy to see that the specialization order on $\VV_{w} X$ is just the stochastic order. Hence $\VV X = \Omega(\VV_{w}X)$. 

We also use $\SSS_{w} X$ ($\MM_{w} X$) to denote the space $\SSS X$ ($\MM X$) endowed with the relative topology from $\VV_{w} X$. Accordingly, $\MM X = \Omega(\MM_{w} X)$, and $\SSS X = \Omega(\SSS_{w} X)$.  Although $\MM X$ is not the dcpo-completion of $\SSS X$ in general, we do have the following:

\begin{proposition}
\label{pro:Mwisd-compofSw}
For each space $X$, $\MM_{w} X$ is a \emph{$\mathbf D$-completion} of $\SSS_{w} X$. That is, $\MM_{w} X$ itself is a d-space, an object in $\mathbf D$; the inclusion map $i \colon \SSS_{w} X\to \MM_{w} X$ is continuous; and for any d-space $Y$ and continuous map $f\colon \SSS_{w} X\to Y$, there exists a unique continuous map $f'\colon \MM_{w} X \to Y$ such that $f= f'\circ i$. 
\end{proposition}

The above proposition is a straightforward  application of Keimel and Lawson's \texttt{K}-category theory~\cite{KeimelL09} to the category~$\mathbf D$. 

\begin{definition}\label{def:Kcat}
A \texttt{K}-category $\mathbf K$ is a full subcategory of $\mathbf{T_{0}}$, whose objects will be called k-spaces,  satisfying:
\begin{enumerate}
\item  Homeomorphic copies of k-spaces are k-spaces;
\item All sober spaces are k-spaces, i.e., $\mathbf{SOB} \subseteq \mathbf K$;
\item \label{item:insectionofkspaces} In a sober space $S$, the intersection of any family of k-subspaces, equipped with the relative topology from~$S$,  is a k-space;
\item \label{def:Kcat4} For any continuous map $f\colon S\to T$ between sober spaces $S$ and $T$, and any k-subspace $K$ of $T$, $f^{-1}(K)$ is k-subspace of $S$. 
\end{enumerate}
\end{definition}

If $\mathbf K$ is a \texttt{K}-category, then the $\mathbf K$-completion\footnote{The definition of $\mathbf K$-completion is similar to that of $\mathbf D$-completion and can be found in~\cite{KeimelL09}.} of any $T_{0}$-space $X$ always exists, and one possible completion process goes as follows~\cite[Theorem 4.4]{KeimelL09}: 
First, pick any $j\colon X\to Y$ such that $Y$ is sober and $j$ is a topological embedding. For example, one can take $j$ as the embedding of $X$ into its standard sobrification. Second, let $\tilde{X}$ be the intersection of all k-subspaces of $Y$ containing $j(X)$ and equip it with the relative topology from $Y$. Then $\tilde{X}$, together with the  co-restriction $i\colon X\to \tilde{X}$ of $j$, is a $\mathbf K$-completion of $X$.

Now we apply this procedure to prove Proposition~\ref{pro:Mwisd-compofSw}. First, note that $\mathbf D$ is indeed a \texttt{K}-category as proved in~\cite[Lemma 6.4]{KeimelL09}.  We embed $\SSS_{w}X$ into the sober space $\VV_{w}X$, and notice that all d-subspaces of $\VV_{w}X$ are precisely sub-dcpo's of $\VV X$. Hence $\MM_{w}X$, which is the intersection of sub-dcpo's $\VV X$ containing $\SSS X$ equipped with the relative topology from $\VV_{w}X$, is a $\mathbf D$-completion of $\SSS_{w}X$.

\subsection{A uniform construction}
Proposition~\ref{pro:Mwisd-compofSw} motivates the next definition.

\begin{definition}
Let $\mathbf K$ be a \texttt{K}-category. For each space $X$, we define $\VV_{\mathbf K} X$ to be the intersection of all k-subspaces of $\VV_{w} X$ containing $\SSS_{w} X$, equipped with the relative topology from~$\VV_{w} X$. 
\end{definition}

As discussed above, $\VV_{\mathbf K} X$ is a  \texttt{K}-completion of $\SSS_{w} X$. 
It was proved in~\cite[Theorem 3.5]{jiamis}\footnote{The authors allow valuations to take values in $[0, \infty]$. However, the theorem is also true for valuations with values in~$[0, 1]$ }
that $\VV_{\mathbf K} \colon \mathbf{T_{0}} \to \mathbf{T_{0}} $ is a monad for each \texttt{K}-category  
$\mathbf K$: The unit of $\VV_{\mathbf K}$ at $X$ maps $x\in X$ to $\delta_{x}$, and for any continuous map $f\colon X\to \VV_{\mathbf K} Y$, the Kleisli extension $f^{\dagger} \colon \VV_{\mathbf K} X \to \VV_{\mathbf K} Y$ maps $\nu$ to $\lambda U\in \mathcal OY. \int_{x\in X}f(x)(U)d\nu$. 
Therefore, if $\mathbf K$ is a full subcategory of $\mathbf D$, then according to the construction $\VV_{\mathbf K}X$ is always a d-space for each $X$, hence the monad $\VV_{\mathbf K} \colon \mathbf{T_{0}} \to \mathbf{T_{0}} $
can be restricted to a monad on $\mathbf D$. 

\begin{theorem}
Let $\mathbf K$ be a \emph{\texttt{K}}-category with $\mathbf K\subseteq \mathbf D$. Then $\VV_{\mathbf K, \leq} \defeq \Omega\circ \VV_{\mathbf K}\circ \Sigma$ is a monad on $\DCPO$.
\end{theorem}
\vspace{-.2in}\[
\begin{tikzcd}
\mathbf{DCPO}
\arrow[rr, "\Sigma"{name=F}, bend left=10] 
&&
\mathbf D \ar[rr, "\mathcal F" {name=H}, bend left = 10] 
\arrow[loop, "\mathcal V_\mathbf K"', distance=2.5em, start anchor={[xshift=1ex]north}, end anchor={[xshift=-1ex]north}]{}{}
\arrow[ll, "\Omega"{name=G}, bend left=10]
\arrow[phantom, from=F, to=G, "\dashv" rotate=-90]
&& \mathbf D^{\mathcal V_{\mathbf K}} \ar[ll, "\mathcal U" {name=K}, bend left = 10] 
\arrow[phantom, from=H, to=K, "\dashv" rotate=-90]
\end{tikzcd}
\]
\begin{proof}
Let $ \mathbf D^{\mathcal V_{\mathbf K}} $ be the Eilenberg-Moore category of $\VV_{\mathbf K}$ over $\mathbf D$ and $\mathcal F\dashv \mathcal U$ be the adjunction that recovers $\VV_{\mathbf K}$, then $\VV_{\mathbf K, \leq} = \Omega\circ \mathcal U\circ \mathcal F\circ \Sigma$. The statements follow from the standard categorical fact that adjoints compose: $ \mathcal F\circ \Sigma \dashv \Omega\circ \mathcal U$.  
\end{proof}

\begin{remark}\label{remark:thesameunitandextension}
The unit of $\VV_{\mathbf K, \leq}$ at dcpo $D$ sends $x\in D$ to $\delta_{x}$, and for dcpo's $D$ and $E$, the Kleisli extension $f^{\dagger} \colon \VV_{\mathbf K, \leq} D\to \VV_{\mathbf K, \leq} E$ of $f\colon D\to  \VV_{\mathbf K, \leq} E$ maps $\nu$ to $\lambda U\in \mathcal \sigma E. \int_{x\in D}f(x)(U)d\nu$.  
\end{remark}

\begin{remark}
$\MM_{w} = \VV_{\mathbf D}$ and $\MM = \VV_{\mathbf D, \leq}$.
\end{remark}

Note that the category $\mathbf{SOB}$ of sober spaces is the smallest {\texttt{K}-category~\cite[Remark 4.1]{KeimelL09}. We denote 
$\VV_{\mathbf{SOB}}$ by $\mathcal P_{w}$  and $\VV_{\mathbf{SOB}, \leq}$ by $\mathcal P$. The following statement is then obvious.

\begin{proposition}
\label{pro:submonad}
Let $\mathbf K$ be a \emph{\texttt{K}}-category with $\mathbf K\subseteq \mathbf D$. Then for each dcpo $D$, we have 
$\SSS D\subseteq \MM D \subseteq \VV_{\mathbf K, \leq} D\subseteq \mathcal PD\subseteq \VV D$. 
\end{proposition}

Heckmann~\cite[Theorem~5.5]{heckmann96} proved that $\mathcal PD$ consists of the so-called \emph{point-continuous} valuations on $D$. We claim that the Equation~\ref{eqn:Fub} holds when either $\nu$ or $\xi$ is point-continuous:
\begin{theorem}
\label{theorem:VPFubini}
Let $D$ and $E$ be dcpo's, and $U\in \sigma(D\times E)$. Then the equation 
\begin{equation*}\label{eqn:Fubpv}
\int_{x\in D}\int_{y\in E}\chi_{U}(x, y)d\xi d\nu = \int_{y\in E}\int_{x\in D}\chi_{U}(x, y)d\nu d\xi, 
\end{equation*}
holds for $(\nu, \xi) \in \mathcal P D\times \VV E$ (equivalently, $(\nu, \xi) \in \VV D\times \mathcal P E$). 
\end{theorem}
As far as we know, this is the most general Fubini theorem on dcpo's. The proof, which relies on the Schr\"oder-Simpson Theorem~\cite{goubaultsstheorem}, is included in Appendix~\ref{section:commu}.  Hence by combining Remark~\ref{remark:thesameunitandextension}, Proposition~\ref{pro:submonad} and Theorem \ref{theorem:VPFubini} we get our next theorem.

\begin{theorem}
\label{theorem:kcofsiscommu}
For any \emph{\texttt{K}}-category $\mathbf K$ with  $\mathbf K\subseteq \mathbf D$, $\VV_{\mathbf K, \leq}$ is a commutative monad on $\DCPO$.
\end{theorem}
\begin{proof}
In Appendix~\ref{section:commu}.
\end{proof}

As promised, we conclude this subsection with a third commutative monad $\mathcal W$ on $\DCPO$ by describing a \texttt{K}-category lying between $\mathbf{SOB}$ and $\mathbf D$, the category $\mathbf{WF}$ consisting of well-filtered spaces and continuous maps.  A $T_{0}$ space $X$ is \emph{well-filtered}  if,  given any filtered family $\{K_{a}\}_{a\in A}$ of compact saturated subsets of $X$ with $\bigcap_{a\in A}K_{a}\subseteq U$, with $U$ open, there is some $a\in A$ with $K_{a}\subseteq U$. A proof that  $\mathbf{WF}$ is a \texttt{K}-category between $\mathbf{SOB}$ and $\mathbf{D}$ can be found in~\cite{wxxx20}. Hence $\mathcal W \defeq \VV_{\mathbf{WF}, \leq}$ is a commutative monad on $\DCPO$ and $\MM D\subseteq \mathcal WD\subseteq \mathcal PD$ for every dcpo~$D$.

\begin{remark}
All subsequent results we present in this paper hold for the three monads $\MM, \mathcal W$ and $\mathcal P$. To avoid cumbersome repetition, we explicitly state them for $\MM.$
\end{remark}

\subsection{Continuous Kegelspitzen and $\MM$-algebras}
\label{sub:M-V-relationship}
Kegelspitzen~\cite{keimelplotkin17} are dcpo's that enjoy a convex structure. In this section, we show  every \emph{continuous} Kegelspitze $K$ has a \emph{linear barycenter map} $\beta\colon \MM K \to K$ making $(K,\beta)$ an $\MM$-algebra and conversely, every  $\MM$-algebra $(K, \beta)$ on $\DCPO$ admits a Kegelspitze structure on $K$ making $\beta\colon \MM K \to K$ a linear map.   We begin with the notion of a barycentric algebra.
\begin{definition}
A \emph{barycentric algebra} is a set $A$ endowed with a binary operation $a+_{r}b$ for every real number $r\in [0, 1]$ such that for all $a, b, c \in A$ and $r, p\in [0, 1]$, the following equations hold:
\begin{align*}
&a+_{1} b = a; ~~~~a+_{r}b = b+_{1-r}a; ~~~~ a+_{r} a = a;\\
&(a+_{p}b)+_{r}c = a+_{pr}(b+_{\frac{r-pr}{1-pr}} c)~~\text{provided}~r, p<1.
\end{align*}
\end{definition}

\begin{definition}
A \emph{pointed barycentric algebra} is a barycentric algebra~$A$ with a distinguished element~$\bot$. 
For $a\in A$ and $r\in [0, 1]$, we define $r\cdot a \defeq a+_{r} \bot$. 
A map $f\colon A\to B$ between pointed barycentric algebras is called \emph{linear} if $f(\bot_{A}) = \bot_{B}$ and $f(a+_{r}b) = f(a)+_{r}f(b)$ for all $a, b\in A, r\in [0, 1]$. 
\end{definition}

\begin{definition}
A \emph{Kegelspitze} is a pointed barycentric algebra $K$ equipped with a directed-complete partial order such that, for every $r$ in the unit interval, the functions determined by
convex combination $(a, b)\mapsto a+_{r}b\colon K\times K\to K$ and scalar multiplication $(r, a)\mapsto r\cdot a \colon [0,1]\times K\to K$ are Scott-continuous in both arguments. A \emph{continuous Kegelspitze} is a Kegelspitze that  is a domain in the equipped order. 
\end{definition}

\begin{remark}
In a Kegelspitze $K$, the map $(r, a)\mapsto r\cdot a = a +_r \bot$ is Scott-continuous, hence monotone, in the $r$-component, which implies $\bot = \bot +_{1} a = a+_{0} \bot = 0\cdot a\leq 1\cdot a = a$ for each $a\in K$, i.e., $\bot$ is the least element of $K$.
\end{remark}

\begin{example}
\label{exa:MDiskegel}
For each dcpo $D$, $\MM D$ is a Kegelspitze: for $\nu_{1}, \nu_{2} \in \MM D$ and $r\in [0, 1]$, $\nu_{1} +_{r} \nu_{2}$ is defined as 
$r \nu_{1} + (1-r ) \nu_{2}$. Lemma~\ref{lemma:MDisconvexclosed} implies this is well-defined.\footnote{Note that  Lemma~\ref{lemma:MDisconvexclosed} is stated only for $\MM$, but it also holds for $\mathcal W$ and $\mathcal P$: one notes that $\nu_{1} \mapsto r\nu_{1}+(1-r)\nu_{2}\colon \VV_{w} D\to \VV_{w} D$ is a continuous map between sober spaces and then uses Definition~\ref{def:Kcat} Item~\eqref{def:Kcat4} to replace ``d-continuity'' in the proof.}
The constantly zero valuation $\mathbf 0_{D}$ is the distinguished element. 
Verifying that $\MM D$ is a Kegelspitze is then straightforward.\\
As a consequence, for each Scott-continuous map $f\colon D\to E$, the map $\MM(f) \colon \MM D\to \MM E\colon \nu \mapsto \lambda U\in \sigma E. \nu(f^{-1}(U))$ is obviously  linear. 
\end{example}

\begin{definition}
\label{def:convex-sums}
In each pointed barycentric algebra~$K$, for $a_{i}\in K, r_{i}\in [0, 1], i=1,\ldots, n$ with $\sum_{i=1}^{n}r_{i}\leq 1$, we define the convex sum inductively
\[\sum_{i=1}^{n}r_{i}a_{i} \defeq \begin{cases}
a_{1} & \text{, if } r_{1}=1,\\
a_{1}+_{r_{1}}(\sum_{i=2}^{n} \frac{r_{i}}{1-r_{1}}a_{i}) &  \text{, if }r_{1}< 1.
						\end{cases}
\]
This  is invariant under index-permutation: for $\pi$ a permutation of $\{1, \ldots, n\}$, $\sum_{i=1}^{n}r_{i}a_{i} = \sum_{i=1}^{n}r_{\pi(i)}a_{\pi(i)}$~\cite[Lemma 5.6]{jones90}. If $K$ is a Kegelspitze, then the expression $\sum_{i=1}^{n}r_{i}a_{i}$ is Scott-continuous in each $r_{i}$ and $a_{i}$. 
A countable convex sum may also be defined: given $a_i \in K$ and $r_i \in [0,1]$, for $i \in I$, with $\sum_{i \in I} r_i \leq 1$, let
$ \sum_{i \in I} r_i a_i \defeq \sup \{ \sum_{j \in J} r_j a_j \ |\ J \subseteq I \text{ and } J \text{ is finite} \} . $
\end{definition}

\begin{lemma}
A function $f\colon K_{1} \to K_{2}$ between pointed barycentric algebras $K_{1}$ and $K_{2}$ is linear if and only if $f( \sum_{i=1}^n r_{i}a_{i})  = \sum_{i=1}^n r_{i}f(a_{i})$ for  $a_{i}\in K_{1}, i=1, \ldots, n$ and $\sum_{i=1}^{n}r_{i}\leq 1$.
\end{lemma}

\begin{definition}
Let $K$ be a Kegelspitze and $s=\sum_{i=1}^{n}r_{i}\delta_{x_{i}}$ be a simple valuation on~$K$. The \emph{barycenter} of $s$ is defined as $\beta_*(s) \defeq \sum_{i=1}^{n}r_{i}x_{i}$. 
\end{definition}

As a straightforward consequence of Jones' Splitting Lemma (\cite[Proposition IV-9.18]{gierzetal:domains}), the map $\beta_{*}(s)$ is monotone from $\SSS K$ to $K$.
If $K$ is 
continuous, 
then $\MM K = \VV K$ and $\SSS K$ is a  basis for $\MM K$  (see Remark~\ref{remark:sameval}). 
We extend $\beta_*$ to the \emph{barycenter map}
\[\beta\colon \MM K\to  K\ \text{by}\  \beta(\nu) \defeq \sup\{\beta_{*}(s)\mid s \in \SSS K \text{ and }s\ll \nu\}.\] 
Note that for each simple valuation $s = \sum_{i=1}^{n}r_{i}\delta_{x_{i}} \in \SSS K$, there exists a directed set $A$ of $\SSS K$ with supremum $s$ consisting of simple valuations way-below~$s$. For example, one can choose $A= \{ \sum_{i=1}^{n}\frac{mr_{i}}{m+1}\delta_{y_{i}}  \ | \ m \in \mathbb N ~\text{and}~ y_{i} \ll x_{i}\}$. 
By~\cite[Lemma IV-9.23.]{gierzetal:domains}, the map $\beta$, as defined above, is a Scott-continuous map extending~$\beta_{*}$, i.e.,  $\beta(\nu) = \beta_{*}(\nu) $ for $\nu\in \SSS K$. Moreover, $\beta$ is a linear map since $\beta_{*}$ is.

\begin{proposition}
\label{prop:EM-algebra}
Each \emph{continuous} Kegelspitze $K$ admits a linear barycenter map $\beta\colon \MM K\to K$ (as above) for which the pair $(K, \beta)$ is an Eilenberg-Moore algebra of $\MM$.
\end{proposition}
\begin{proof}
Clearly, $\beta\circ \eta^{\MM}_{K} = \id_{K}$. To prove that $\beta\circ \mu^{\MM}_{K} =\beta\circ \MM( \beta)$, we only need to prove both sides are equal on simple valuations in $\MM(\MM K)$, since $\SSS(\MM K)$ is dense in  $\MM(\MM K)$ in the d-topology, and both sides of the equation are d-continuous functions. However, when applied to the simple valuation $\sum_{i=1}^{n}r_{i}\delta_{\nu_{i}} \in \SSS(\MM K)$, both sides equal $\sum_{i=1}^{n}r_{i}\beta(\nu_{i})$. This follows from direct computation by employing Remark~\ref{remark:multi} and linearity of~$\beta$. 
\end{proof}

We next show that every Eilenberg-Moore algebra $(K, \beta)$ of $\MM$ on $\dcpo$ admits a Kegelspitze structure on $K$ making $\beta\colon \MM K\to K$ a linear map. 

\begin{proposition}
\label{prop:Keglstructure}
Let  $(K, \beta)$ be an $\MM$-algebra on $\DCPO$. For $a, b\in K$ and $r\in[0,1]$, define $a+_{r}b \defeq \beta (\delta_{a} +_{r} \delta_{b})$.  Then with the operation~$+_{r}$,  $K$ is a Kegelspitze and $\beta\colon \MM K\to K$ is linear. 
\end{proposition}
\begin{proof}
See Appendix~\ref{section:commu}.
\end{proof}

\begin{proposition}
\label{prop:EMmaps}
Let  $(K_{1}, \beta_{1})$ and $(K_{2}, \beta_{2})$ be $\MM$-algebras on $\DOM$. A Scott-continuous function $f\colon K_{1}\to K_{2}$ is   an algebra morphism from $(K_{1}, \beta_{1})$ to $(K_{2}, \beta_{2})$ if and only if $f$ is linear with respect to the Kegelspitze structure on $K_{1}$ and $K_{2}$ introduced by $\beta_{1}$ and $\beta_{2}$, respectively, as in Proposition~\ref{prop:Keglstructure}.
\end{proposition}
\begin{proof}
See Appendix~\ref{section:commu}.
\end{proof}

\begin{theorem}
\label{thm:emcategory}
The Eilenberg-Moore category $\DOM^{\MM}$ of $\MM$ over $\DOM$ is isomorphic to the category of continuous Kegelspitzen and Scott-continuous linear maps. 
\end{theorem}
\begin{proof}
Combine Propositions~\ref{prop:EM-algebra},~\ref{prop:Keglstructure} and~\ref{prop:EMmaps}.
\end{proof}

\begin{remark}\label{rem:kegelspitz}
Theorem~\ref{thm:emcategory} characterises  $\DOM^{\MM}$, which equals $\DOM^{\VV_{\mathbf K, \leq}}$ for any \texttt{K}-category $\mathbf K$ with  $\mathbf K\subseteq \mathbf D$ since $\VV = \MM$ on domains (see Remark~\ref{remark:sameval} and Proposition~\ref{pro:submonad}). This corrects an error in~\cite{jones90}: there it is proved that \emph{continuous abstract probabilistic domains} and linear maps form a full subcategory of $\DOM^{\VV}$.  But there is a claim that all objects in $\DOM^{\VV}$ are abstract probabilistic domains. A separating example is the extended non-negative reals $[0, \infty]$, which is a continuous Kegelspitze but not an abstract probabilistic domain. 
\end{remark}

\section{Categorical Model}
\label{sec:model}

In this section we describe the categorical properties of the Kleisli category of our monad $\MM$.
Everything we say in this section is also true for our other two monads as well.

We write $\KL$ for the Kleisli category of our monad $\MM : \dcpo \to \dcpo$.
In order
to distinguish between the categorical primitives of $\DCPO$ and $\KL$,
we indicate with $f: A \kto B$ the morphisms of $\KL$ and
we write $f \kcirc g \eqdef \mu \circ \MM(f) \circ g$ for the Kleisli composition of morphisms in $\KL$.
We write $\kid_A : A \kto A$ with $\kid_A = \eta_A  : A \to \MM A$ for the identity morphisms in $\KL.$
The monad $\MM$ induces an adjunction $\JJ \dashv \mathcal U : \KL \to \dcpo,$ where:
\begin{align*}
\JJ A \defeq A , \quad  \JJ f \defeq \eta \circ f  , \quad
\UU A \defeq \MM A , \quad \UU f \defeq \mu \circ \MM f .
\end{align*}

\subsubsection{Coproducts}
The category $\KL$ inherits (small) coproducts from $\dcpo$ in the standard way
\cite[pp. 264]{jacobs-coalgebra} and we write $A_1 \kplus A_2 \eqdef A_1 + A_2$ for the induced
(binary) coproduct. The induced coprojections are given by $\JJ(\emph{in}_1) \colon A_1 \kto A_1 \kplus A_2$
and $\JJ(\emph{in}_2) \colon A_2 \kto A_1 \kplus A_2.$
Then for $f\colon A\kto C$ and $g\colon B\kto D$, $f\kplus g = [\MM(\emph{in}_{C}) \circ f, \MM(\emph{in}_{D}) \circ g]$.

\subsubsection{Symmetric monoidal structure}
Because our monad $\MM$ is \emph{commutative}, it induces a
symmetric monoidal structure on $\KL$ in a canonical way \cite[pp. 462]{premonoidal}.
The induced tensor product is $A \ktimes B \eqdef A \times B$ and the
Kleisli projections are
$\JJ(\pi_A) : A \ktimes B \kto A$ and $\JJ(\pi_B) : A \ktimes B \kto B$ \footnote{These projections do \emph{not}
satisfy the universal property of a product.}.
For $f\colon A\kto C$ and $g\colon B\kto D$, their tensor product is given by
$f\ktimes g = \lambda (a, b). f(a)\otimes g(b)$. Note that the last expression uses the double strength of $\MM$, see Remark~\ref{remark:tensorofvaluations}. 

Standard categorical arguments now show that the Kleisli products distribute over the Kleisli coproducts. We write $d_{A,B,C} : A \ktimes (B \kplus C) \cong (A \ktimes B) \kplus (A \ktimes C)$ for this natural isomorphism.

\subsubsection{The left adjoint $\JJ$}
The functor $\JJ$, whose action is the identity on objects, preserves
the monoidal structure and the coproduct structure up to equality (and not merely up to isomorphism). That is, $\JJ(A \star B) = JA \kstar JB$
and $J(f \star g) = Jf \kstar Jg,$ where $\star \in \{\times, +\}.$

\subsubsection{Kleisli Exponential}
Our Kleisli adjunction also contains the structure of a \emph{Kleisli-exponential}
(which is also known as a \emph{$\MM$-exponential}). Following Moggi \cite{moggi-monads}, we will use
this to interpret higher-order function types. Next, we describe this structure in greater
detail.

The functor $J(-) \ktimes B : \DCPO \to \KL$ has a right adjoint, which
we write as $[B \kto -] : \KL \to \DCPO,$ for each dcpo $B$. In
particular $[B \kto - ] \defeq [B \to \UU(-)],$ which means that,
on objects, $[B \kto C] = [B \to \MM C].$ This data provides us with a
family of Scott-continuous bijections
\begin{equation}
\label{eq:currying-simple}
\lambda : \KL(\JJ A \ktimes B, C) \cong \DCPO(A, [B \kto C])
\end{equation}
natural in $A$ and $C$, called \emph{currying}. We also denote with $\epsilon : \JJ [B \kto -] \ktimes B \naturalto \Id,$ the counit of the adjunctions \eqref{eq:currying-simple}, often called \emph{evaluation}.
Because this family of adjunctions is parameterised by objects $B$ of $\KL$, it
follows using standard categorical results \cite[\S IV.7]{maclane} that the
assignment $[B \kto -] : \KL \to \DCPO$ may be  extended uniquely to a bifunctor
$ [- \kto -] : \KL^\op \times \KL \to \DCPO , $
such that the bijections $\lambda$ in \eqref{eq:currying-simple} are natural
in all components\footnote{This extension is canonically given by $[f \kto g] \defeq \lambda(g \kcirc \epsilon_{C_1} \kcirc (\kid \ktimes f) ).$}.

\begin{remark}
Some authors describe currying and evaluation for Kleisli exponentials without
referring to the functor $\JJ$. This cannot lead to confusion on the object level, 
but to be fully precise, one has to specify that the naturality
properties on the $A$-component hold only for \emph{total} maps. We make this
explicit by including $\JJ$ in our presentation.
\end{remark}

\subsubsection{Enrichment Structure}
The category $\DCPO_{\MM}$ is enriched over $\dcpobs$: for all dcpo's~$A, B$ and $C$, the Kleisli exponential $[A\kto B] = [A\to \MM B] = \KL(A,B)$ is a pointed dcpo in the pointwise order, and the Kleisli composition 
\[ 
\kcirc \colon [A\kto B]\times [B\kto C] \to [A\kto C] \colon (f, g) \mapsto g\kcirc f = g^{\ddagger} \circ f
\]
is obviously a strict Scott-continuous map. 
Moreover, the adjunction $\JJ \dashv \UU \colon \KL \to \dcpo$ is also $\DCPO$-enriched (see {\rm \cite[Definition~6.7.1]{borceux:handbook2}} for definition) and so are the bifunctors
$(- \ktimes -),  (- \kplus -) $ and $[- \kto - ] .$

We interpret probabilistic effects using the convex structure of our model which we now describe.
For each dcpo~$B$, $\MM B$ is a Kegelspitze in the stochastic order (Example~\ref{exa:MDiskegel}) : for $r\in[0, 1]$ and $\nu_{1}, \nu_{2}\in \MM B$, 
$\nu_{1}+_{r}\nu_{2}$ is defined as $r\nu_{1}+(1-r)\nu_{2}$;
the zero-valuation $\mathbf 0_{B}$ is the distinguished element (which is also least).
It follows that $[A\kto B] = \KL(A,B)$ is a Kegelspitze in the pointwise order: for $f, g\in [A\kto B]$,
$f+_{r} g$ is defined as $\lambda x. f(x)+_{r}g(x)$.
Next, we note that this convex structure is preserved by Kleisli composition~$\kcirc$, Kleisli coproduct $\kplus$ and Kleisli product $\ktimes$.

\begin{lemma}
\label{lemma:respectkegel}
Let $A, B, C, D$ be dcpo's,  $f, f_{1}, f_{2} \in [A\kto B]$, $g, g_{1}, g_{2} \in [B\kto C], h, h_{1}, h_{2}\in [C\kto D]$ and $r\in [0, 1]$. Then we have:
\begin{itemize}
\item $(g_{1}+_{r}g_{2}) \kcirc f = g_{1} \kcirc f +_{r} g_{2}\kcirc f;$
\item $g\kcirc (f_{1}+_{r}f_{2})  = g \kcirc f_{1} +_{r} g\kcirc f_{2};$ 
\item $(f_{1}+_{r}f_{2})\kstar h = f_{1}\kstar h +_{r} f_{2} \kstar h;$
\item $f\kstar (h_{1}+_{r}h_{2})= f\kstar h_{1} +_{r} f\kstar h_{2},$
\end{itemize}
where $\star\in \{\times, +\}$ in the last two cases. 
\end{lemma}
\begin{proof}
See Appendix~\ref{app:timesplusconvex}. 
\end{proof}

\subsubsection{Important Subcategories}
\label{subsub:subcategories}
In order to describe our denotational semantics, we have to identify two important subcategories of $\KL$.

\begin{definition}
\label{def:T}
The subcategory of \emph{deterministic total maps}, denoted $\TD$, is the full-on-objects subcategory of $\KL$ each of whose morphisms $f \colon X \kto Y$  admits a factorisation $f = \JJ(f') = \left( X \xrightarrow{f'} Y \xrightarrow{\eta_Y} \MM Y \right).$
\end{definition}
Therefore, by definition, each map $f: X \kto Y$ in $\TD$ satisfies
$f(x) = \delta_y$ for some $y \in Y$. These maps are \emph{deterministic} in
the sense that they carry no interesting convex structure and they are
\emph{total} in the sense that they map all inputs $x \in X$ to non-zero
valuations. The importance of this subcategory is that all values of
our language admit an interpretation within $\TD$. Moreover, the categorical
structure of $\TD$ is very easy to describe, as our next proposition shows.

\begin{proposition}
There exists a $\DCPO$-enriched isomorphism of categories $\DCPO \cong \TD$.
\end{proposition}
\begin{proof}
Each map $\eta_X \colon X \to \MM X$ is injective, because $\Sigma X$ is a $T_0$ space and so $\JJ \colon \DCPO \to \KL$ is faithful. Its corestriction to $\TD$ is the required isomorphism.
\end{proof}

In our model, the canonical copy map at an object $A$ is given by the map $\JJ
\langle \id_A, \id_A \rangle \colon A  \kto A \ktimes A$ and the canonical
discarding map at $A$ is the map $\JJ(1_A) \colon A \kto 1,$ where $1_A \colon
A \to 1$ is the terminal map of $\dcpo$.  Because maps in $\TD$ are in the
image of $\JJ$, it follows that they are compatible with the copy and
discard maps and thus also with weakening and contraction \cite{benton-small,benton-wadler}.

The next subcategory we introduce is important, because we will use it for the interpretation of open types.
It has sufficient structure to solve recursive domain equations.

\begin{definition}
\label{def:D}
The subcategory of \emph{deterministic partial maps}, denoted $\PD$, is the full-on-objects subcategory of $\KL$ each of whose morphisms $f \colon X \kto Y$ admits a factorisation $f = \left( X \xrightarrow{f'} Y_\perp \xrightarrow{\phi_Y} \MM Y \right),$
where $Y_\perp$ is the dcpo obtained from $Y$ by freely adding a least element $\perp$, and $\phi_Y$ is the map:
\[
\phi_Y \colon Y_\perp \to \MM Y :: y \mapsto
  \begin{cases}
    {\mathbf 0}_{Y}        & \text{, if } y = \perp \\
    \delta_y & \text{, if } y \neq \perp .
  \end{cases}
\]
\end{definition}
These maps are  \emph{partial} because some inputs are mapped to $\mathbf 0$, but also deterministic, because the convex structure is trivial in both cases.
This is further justified by the next proposition.

\begin{proposition}
\label{prop:d-iso}
There exists a $\dcpobs$-enriched isomorphism of categories $\dcpo_{\LL} \cong \PD$,
where $\dcpo_{\LL}$ is the Kleisli category of the lift monad $\LL : \DCPO \to \DCPO$.
\end{proposition}
\begin{proof}
The assignment $\phi$ from Definition \ref{def:D} is a \emph{strong map of monads}
$\phi : \LL \naturalto \MM$ which then
induces a functor $\FF \colon \dcpo_\LL \to \KL$ (Appendix \ref{app:domain-equations}).
Each $\phi_Y$ is injective, so the corestriction of $\FF$ to $\PD$ is the required isomorphism.
\end{proof}

\subsubsection{Solving Recursive Domain Equations}
\label{subsub:domain-equations}
In order to interpret recursive types, we solve the required recursive domain equations by constructing \emph{parameterised initial algebras} \cite{fiore-thesis,fiore-plotkin} within (the subcategory of embeddings of) $\PD$ using the limit-colimit coincidence theorem \cite{smyth-plotkin:domain-equations}.

\begin{definition}[see {\cite[\S 6.1]{fiore-thesis}}]
\label{def:initial-algebra}
  Given a category $\CC$ and a functor $\TTT \colon \CC^{n+1} \to \CC,$ a \emph{parameterised initial algebra}
  for $\TTT$ is a pair $(\TTT^\sharp, \iota^\TTT),$ such that:
  \begin{itemize}
    \item $\TTT^\sharp \colon \CC^n \to \CC$ is a functor;
    \item $\iota^\TTT \colon \TTT \circ \langle \Id, \TTT^\sharp \rangle \naturalto \TTT^\sharp : \CC^n \to \CC$ is a natural transformation;
    \item For every $\vec C \in \Ob(\CC^n)$, the pair $(\TTT^\sharp \vec C, \iota^\TTT_{\vec C})$ is an initial $\TTT(\vec C, -)$-algebra.
  \end{itemize}
\end{definition}

In the special case when $n=1$, we recover the usual notion of initial algebra. We consider \emph{parameterised} initial algebras because we need to interpret mutual type recursion. Similarly, one can also define the dual notion of \emph{parameterised final coalgebra}.

\begin{proposition}[see {\cite[\S 4.3]{lnl-fpc-lmcs}}]
\label{prop:par-initial-algebra}
Let $\CC$ be a category with an initial object and all $\omega$-colimits and let $\TTT \colon \CC^{n+1} \to \CC$ be an $\omega$-cocontinuous functor. Then $\TTT$ has a  parameterised initial algebra $(\TTT^\sharp, \iota^\TTT)$ and the functor $\TTT^{\sharp} \colon \CC^n \to \CC $ is also $\omega$-cocontinuous.
\end{proposition}

The next proposition shows that the subcategory $\PD$ has sufficient structure to solve recursive domain equations.

\begin{proposition}
\label{prop:compactness}
The subcategory $\PD$ is \emph{(parameterised) $\DCPO$-algebraically compact}.
More specifically, every $\dcpo$-enriched functor $\TTT: \PD^{n+1} \to \PD$ has a parameterised compact algebra, i.e., a parameterised initial algebra whose inverse is a
parameterised final coalgebra for $\TTT$.
\end{proposition}
\begin{proof}
By Proposition \ref{prop:d-iso}, we have $\PD \cong \dcpo_{\LL} \cong \dcpobs$ and the latter two categories are well-known to be $\dcpo$-algebraically compact (which may be easily established using \cite[Corollary 7.2.4]{fiore-thesis}).
\end{proof}

Therefore, every $\dcpo$-enriched \emph{covariant} functor on $\KL$ which
restricts to $\PD$ can be equipped with a parameterised compact algebra. In
order to solve equations involving \emph{mixed-variance} functors (induced by function types),
we use the limit-colimit coincidence theorem
\cite{smyth-plotkin:domain-equations}. In particular, an important observation
made by Smyth and Plotkin in \cite{smyth-plotkin:domain-equations} allows us to
interpret all type expressions (including function spaces) as covariant
functors on \emph{subcategories of embeddings}.
These ideas are developed in detail in \cite{icfp19,lnl-fpc-lmcs} and here we
also follow this approach.

\begin{definition}
Given a $\DCPO$-enriched category $\CC$, an \emph{embedding} of $\CC$ is a morphism $e \colon X \to Y$, such that there exists (a necessarily unique) morphism $e^p \colon Y \to X$, called a \emph{projection}, with the properties: $e^p \circ e = \id_X$ and $e \circ e^p \leq \id_Y.$
We denote with $\CC_e$ the full-on-objects subcategory of $\CC$ whose morphisms are the embeddings of $\CC$.
\end{definition}

\begin{proposition}
\label{prop:omega-functors}
The category $\PD_e$ has an initial object and all $\omega$-colimits, and the following assignments:
\begin{itemize}
\item $\ktimes_e \colon \PD_e \times \PD_e \to \PD_e$ by\\
  $X \ktimes_e Y \defeq X \ktimes Y\ \text{and}\ e_1 \ktimes_e e_2 \defeq e_1 \ktimes e_2.$\\
  \item $ \kplus_e \colon \PD_e \times \PD_e \to \PD_e$ by\\
$X \kplus_e Y \defeq X \kplus Y\ \text{and}\ e_1 \kplus_e e_2 \defeq e_1 \kplus e_2$\\
\item
$[\kto]_e^\JJ \colon \PD_e \times \PD_e \to \PD_e$ \\
$[X \kto Y]_e^\JJ \defeq \JJ [ X \kto Y]\ \text{and}\
[e_1 \kto e_2]_e^\JJ \defeq \JJ [e_1^p \kto e_2]$
\end{itemize}
define \emph{covariant} $\omega$-cocontinuous bifunctors on $\PD_e$.
\end{proposition}
\begin{proof}
This follows using results from \cite{smyth-plotkin:domain-equations} together with some restriction arguments which we present in Appendix \ref{app:domain-equations}.
\end{proof}

Therefore, by Proposition
\ref{prop:par-initial-algebra} and Proposition \ref{prop:omega-functors} we can
solve recursive domain equations induced by all well-formed type expressions
(with no restrictions on the admissible logical polarities of the types) within $\PD_e$.
However, since our judgements support weakening and contraction, we have
an extra proof obligation: showing each isomorphism that is a solution
to a recursive domain equation can be copied and discarded. This is indeed true
(for any isomorphism in $\PD$) because of the next proposition.

\begin{proposition}
\label{prop:iso-reflect}
Every isomorphism of $\PD$ (and $\PD_e$) is also an isomorphism of $\TD$.
\end{proposition}
\begin{proof}
In Appendix \ref{app:domain-equations}.
\end{proof}

We have already explained that morphisms of $\TD$ are compatible with weakening
and contraction, so the above proposition suffices for our purposes.

\section{Denotational Semantics}
\label{sec:semantics}

We now give the denotational semantics of our language by using ideas from \cite{icfp19,lnl-fpc-lmcs}.

\subsection{Interpretation of Types}
\label{sub:type-interpretation}

We begin with the interpretation of (open) types.
Every type $\Theta \vdash A$ is interpreted as a functor $\lrb{\Theta \vdash A} \colon\ \PDe^{|\Theta|} \to \PDe $ and its interpretation is defined by
induction on the derivation of $\Theta \vdash A$ in Figure \ref{fig:type-interpretation}.
The validity of this definition is justified by the next proposition.

\begin{figure}
{\centering
\begin{align*}
\lrb{\Theta \vdash A} &\colon\ \PDe^{|\Theta|} \to \PDe \\
\lrb{\Theta \vdash \Theta_i} &\defeq \Pi_i \\
\lrb{\Theta \vdash A + B } &\defeq\ \kplus_e \circ\ \langle \lrb{\Theta \vdash A}, \lrb{\Theta \vdash B} \rangle \\
\lrb{\Theta \vdash A \times B } &\defeq\ \ktimes_e \circ\ \langle \lrb{\Theta \vdash A}, \lrb{\Theta \vdash B} \rangle \\
\lrb{\Theta \vdash A \to B } &\defeq [\kto]_e^\JJ \circ \langle \lrb{\Theta \vdash A}, \lrb{\Theta \vdash B} \rangle \\
\lrb{\Theta \vdash \mu X.A} &\defeq \lrb{\Theta, X \vdash A}^\sharp 
\end{align*}
\caption{Interpretation of types.}
\label{fig:type-interpretation}
\begin{align*}
   \lrb{A \times B} &= \lrb A \times \lrb B   \\
   \lrb{A+B} &= \lrb A + \lrb B \\
   \lrb{A \to B} &= \left[ \lrb A \to \MM \lrb B \right] \\
   \lrb{\mu X. A} &\cong  \lrb{A[\mu X. A/ X]}
\end{align*}
\caption{Derived equations for closed types.}
\label{fig:derived-type-semantics}}
\end{figure}
\begin{figure}
{\[
\begin{array}{l}
\lrb{\Gamma \vdash M : A} : \lrb \Gamma \kto \lrb A \text{ in $\KL$} \\
\lrb{\Gamma, x:A \vdash x: A} \eqdef \JJ \pi_2 \\
\lrb{\Gamma \vdash (M, N) : A \times B} \eqdef (\lrb M \ktimes \lrb N) \kcirc \JJ \langle \id, \id \rangle  \\
\lrb{\Gamma \vdash \pi_i M : A_i} \eqdef \JJ \pi_i \kcirc \lrb M   , \text{ for } i \in \{ 1, 2\} \\
\lrb{\Gamma \vdash \mathtt{in}_{i} M : A_1+A_2} \eqdef \JJ \emph{in}_i \kcirc \lrb{M}, \text{ for } i \in \{ 1, 2\} \\
\lrb{\Gamma \vdash ( \mathtt{case}\ M\ \mathtt{of}\ \mathtt{in}_1 x \Rightarrow N_1\ |\ \mathtt{in}_2 y \Rightarrow N_2 ) : B} \eqdef \\
\qquad \qquad  [\lrb{N_1}, \lrb{N_2}] \kcirc d \kcirc (\id \ktimes \lrb M) \kcirc \JJ \langle \id, \id \rangle \\
\lrb{\Gamma \vdash \lambda x^A . M : A \to B} \eqdef \JJ \lambda(\lrb M) \\
\lrb{\Gamma \vdash MN : B} \eqdef \epsilon \kcirc (\lrb M \ktimes \lrb N) \kcirc \JJ \langle \id, \id \rangle  \\
\lrb{\Gamma \vdash \mathtt{fold}\ M: \mu X. A} \eqdef \sfold \kcirc \lrb M  \\
\lrb{\Gamma \vdash \mathtt{unfold}\ M : A[\mu X. A / X]} \eqdef \sunfold \kcirc \lrb M \\
\lrb{\Gamma \vdash M\ \mathtt{or}_p\ N : A} \eqdef \lrb M +_{p} \lrb N
\end{array}
\]
\caption{Interpretation of term judgements.}
\label{fig:term-semantics}}
\end{figure}

\begin{proposition}
The assignments $\lrb{\Theta \vdash A} \colon\ \PDe^{|\Theta|} \to \PDe$ are $\omega$-cocontinuous functors.
\end{proposition}
\begin{proof}
By induction using Propositions \ref{prop:par-initial-algebra} and \ref{prop:omega-functors}.
\end{proof}

We are primarily interested in closed types and for them we simply write $\lrb A \defeq \lrb{\cdot \vdash A}(*)$, where $*$ is the unique object of the terminal category $\mathbf 1 = \PDe^{0}.$
For closed types, it follows that $\lrb{A} \in \Ob(\PDe) = \Ob(\DCPO)$. 

We proceed by defining the folding/unfolding isomorphisms for recursive types and proving a necessary lemma.

\begin{lemma}[Substitution]
\label{lem:type-substitution}
If $\Theta, X \vdash A$ and $\Theta \vdash B$, then:
\begin{align*}
\lrb{\Theta \vdash A[B/X]} = \lrb{\Theta, X \vdash A} \circ \langle \Id, \lrb{\Theta \vdash B} \rangle.
\end{align*}
\end{lemma}

\begin{definition}
\label{def:fold/unfold}
For closed types $\mu X.A$, we define:
\begin{align*}
\sfold_{\mu X. A} &: \lrb{A[\mu X. A/ X]}  = \lrb{X \vdash A} \lrb{\mu X. A}   \cong \lrb{\mu X. A},
\end{align*}
where the equality is Lemma~\ref{lem:type-substitution} and the isomorphism is the initial algebra.
We write $\sunfold_{\mu X. A}$ for the inverse isomorphism.
Note that both of them are isomorphisms in $\TD$.
\end{definition}

Now the equations for closed types in Figure \ref{fig:derived-type-semantics} follow immediately.

\subsection{Interpretation of Terms}
\label{sub:term-interpretation}
A context $\Gamma = x_1 \colon A_1, \ldots, x_n \colon A_n$ is interpreted as the dcpo $\lrb \Gamma \eqdef \lrb{A_1} \times \cdots \times \lrb{A_n}.$
A term $\Gamma \vdash M : A$ is, as usual, interpreted as a morphism $\lrb{\Gamma \vdash M : A} \colon \lrb \Gamma \kto \lrb A$ in $\KL$ and we will
abbreviate this by writing $\lrb M$ when its type and context are clear.
The interpretation of term judgements are defined by induction in Figure \ref{fig:term-semantics}. This interpretation is defined in the standard categorical way using the structure of $\KL$
and using the structure of the Kleisli exponential following Moggi \cite{moggi-monads}.
To interpret probabilistic choice, we use the convex structure of $\KL.$
All the notation used in Figure \ref{fig:term-semantics} is introduced in Section \ref{sec:model} and Section \ref{sec:semantics}.

\subsection{Soundness and Computational Adequacy}
\label{sub:soundness-adequacy}
In this subsection we prove the main semantic results for our model -- soundness and (strong) adequacy. In order to do so, we first have to prove some useful lemmas.

As usual, the interpretation of values enjoys additional structural properties.

\begin{lemma}
\label{lem:values}
For any value $\Gamma \vdash V : A$, its interpretation $ \lrb V $ is a morphism of $\TD$. Equivalently, it is in the image of $\JJ$.
\end{lemma}
\begin{proof}
Straightforward induction on the derivation of $V$.
\end{proof}

This means the interpretation of each closed value may be seen as a Dirac valuation. Next, we prove a substitution lemma.

\begin{lemma}[Substitution]
\label{lem:term-substitution}
Let $\Gamma \vdash V : A$ be a value and $\Gamma, x : A \vdash M : B$ a term. Then:
\[ \lrb{ M[V/x] } = \lrb M \kcirc (\id_{\lrb \Gamma} \ktimes \lrb V) \kcirc \JJ \langle \id_{\lrb \Gamma}, \id_{\lrb \Gamma} \rangle . \]
\end{lemma}
\begin{proof}
By induction on $M$ using Lemma \ref{lem:values}.
\end{proof}

Soundness and (strong) adequacy are formulated in terms of convex sums of the interpretations of terms. 
For a collection of terms $M_i$ with $\Gamma \vdash M_i : A$, for
each $i \in I$, each interpretation $\lrb{M_i}$ is a map in the
Kegelspitze $\KL(\lrb \Gamma, \lrb A),$ so, we may form convex sums
of these maps. 

Soundness is the statement that our interpretation is invariant under single-step reduction (in a probabilistic sense).

\begin{theorem}[Soundness]
\label{thm:soundness}
For any term $\Gamma \vdash M : A$,
\begin{align*}
\lrb M = \sum_{M \probto{p} M'} p \lrb{M'} , 
\end{align*}
assuming $M \probto{p} M'$ for some rule from Figure \ref{fig:operational} and where the convex sum ranges over all such rules.
\end{theorem}
\begin{proof}
Straightforward induction using Lemma \ref{lem:term-substitution}.
\end{proof}

In the above theorem, the convex sum has at most two summands which are reached after a single reduction step.
The next, considerably stronger statement, generalises this result to reductions involving an arbitrary number of steps.
\emph{Strong adequacy} is the statement that the denotational interpretation is invariant with respect to reduction in a \emph{big-step} sense
(see \cite{pagani-strong-adequacy}, \cite{vakaretal,jones90} where such results are proven).

\begin{theorem}[Strong Adequacy]
\label{thm:strong-adequacy}
For any term $\cdot \vdash M : A,$
\[ \lrb M = \sum_{V \in \Val(M)} P(M \probto{}_* V) \lrb V . \]
\end{theorem}
\begin{proof}
In Appendix \ref{app:adequacy}.
\end{proof}

\begin{remark}
In the above theorem, $\Val(M)$ is defined in \eqref{eq:val(M)} and it may
contain (countably) infinitely many elements; the convex sum is defined in
Definition \ref{def:convex-sums}.
\end{remark}

This theorem is also true to its name, because it immediately implies the usual notion of adequacy.

\begin{corollary}[Adequacy]
\label{cor:adequacy}
Let $\cdot \vdash M : 1$ be a term. Then
\begin{align*}
\lrb M(*)(\{*\}) = \mathrm{Halt}(M) , &&\text{(see \eqref{eq:halt})}
\end{align*}
where $*$ is the unique element of the singleton dcpo $1$.
\end{corollary}
\begin{proof}
Special case of Theorem \ref{thm:strong-adequacy} when $A = 1$ using the fact that if $\cdot \vdash V : 1$ is a value, then $\lrb V(*)(\{*\}) = 1 \in \mathbb R.$
\end{proof}

The commutativity of our monad $\MM$ implies that given
any well-formed terms $\Gamma \vdash M_1: A_1$ and $\Gamma \vdash M_2: A_2$ and $\Gamma, x_1 : A_1, x_2 : A_2 \vdash N : B$, then
\begin{align}
 \lrb{\mathtt{let}\ x_1 = M_1\ \mathtt{in}\ \mathtt{let}\ x_2 = M_2\ \mathtt{in}\ N} = \label{eq:contextual} 
 \lrb{\mathtt{let}\ x_2 = M_2\ \mathtt{in}\ \mathtt{let}\ x_1 = M_1\ \mathtt{in}\ N} ,
\end{align}
where $\mathtt{let}\ x = M\ \mathtt{in}\ N$ may be defined using the usual syntactic sugar.
This, together with adequacy (Corollary \ref{cor:adequacy}) and some standard arguments (see \cite{vakaretal}) implies that the programs in \eqref{eq:contextual} are contextually equivalent.
This improves on the results obtained by Jones~\cite{jones90}, because Equation \ref{eq:contextual} could not be established in her model without a proof that the monad $\VV$ on $\dcpo$ is commutative; as we commented earlier, this remains an open problem.
We finally note that all results in this section also hold for the monads $\mathcal W$ and $\mathcal P.$

\section*{Summary and Future work}\label{sec:summfut}
We have constructed three commutative valuations monads on $\mathbf{DCPO}$
that contain the simple valuations, and shown how to use any of them to give
purely domain-theoretic models for $\mathtt{PFPC}$ that are sound and
adequate.
Our construction using topological methods can be applied
to any \texttt{K}-category $\mathbf K$ with  $\mathbf K\subseteq \mathbf D$,
offering the possibility of further such monads. We also identified the
Eilenberg-Moore algebras of each monad as consisting of Kegelspitzen. In the special case where we consider continuous domains, we
characterized the Eilenberg-Moore algebras over $\mathbf{DOM}$ of all three of our monads and also the $\VV$ monad as precisely
the continuous Kegelspitzen. 
We also proved the most general Fubini theorem for dcpo's yet available.

For future work, we are interested in applying our constructions to extensions of
$\mathtt{PFPC}$. For example, we believe our constructions can be extended to
add sampling, scoring, conditioning and the other tools needed to model statistical
probabilistic programming languages, such as those considered
in~\cite{statonetal,vakaretal}. In particular, the authors of~\cite{vakaretal}
comment that the lack of a commutative monad of valuations on $\mathbf{DCPO}$
is what required them to develop the theory of $\omega$-quasi-Borel spaces.
We believe our approach could support a model of such a statistical programming language solely using domain-theoretic methods,
where we can adapt the ideas from~\cite{mislovedrv} to model random elements; we believe such a model would lead to a simplification of the development.

In a different vein, we plan to apply our results to construct a
model of a programming language that supports both classical
probabilistic effects and also quantum resources.  We have already identified a
suitable type system, where
the probabilistic effects are induced by quantum measurements.  We plan to
interpret the quantum fragment in a category of von Neumann algebras \cite{Kornell18}.
We also plan to show how the decomposition of classical
probabilistic effects in terms of quantum ones can be interpreted by moving
between the Kleisli category of our monad $\MM$ and the category of von Neumann
algebras we identified using the barycentre maps we described in this paper.



\bibliographystyle{ieeetran}
\bibliography{refs}


\appendices
\newpage
\section{Monads, commutativity and $\MM$-algebras}
\label{section:commu}
Let $D$ be a dcpo. Recall that the \emph{d-topology} on $D$ consists of all sub-dcpo's of $D$ as closed subsets. The d-topology on $D$ is finer than the Scott topology. In fact $D$ is even Hausdorff in the d-topology: for $x\not\leq y$ in $D$, $D\setminus {\downarrow}y$ and ${\downarrow}y$ are disjoint open sets in the d-topology, containing $x$ and $y$ respectively. Functions that are continuous between dcpo's equipped with the d-topology are called \emph{d-continuous} functions. Scott-continuous functions between dcpo's are d-continuous~\cite[Lemma 5]{zhao10}.

Recall that $\MM D$ is the smallest sub-dcpo of $\VV D$ that contains $\SSS D$, hence $\MM D$ is actually the topological closure of $\SSS D$ in $\VV D$ equipped with the d-topology. Hence we also say that $\MM D$ is the \emph{d-closure} of $\SSS D$ inside $\VV D$. 

Let $f\colon D\to [0, 1]$ be a Scott-continuous function and $\nu\in \VV D$. The integral $\int_{x\in D}f(x)d\nu$, defined as the Riemann integral $\int_{0}^{1}\nu(f^{-1}((t, 1]))dt$, satisfies the following properties, which can be found in~\cite{jones90}. 
\begin{proposition} 
\label{prop:sumproperty}
Let $D$ be a dcpo, $f\colon D\to [0, 1]$ be a Scott-continuous function. Then we have the following:
\begin{enumerate}
\item \label{sump2} The map $(\nu_{i} \mapsto \sum_{i=1}^{n}r_{i}\nu_{i})\colon \VV D\to \VV D$ is Scott-continuous hence d-continuous, for fixed \label{} $\nu_{j}, j\neq i$ and $r_{i}, i=1,\ldots, n$ with $\sum_{i=1}^{n}r_{i}\leq 1$. 
\item \label{sump4} For $ \sum_{i=1}^{n}r_{i}\nu_{i} \in \VV D$, it is true that $\int fd \sum_{i=1}^{n}r_{i}\nu_{i} =  \sum_{i=1}^{n}r_{i}\int f d\nu_{i}$. 
\item \label{sump44} For $\nu\in \VV D$ and $f, g\in [D\to [0, 1]]$, $\int rf+ sg d\nu = r\int fd\nu + s\int gd\nu$ for $r+s\leq 1$. 
\end{enumerate}
\end{proposition}

\begin{proof}[{\bf Proof of Lemma~\ref{lemma:MDisconvexclosed}   }]~
We prove the case $n=2$ and the general case can be proved similarly. We realize that for a fixed simple valuation $s\in \SSS D$, the map $(\nu\mapsto r_{1}\nu+r_{2}s)\colon \VV D\to \VV D$ maps $\SSS D$ into $\SSS D$. From the previous proposition, Item~\ref{sump2},  this map is d-continuous, it then maps the dcpo-closure of $\SSS D$, which is $\MM D$, into $\MM D$, the dcpo-closure of $\SSS D$. That is, for each simple valuation $s$ and each $\nu\in \MM D$, $ r_{1}\nu+r_{2}s \in \MM D$.  Now we fix $\nu\in \MM D$. Then the map $\xi \mapsto r_{1}\nu + r_{2}\xi \colon \VV D\to \VV D$ maps $\SSS D$ into $\MM D$, hence it also maps $\MM D$ into $\MM D$ since it is d-continuous. This means for $\xi, \nu\in \MM D$, $r_{1}, r_{2}\in [0, 1]$ with $r_{1}+r_{2}\leq 1$, $r_{1}\nu + r_{2}\xi \in \MM D$. 
\end{proof}

\begin{proof}[{\bf Proof of Theorem~\ref{theorem:VPFubini}}]~
To prove this theorem, we first recall two results due to Heckmann~\cite[Theorem~2.4, Theorem~5.5]{heckmann96}. Specifying these results to dcpo $D$, it implies that if $\nu$ is a point-continuous valuation in $\mathcal PD$, and $\nu \in \mathcal O$ for $\mathcal O$ an open set in $\mathcal P_{w} D$,  then there exists a simple valuation $\sum_{i=1}^{n} r_{i}\delta_{x_{i}} \in \SSS D$ such that $\sum_{i=1}^{n} r_{i}\delta_{x_{i}} \leq \nu$ and $\sum_{i=1}^{n} r_{i}\delta_{x_{i}} \in \mathcal O$. 

Now we fix $\xi \in \mathcal P E$ and $U\in \sigma(D\times E)$, and consider the functions 
$$ F\colon  \VV_{w}D \to [0, \infty]\colon \nu \mapsto   \int_{x\in D}\int_{y\in E} \chi_{U}(x, y) d \xi d\nu  $$
and 
$$ G \colon \VV_{w}D \to [0, \infty] \colon \nu \mapsto  \int_{y\in E}\int_{x\in D} \chi_{U}(x, y) d \nu d\xi,  $$
where $[0, \infty]$ is equipped with the Scott topology. We claim that $F$ and $G$ are continuous. 

The fact that $F$ is continuous is straightforward from Remark~\ref{remark:fgeqr}. To see that $G$ is continuous, we assume that $ \int_{y\in E}\int_{x\in D} \chi_{U} d \nu d\xi > r$ and aim to find an open set $\mathcal U$ of $\VV_{w}D$ such that $\nu \in \mathcal U$ and for any $\nu' \in \mathcal U$, $ \int_{y\in E}\int_{x\in D} \chi_{U} d \nu' d\xi > r$. To this end, we note that $g\colon E\to [0, 1] \colon  y \mapsto  \int_{x\in D} \chi_{U}(x, y) d \nu $ is Scott-continuous. Hence $[g>r] \cap \mathcal PE $ is an open subset of $\mathcal P_{w}E$ that contains $\xi$. Applying the aforementioned result we find a simple valuation $\sum_{i=1}^{n} r_{i}\delta_{y_{i}} \in \SSS E$ such that $\sum_{i=1}^{n} r_{i}\delta_{y_{i}} \leq \xi$ and $\sum_{i=1}^{n} r_{i}\delta_{y_{i}} \in [g> r]$. This implies that 
$$ \int_{y\in E}\int_{x\in D} \chi_{U}(x, y) d \nu d \sum_{i=1}^{n} r_{i}\delta_{y_{i}} >r.$$
By applying Equation~\ref{eq:nestedintwithsim}, this in turn implies that 
$$\sum_{i=1}^{n}  \int_{x\in D} r_{i}\chi_{U}(x, y_{i}) d \nu >r.$$
Obviously, we could find $s_{i}\geq 0, i=1, \ldots, n$ such that $  \int_{x\in D} r_{i}\chi_{U}(x, y_{i}) d \nu > s_{i}$ and $\sum_{i=1}^{n}s_{i} > r$. Now we let $$\mathcal U = \bigcap_{i=1}^{n} [r_{i}\chi_{U}(x, y_{i}) > s_{i} ]. $$
By Remark~\ref{remark:fgeqr} the set $\mathcal U$ is open in $\VV_{w}D$ and obviously $\nu \in \mathcal U$. Moreover, for any $\nu'\in \mathcal U$, we have
$$ \int_{y\in E}\int_{x\in D} \chi_{U}(x,y) d \nu' d\xi \geq  \int_{y\in E}\int_{x\in D} \chi_{U}(x,y) d \nu' d  \sum_{i=1}^{n} r_{i}\delta_{y_{i}} =  \sum_{i=1}^{n} \int_{x\in D} r_{i} \chi_{U}(x, y_{i}) d\nu'  \geq  \sum_{i=1}^{n} s_{i} > r. $$
Hence $G$ is continuous indeed. 

The functions $F$ and $G$ are also linear from Proposition~\ref{prop:sumproperty}, Item~\ref{sump4}. Hence both $F$ and $G$ are continuous linear map from $\VV_{w}D$ to $[0, \infty]$, we now apply a varied version of the Schr\"oder-Simpson Theorem, which can be found in~\cite[Corollary~2.5]{goubaultsstheorem}, to see that $F$ and $G$ are uniquely determined by their actions on Dirac measures $\delta_{a}, a\in D$. However, we note that $F(\delta_{a}) = \int_{y\in E} \chi_{U}(a, y) d \xi  = G(\delta_{a})$, again by Equation~\ref{eq:nestedintwithsim}. Hence $F=G$, and we finish the proof by letting $\xi$ range in~$\mathcal P_{w}E$. 
\end{proof}

\begin{proof}[{\bf Proof of Theorem~\ref{theorem:kcofsiscommu}}]~
We only need to prove that the strength of $\VV_{\mathbf K,\leq}$ exists, and is of the same form as $\tau^{\VV}$, the strength of $\VV$, and then conclude with Theorem~\ref{theorem:VPFubini}.

We know that for each \texttt{K}-category $\mathbf K\subseteq \mathbf D$, $\VV_{\mathbf K,\leq}$ is a monad on $\DCPO$. Hence, for any dcpo's $D$ and $E$, and  any Scott-continuous map $f\colon D\to \VV_{\mathbf K,\leq} E$, the function  $$f^{\dagger}\colon \VV_{\mathbf K,\leq} D \to \VV_{\mathbf K,\leq} E \colon \nu \mapsto \lambda U\in \sigma E. \int_{x\in D} f(x)(U)d\nu$$
is a well-defined Scott-continuous map. 

Now we apply this fact to the map $g\colon E\to \VV_{\mathbf K,\leq}(D\times E) \colon y \mapsto \delta_{(a, y)}$, where $a$ is any fixed element in $D$. The map $g$ is obviously Scott-continuous.  Hence for any $\nu\in \VV_{\mathbf K,\leq}E$,  
$$g^{\dagger}(\nu) = \lambda U\in \sigma(D\times E). \int_{y\in E} \delta_{(a, y)}(U) d\nu = \lambda U\in \sigma(D\times E). \int_{y\in E}\chi_{U}(a, y)d\nu$$
is in $\VV_{\mathbf K,\leq}(D\times E)$. 
This implies the map $$\tau_{D, E} \colon D\times \VV_{\mathbf K,\leq} E \to \VV_{\mathbf K,\leq}(D\times E) \colon (a, \nu) \mapsto \lambda U\in \sigma(D\times E). \int_{y\in E}\chi_{U}(a, y)d\nu$$
is well-defined, and it is obviously Scott-continuous. Note that apart from the domain and codomain, the map $\tau_{D, E}$ is same to the strength $\tau_{D, E}^{\VV}$ of $\VV$ at $(D, E)$. Then the same arguments as in Jones' thesis would show that $\tau_{D, E}$ is the strength of $\VV_{\mathbf K,\leq}$ at $(D, E)$. Hence $\VV_{\mathbf K,\leq}$ is a strong monad. 
\end{proof}

\begin{proof}[{\bf Proof of Proposition~\ref{prop:Keglstructure}}]~
We first prove that $K$ is a pointed barycentric algebra. 
It is easy to see that $\beta({\bf 0}_{K})$ is the least element in $K$, since for any $x\in K$,  $\beta({\bf 0}_{K}) \leq \beta(\delta_{x}) = x$. It is also easy to see that $a+_{1}b = a$, $a+_{r}b = b+_{1-r}a$ and $a+_{r} a = a$. We now proceed to prove that  $(a+_{p}b)+_{r}c = a+_{pr}(b+_{\frac{r-pr}{1-pr}} c)$ for any $r, p<1$ and $a, b, c\in K$. To this end, we perform the following:
\begin{align*}
(a+_{p}b)+_{r}c &= \beta(\delta_{a+_{p}b} +_{r} \delta_{c})           						 & \text{definition of $+_{r}$}\\
			    &= \beta(\delta_{\beta(\delta_{a}+_{p}\delta_{b})} +_{r} \delta_{\beta(\delta_{c})})           						 & \text{definition of $+_{p}$ and $\beta(\delta_{c}) =c$}\\
			    &= \beta(\MM(\beta) (\delta_{\delta_{a}+_{p} \delta_{b}}+_{r}\delta_{\delta_{c}})) 	& \text{$\MM(\beta)$ is linear and $\MM(\beta)(\delta_{\nu}) = \delta_{\beta(\nu)}$} \\
			    &= \beta(\mu^{\MM}_{K} (\delta_{\delta_{a}+_{p} \delta_{b}}+_{r}\delta_{\delta_{c}})) &\text{$(K, \beta)$ is an \MM-algebra}\\
			    &= \beta((\delta_{a}+_{p}\delta_{b})+_{r}\delta_{c} )						&\text{$\mu^{\MM}_{K}$ is the multiplication of~$\MM$ at $K$}\\
			    &= \beta(\delta_{a}+_{pr}(\delta_{b}+_{\frac{r-rp}{1-pr}}\delta_{c} ))	       &\text{$\MM K$ is a Kegelspitze}\\	
			    &= \beta(\mu^{\MM}_{K}( \delta_{\delta_{a}}+_{pr}\delta_{(\delta_{b}+_{\frac{r-rp}{1-pr}}\delta_{c} )}))  &\text{$\mu^{\MM}_{K}$ is the multiplication of~$\MM$ at $K$}\\
			    &= \beta(\MM(\beta)(  \delta_{\delta_{a}}+_{pr}\delta_{(\delta_{b}+_{\frac{r-rp}{1-pr}}\delta_{c} )} )) &\text{$(K, \beta)$ is an \MM-algebra}\\
			    &= \beta( \delta_{\beta(\delta_{a})} +_{pr} \delta_{\beta(  \delta_{b}+_{\frac{r-rp}{1-pr}}\delta_{c})}   ) & \text{$\MM(\beta)$ is linear and $\MM(\beta)(\delta_{\nu}) = \delta_{\beta(\nu)}$}\\
			    &= \beta( \delta_{a} +_{pr} \delta_{({b}+_{\frac{r-rp}{1-pr}}{c} )})       	 & \text{definition of ${b}+_{\frac{r-rp}{1-pr}}{c} $ and $\beta(\delta_{a}) =a$}\\
			    &= a+_{pr}(b+_{\frac{r-pr}{1-pr}} c).									 & \text{definition of $+_{pr}$}
\end{align*}
The map $(a, b)\mapsto a+_{r}b = \beta (\delta_{a} +_{r} \delta_{b}) \colon K\times K \to K$ is Scott-continuous since $\beta$ and $\delta$ are Scott-continuous and $\MM K$ is a Kegelspitze. The map $(r, a) \mapsto ra = a+_{r} \beta({\mathbf 0}_{K}) = \beta(\delta_{a} +_{r} \delta_{\beta({\mathbf 0}_{K})})\colon [0, 1]\times K\to K$ is Scott-continuous in $a$ for the exactly same reasons; to see that it also is Scott-continuous in $r$, we only need to show that $r\mapsto \delta_{a} +_{r} \delta_{\beta({\mathbf 0}_{K})}\colon [0, 1] \to \MM K$ is Scott-continuous for any fixed $a\in K$. This is true if $\beta(\mathbf 0_{K})\leq a$. However, we already see that $\beta(\mathbf 0_{K})$ is the least element in $K$.  Hence we have proved that $K$ is a Kegelspitze. The map $\beta$ is clearly linear. 
\end{proof}

\begin{proof}[{\bf Proof of Proposition~\ref{prop:EMmaps}}]~
\paragraph*{\textbf{The ``if'' direction}} Assume that $f\colon K_{1}\to K_{2}$ is linear. We need to prove that $f\circ \beta_{1} = \beta_{2}\circ \MM(f)$. Since both sides are Scott-continuous hence d-continuous and $K_{2}$ is Hausdorff in the d-topology (if $K_{2}$ has more than one elements). We only need to prove they are equal on simple valuations on $K_{1}$. To this end, we pick $\sum_{i=1}^{n}r_{i}\delta_{x_{i}}\in \MM K_{1}$, and see 
\begin{align*}
f( \beta_{1} (\sum_{i=1}^{n}r_{i}\delta_{x_{i}})) &= f( \sum_{i=1}^{n}r_{i}x_{i} )  &\text{$\beta_{1}$ is linear and $\beta_{1}(\delta_{x_{i}}) = x_{i}$}\\
										 &= \sum_{i=1}^{n}r_{i}f(x_{i})   &\text{$f$ is linear }\\
										 &= \beta_{2} (\sum_{i=1}^{n}r_{i}\delta_{f(x_{i})} )  &\text{$\beta_{2}$ is linear and  $\beta_{2}(\delta_{f(x_{i})}) =f( x_{i})$}\\
										 &= \beta_{2}(\MM(f)(\sum_{i=1}^{n}r_{i}\delta_{x_{i}} )).  &\text{$\MM(f)$ is linear and $\MM(f)(\delta_{x_{i}}) = \delta_{f(x_{i})}$}\\
\end{align*}
\paragraph*{\textbf{The ``only if'' direction}} Assume that $f\colon K_{1} \to K_{2}$ is an algebra morphism from $(K_{1}, \beta_{1})$ to $(K_{2}, \beta_{2})$. Then we know that  $f\circ \beta_{1} = \beta_{2}\circ \MM(f)$. We prove that $f$ is linear. First, for $a, b\in K_{1}$ and $r\in [0, 1]$, we have
\begin{align*}
f(a+_{r} b) &= f( \beta_{1}(\delta_{a}+_{r}\delta_{b}))  &\text{definition of $a+_{r}b$}\\
		   &= \beta_{2}( \MM(f) ( \delta_{a}+_{r}\delta_{b} ) ) &\text{$f$ is an algebra morphism }\\
		   &= \beta_{2} (\delta_{f(a)} +_{r} \delta_{f(b)}) &\text{$\MM (f)$ is linear and $\MM(f)(\delta_{x}) = \delta_{f(x)}$}\\
		   &=f(a) +_{r} f(b).  &\text{definition of $f(a)+_{r}f(b)$}
\end{align*}
Second, to prove that $f$ maps $\beta(\mathbf 0_{K_{1}})$ to $\beta_{2}(\mathbf 0_{K_{2}})$, we see that $f(\beta_{1}(\mathbf 0_{K_{1}})) = \beta_{2}(\MM(f)(\mathbf 0_{K_{1}}))= \beta_{2}(\mathbf 0_{K_{2}})$ because $\MM(f)$ is linear. 
\end{proof}

\newpage
\section{Solving Recursive Domain Equations in $\KL$}
\label{app:domain-equations}

We use $(\MM, \eta, \mu, \tau)$ to indicate our commutative monad and we write $(\LL, \eta^\LL, \mu^\LL, \tau^\LL)$ to indicate the lift monad on $\DCPO$, which is also commutative.

Recall that the lift monad $\LL : \DCPO \to \DCPO$ freely adds a new least element, often denoted $\perp$, to a dcpo $X$. The resulting dcpo is $\LL X \defeq X_\perp$. The monad structure of $\LL$ is defined by the following assignments:

\begin{minipage}{0.3\textwidth}
\begin{align*}
\eta^\LL_X : X &\to X_\perp \\ 
x &\mapsto x
\end{align*}
\end{minipage}
\begin{minipage}{0.3\textwidth}
\begin{align*}
\mu^\LL_X : (X_{\perp_1})_{\perp_2} &\to X_\perp \\
x                                &\mapsto
  \begin{cases}
    \perp & \text{, if } x = \perp_1 \\
    \perp & \text{, if } x = \perp_2 \\
    x     & \text{, if } \perp_1 \neq x \neq \perp_2
  \end{cases}
\end{align*}
\end{minipage}
\begin{minipage}{0.3\textwidth}
\begin{align*}
\tau^\LL_{XY} : X_\perp \times Y &\to (X \times Y)_\perp \\
(x,y)                      &\mapsto
  \begin{cases}
    \perp & \text{, if } x = \perp \\
    (x,y) & \text{, if } x \neq \perp
  \end{cases}
\end{align*}
\end{minipage}

We write $\DCPO_\LL$ for the Kleisli category of $\LL$ and we write its morphisms as $f : X \pto Y,$ which is by definition a morphism $f : X \to Y_\perp$ in $\DCPO$.
We write $X \otimes Y$ and $X \oplus Y$ for the symmetric monoidal product and coproduct, respectively, which are (canonically) induced by the commutative monad $\LL$.

\begin{proposition}
\label{prop:map-of-monads}
The assignment $\phi : \LL \naturalto \MM$ defined by
\begin{align*}
\phi_X : X_\perp &\to \MM X \\
x                &\mapsto
  \begin{cases}
    {\mathbf 0}_{X}        & \text{, if } x = \perp \\
    \delta_x & \text{, if } x \neq \perp 
  \end{cases}
\end{align*}
is a strong map of monads (see \cite[Definition 5.2.9]{jacobs-coalgebra} for more details).
\end{proposition}
\begin{proof}
To see that $\phi$ is a natural transformation, we need to show, for any Scott-continuous map $f\colon X\to Y$, $\phi_{Y}\circ \LL f = \MM f \circ \phi_{X}\colon X_{\perp} \to \MM Y$. However, it is easy to see that both sides send $\perp$ to ${\mathbf 0}_{Y}$ and $x$ that is not $\perp$ to~$\delta_{f(x)}$.  

Now, we first verify that $\phi$ is a map of monads. That is, for each dcpo $X$, we need to prove that $\phi_{X}\circ \eta^\LL_{X}= \eta_{X}$ and $\phi_{X}\circ \mu^\LL_{X} = \mu_{X}\circ \MM(\phi_{X})\circ \phi_{X_{\perp}}\colon (X_{\perp_{1}})_{\perp_{2}} \to \MM(X)$. 

The first equation is trivial, hence we proceed to prove the second. For this, we see
\[\phi_{X}\circ \mu^\LL_{X} (x) =  
\begin{cases}
\phi_{X}(\perp) = {\mathbf 0}_{X} & \text{, if } x = \perp_{1}~\text{or}~x=\perp_2\\
\phi_{X}(x) = \delta_{x} & \text{, if } \perp_1 \neq x \neq \perp_2
\end{cases}
\]
and 
\[
\mu_{X}\circ \MM(\phi_{X})\circ \phi_{X_{\perp}}(x) =  
\begin{cases}
\mu_{X}\circ \MM(\phi_{X})({\mathbf 0}_{X_{\perp}}) = \mu_{X}({\mathbf 0}_{\MM X})   ={\mathbf 0}_{X}& \text{, if } x = \perp_{2}\\
\mu_{X}\circ \MM(\phi_{X})(\delta_{\perp}) = \mu_{X}(\delta_{\phi_{X}(\perp)}) =\mu_{X}(\delta_{{\mathbf 0}_{X}}) = {\mathbf 0}_{X} & \text{, if } x = \perp_{1}\\
\mu_{X}\circ \MM(\phi_{X})(\delta_{x}) = \mu_{X}(\delta_{\phi_{X}(x)}) = \mu_{X}(\delta_{\delta_{x}}) = \delta_{x}  & \text{, if } \perp_1 \neq x \neq \perp_2.
\end{cases}
\]
Hence $\phi: \LL \naturalto \MM$ is a map of monads. 

To prove that $\phi$ is a strong map of monads, we need to show that for any dcpo's $X$ and $Y$, 
\[ \tau_{XY}\circ (\phi_{X}\times \id_{Y}) =  \phi_{XY}\circ \tau^\LL_{XY} \colon X_{\perp}\times Y \to \MM(X\times Y).\]
The strength $\tau$ of $\MM$ at $(X, Y)$ is defined as follows:
\[\tau_{XY}\colon \MM X \times Y \to \MM(X\times Y)\colon (\nu, y) \mapsto \lambda U. \int_{x\in X} \chi_{U}(x, y)d\nu,\]
where $\chi_{U}$ is the characteristic function of $U\in \sigma(X\times Y)$, i.e., $\chi_{U}(x, y) = 1$ if $(x, y)\in U$ and  $\chi_{U}(x, y) = 0$, otherwise.  
Now we perform the following computation
\[  
\tau_{XY}\circ (\phi_{X}\times \id_{Y})(x, y) =  
\begin{cases}
\tau_{XY} ({\mathbf 0}_{X}, y) = \lambda U. \int_{x\in X}\chi_{U}(x, y)d{\mathbf 0}_{X} = \lambda U. 0 = {\mathbf 0}_{X\times Y} & \text{, if } x = \perp\\
\tau_{XY} (\delta_{x}, y) =  \lambda U. \int_{x\in X}\chi_{U}(x, y)d\delta_{x} = \lambda U.\chi_{U}(x, y) = \delta_{(x, y)}    & \text{, if }  x \neq \perp
\end{cases}
\]
and 
\[ 
\phi_{XY}\circ \tau^\LL_{XY}(x, y) =  
\begin{cases}
\phi_{XY}(\perp)= {\mathbf 0}_{X\times Y} & \text{, if } x = \perp\\
\phi_{XY}((x, y)) = \delta_{(x, y)}    & \text{, if }  x \neq \perp
\end{cases}
\]
which concludes the proof.
\end{proof}

Recall that any map of monads induces a functor between the corresponding Kleisli categories of the two monads (see \cite[Exercise 5.2.1]{jacobs-coalgebra}).
This allows us to show the next corollary.

\begin{corollary}
\label{cor:preservation}
The functor $\FF : \DCPO_\LL \to \DCPO_\MM$, induced by $\phi: \LL \naturalto \MM,$ and defined by:
\begin{align*}
\FF X &\defeq X \\
\FF (f: X \pto Y) &\defeq \phi_Y \circ f
\end{align*}
\emph{strictly} preserves the monoidal and coproduct structures in the sense that
the following equalities:
\begin{align*}
\FF(X \otimes Y) &= \FF X \ktimes \FF Y & \FF(X \oplus Y) &= \FF X \kplus \FF Y \\
\FF(f \otimes g) &= \FF f \ktimes \FF g & \FF(f \oplus g) &= \FF f \kplus \FF g 
\end{align*}
hold.
\end{corollary}
\begin{proof}
This follows by canonical categorical arguments and is just a straightforward verification.
\end{proof}

Before we may prove our next proposition, let us recall an important result from \cite{smyth-plotkin:domain-equations}.
\begin{proposition}
\label{prop:functor-on-embeddings}
Let $\AAA$, $\BB$ and $\CC$ be $\DCPO$-enriched categories. Assume further that $\AAA$ and $\BB$ have all $\omega$-colimits (or all $\omega^\op$-limits). If $\TTT : \AAA^\op \times \BB \to \CC$ is a $\dcpo$-enriched functor,
then the assignment
\begin{align*}
\TTT^E \colon \AAA_e \times \BB_e &\to \CC_e \\
\TTT^E(A,B) &\defeq \TTT(A,B) \\
\TTT^E(e_1, e_2) &\defeq \TTT(e_1^p, e_2)
\end{align*}
defines a \emph{covariant} $\omega$-cocontinuous functor.
\end{proposition}
\begin{proof}
This follows by combining several results from \cite{smyth-plotkin:domain-equations}, namely Theorem 2, the corollary after it and Theorem 3.
\end{proof}

Therefore, by trivialising the category $\AAA$, we may obtain results for
purely covariant functors. When neither category is trivialised, this allows us
to interpret mixed-variance functors (such as function space) as covariant
functors on subcategories of embeddings.

\begin{proposition}
\label{app:omega-functors}
The category $\PD_e$ has an initial object and all $\omega$-colimits and the following assignments:
\begin{align*}
\ktimes_e \colon \PD_e \times \PD_e &\to \PD_e  & \kplus_e \colon \PD_e \times \PD_e &\to \PD_e \\
X \ktimes_e Y &\defeq X \ktimes Y & X \kplus_e Y &\defeq X \kplus Y \\
e_1 \ktimes_e e_2 &\defeq e_1 \ktimes e_2 & e_1 \kplus_e e_2 &\defeq e_1 \kplus e_2
\end{align*}
\begin{align*}
[\kto]_e^\JJ \colon \PD_e \times \PD_e &\to \PD_e \\
[X \kto Y]_e^\JJ &\defeq \JJ [ X \kto Y] \\
[e_1 \kto e_2]_e^\JJ &\defeq \JJ [e_1^p \kto e_2]
\end{align*}
define \emph{covariant} $\omega$-cocontinuous functors on $\PD_e$.
\end{proposition}
\begin{proof}
The empty dcpo $\varnothing$ is a zero object in $\PD$ such that each map $e : \varnothing \kto X $ is an embedding and each map $p : X \kto \varnothing $ is a projection.
Therefore, $\varnothing$ is initial in $\PD_e$. The existence of all $\omega$-colimits in $\PD_e$ follows from the existence of all $\omega$-colimits of $\PD$ together with results from \cite{smyth-plotkin:domain-equations}.

Next, we show that $\ktimes \colon \KL \times \KL \to \KL$ restricts to a
functor $\ktimes^\PD \colon \PD \times \PD \to \PD.$ On objects, this is
obvious. For morphisms, observe that the morphisms of $\PD$ are exactly those
which are in the image of $\FF$. Therefore $\ktimes^\PD$ restricts as indicated
because $\FF f \ktimes \FF g = \FF(f \otimes g)$ by Corollary
\ref{cor:preservation}. Then, by Proposition \ref{prop:functor-on-embeddings}, it follows that
$(\ktimes^\PD)^E \colon \PD_e \times \PD_e \to \PD_e$ is a covariant $\omega$-cocontinuous functor.
However, by definition, $\ktimes_e = (\ktimes^\PD)^E$ which shows the result for $\ktimes_e$.

Exactly the same argument (swapping $\ktimes$ for $\kplus$ and $\otimes$ for $\oplus$) shows the
result for $\kplus_e$.

For function spaces, consider the functor $\JJ \circ [ \kto ] : \KL^\op \times \KL \to \KL.$ This composition (co)restricts to a functor
$(\JJ \circ [\kto])^\PD \colon \PD^\op \times \PD \to \PD,$ because $\JJ(f \kto g) = \eta \circ (f \kto g) = \phi \circ \eta^\LL \circ (f \kto g) = \FF(\eta^\LL \circ (f \kto g)).$ 
By Proposition \ref{prop:functor-on-embeddings}, it follows $((\JJ \circ [\kto])^\PD)^E \colon \PD_e \times \PD_e \to \PD_e$ is a covariant $\omega$-cocontinuous functor.
Finally, by definition, $[\kto]_e^\JJ = ((\JJ \circ [\kto])^\PD)^E$ which concludes the proof.
\end{proof}

We conclude the appendix with a proof that the subcategories $\TD$ and $\PD$ contain the same isomorphisms.

\begin{proposition}
Every isomorphism of $\PD$ is also an isomorphism of $\TD$.
\end{proposition}
\begin{proof}
Observe that, by definition, the morphisms of $\TD$ are those in the image of $\JJ \colon \DCPO \to \KL$ and the morphisms of $\PD$ are those in the image of $\FF \colon \DCPO_\LL \to \KL.$
Then, it is easy to see that the following diagram:
\[ \stikz{dcpo-lift-v.tikz} \]
commutes, where:
\begin{itemize}
\item the top arrow is the subcategory inclusion $\TD \hookrightarrow \PD$;
\item the left vertical isomorphism is the corestriction of $\JJ$ to $\TD$;
\item the right vertical isomorphism is the corestriction of $\FF$ to $\PD$;
\item the functor $\JJ^\LL$ is the Kleisli inclusion of $\DCPO$ into $\DCPO_\LL$, defined by $\JJ^\LL(X) \defeq X$ and $\JJ^\LL(f) \defeq \eta^\LL \circ f.$
\end{itemize}
It is well-known (and easy to prove) that if $f: X \pto Y$ in $\DCPO_\LL$ is an isomorphism, then there exists $f' : X \to Y$ in $\DCPO$ which is also an isomorphism and $f = \JJ^\LL (f').$
The proof is finished by a simple diagram chase using this fact.
\end{proof}

\newpage
\section{ Products, coproducts and Kleisli composition preserve barycentric sums of functions}
\label{app:timesplusconvex}
The  monoidal product $\_\ktimes \_ \colon \KL \times \KL\to \KL$ is defined as: for dcpo's $A$ and $B$, $A\ktimes B\defeq A\times B$, and  for Scott-continuous maps $f\colon A\to \MM C$ and $g\colon B\to \MM D$, $f\ktimes g \defeq \lambda(a, b). f(a)\otimes g(b)$, where $f(a)\otimes g(b)$ is defined in Remark~\ref{remark:tensorofvaluations}. For $f, h\colon A\to \MM C$ and $r\in [0, 1]$, $f+_{r}h$ is defined pointwise, that is, $(f+_{r}h)(a) = f(a)+_{r}h(a) = rf(a) + (1-r)h(a).$
It follows from Lemma~\ref{lemma:MDisconvexclosed} that $f+_{r}h$ is well-defined and obviously $f+_{r}h$ is Scott-continuous, hence $f+_{r}h\in [A\to \MM C]$. 

\begin{proposition}
For $f, h\colon A\to \MM C$, $g\colon B\to \MM D$ and $r\in [0, 1]$, we have
\begin{enumerate}
\item $(f+_{r}h)\ktimes g = f\ktimes g +_{r} h\ktimes g \colon A\ktimes B\to \MM(C\ktimes D)$;
\item $g\ktimes (f+_{r}h) = g\ktimes f +_{r} g\ktimes h \colon B\ktimes A\to \MM(D\ktimes C)$. 
\end{enumerate}

\end{proposition}
\begin{proof}
We only prove Item~1, the second item can be proved similarly. \\
For each $(a, b)\in A\ktimes B$, we have the following:
\begin{align*}
 &((f+_{r}h)\ktimes g) (a, b)  	 \\&= (f+_{r}h)(a) \otimes g(b) &\text{definition of $\_\ktimes \_$}\\
                                            &= (f(a)+_{r}h(a))\otimes g(b) & \text{definition of $f+_{r}h$}\\
                                            &= \lambda U\in \sigma(C\ktimes D). \int_{y\in D}\int_{x\in C}\chi_{U}(x, y)d(f(a)+_{r}h(a))d g(b) &\text{definition of the operation~$\otimes$}\\ 
                                            &= \lambda U. \int_{y\in D}( \int_{x\in C}\chi_{U}(x, y)df(a) +_{r} \int_{x\in C}\chi_{U}(x, y)dh(a))d g(b)& \text{by Proposition~\ref{prop:sumproperty}, Item~\ref{sump4}}\\ 
                                            &= \lambda U. \int_{y\in D} \int_{x\in C}\chi_{U}(x, y)df(a) d g(b) +_{r}  \int_{y\in D}\int_{x\in C}\chi_{U}(x, y)d h(a)d g(b)& \text{by Proposition~\ref{prop:sumproperty}, Item~\ref{sump44}}\\     
                                            &= \lambda U. \int_{y\in D} \int_{x\in C}\chi_{U}(x, y)df(a) d g(b) +_{r} \lambda U. \int_{y\in D}\int_{x\in C}\chi_{U}(x, y)d h(a)d g(b)& \text{by definition of $+_{r}$ of valuations}\\ 
                                            &= (f \ktimes g)(a, b)  +_{r} (h\ktimes g)(a, b)  & \text{definition of $\_\ktimes \_$}\\
                                            &=(f \ktimes g +_{r} h\ktimes g)(a, b)  & \text{definition of $+_{r}$ of functions.}  
\end{align*}
Hence the proof is completed. 
\end{proof}

The functor $\_\kplus \_ \colon \KL \times \KL\to \KL$ is defined as: for dcpo's $A$ and $B$, $A\kplus B\defeq A+ B$, and  for Scott-continuous maps $f\colon A\to \MM C$ and $g\colon B\to \MM D$, $f\kplus g = [\MM(i_{C}) \circ f, \MM(i_{D}) \circ g]$, where $i_{C}\colon C\to C+D$ and $i_{D}\colon D\to C+D$ are the obvious injections. 
\begin{proposition}
For $f, h\colon A\to \MM C, g\colon B\to \MM D$ and $r\in [0, 1]$, we have
\begin{enumerate}
\item $(f+_{r}h)\kplus g = (f\kplus g) +_{r} (h\kplus g)$;
\item $g\kplus (f+_{r}h) = (g\kplus f) +_{r} (g\kplus h)$. 
\end{enumerate}
\end{proposition}
\begin{proof}
Again, we only prove the first claim as the second can be proved similarly. Let $a\in A$, we perform the following computation:
\begin{align*}
((f+_{r}h)\kplus g )(i_{A}(a)) 	 &= [\MM(i_{C}) \circ (f+_{r}h), \MM(i_{D}) \circ g] (i_{A}(a))   &\text{definition of $\_\kplus \_$}\\
                                            &= \MM(i_{C})((f+_{r}h)(a))  &\text{obvious}\\
                                            &= \MM(i_{C})(f(a) +_{r} h(a))   &\text{definition of $f+_{r}h$} \\
                                            &= \lambda U. (f(a) +_{r} h(a))( i^{-1}_{C}(U)) & \text{definition of $\MM(i_{C})$}   \\
                                            &= \lambda U. f(a)(i^{-1}_{C}(U)) +_{r} h(a)( i^{-1}_{C}(U)) &\text{definition of $f(a) +_{r} h(a)$}\\
                                            &= \lambda U. f(a)(i^{-1}_{C}(U)) +_{r} \lambda U. h(a)( i^{-1}_{C}(U)) &\text{definition of $+_{r}$ of valuations}\\
                                            &= \MM(i_{C})(f(a)) +_{r}  \MM(i_{C})(h(a)) , & \text{definition of $\MM(i_{C})$} \\ 
                                            &=(f\kplus g)(i_{A}(a)) +_{r} (h\kplus g)(i_{A}(a)) &\text{definition of $\_\kplus \_$}\\
                                            &=( (f\kplus g) +_{r} (h\kplus g))(i_{A}(a)). &\text{definition of  $(f\kplus g) +_{r} (h\kplus g)$}
\end{align*}
Moreover, it is easy to see that for $b\in B$, $((f+_{r}h)\kplus g )(i_{B}(b)) = \MM (i_{D}) (g(b)) = \MM (i_{D}) (g(b)) +_{r} \MM (i_{D}) (g(b))= ( (f\kplus g) +_{r} (h\kplus g))(i_{B}(b))$. Hence we finish the proof. 
\end{proof}

Recall that in $\KL$ the Kleisli composition  $\kcirc \colon [A\kto B]\times [B\kto C] \to [A\kto C]$ is given by  
\[
(f, g) \mapsto g\kcirc f = g^{\ddagger} \circ f.
\]

\begin{proposition}
For $f, h\colon A\to \MM B, g, k \colon B\to \MM C$ and $r\in [0, 1]$, we have
\begin{enumerate}
\item $g\kcirc (f+_{r}h) = g\kcirc f +_{r} g\kcirc h$;
\item $(g+_{r} k) \kcirc f = g\kcirc f +_{r} k\kcirc f$.
\end{enumerate}
\end{proposition}
\begin{proof}
\begin{enumerate}
\item Let $a\in A$. We have 
\begin{align*}
g\kcirc (f+_{r}h)(a)  &=  (g^{\ddagger}\circ (f+_{r}h)) (a)  &\text{definition of $\kcirc$}\\
				   &=  g^{\ddagger} (f(a)+_{r}h(a))  &\text{definition of $f+_{r}h$}\\
				   &= \lambda U. \int_{x\in B} g(x)(U)d (f(a)+_{r}h(a)) &\text{definition of $g^{\ddagger}$}\\
				   &= \lambda U. \int_{x\in B} g(x)(U)d f(a) +_{r} \lambda U. \int_{x\in B}g(x)(U) d h(a) &\text{by Proposition~\ref{prop:sumproperty}, Item~\ref{sump4}} \\
				   &= g^{\ddagger}( f(a) )+_{r} g^{\ddagger}(h(a)) &\text{definition of $g^{\ddagger}$}\\
				   &= (g\kcirc f +_{r} g\kcirc h)(a). 
\end{align*}
\item Let $a\in A$. We have 
\begin{align*}
((g+_{r} k) \kcirc f) (a)  &= (g+_{r}k)^{\ddagger}(f(a)) &\text{definition of $\kcirc$}\\
				   &= \lambda U. \int_{x\in B} (g+_{r}k)(x)(U) d f(a) &\text{definition of $\_^{\ddagger}$}\\
				   &= \lambda U. \int_{x\in B} g(x)(U)d f(a) +_{r} \lambda U. \int_{x\in B}k(x)(U) d f(a) &\text{by Proposition~\ref{prop:sumproperty}, Item~\ref{sump44}} \\
				   &= g^{\ddagger}( f(a) )+_{r} k^{\ddagger}(f(a)) &\text{definition of $\_^{\ddagger}$}\\
				   &= (g\kcirc f +_{r} k\kcirc f)(a). 
\end{align*}
\end{enumerate}
\end{proof}

\newpage
\section{Proof of Strong Adequacy}
\label{app:adequacy}

The purpose of this appendix is to provide a proof Theorem \ref{thm:strong-adequacy}.
We begin by stating a corollary for the soundness theorem.

\begin{corollary}
\label{cor:soundness-inequality}
For any closed term $\cdot \vdash M: A$, we have:
\[ \lrb M \geq \sum_{V \in \Val(M)} P(M \probto{}_* V) \lrb V . \]
\end{corollary}
\begin{proof}
First, let us decompose the convex sum on the right-hand side.
\begin{align*}
\sum_{V \in \Val(M)} P(M \probto{}_* V) \lrb V &= \sup_{\substack{F \subseteq \Val(M) \\ F \text{ finite}}} \sum_{V \in F} P(M \probto{}_* V) \lrb V & \text{(Definition)} \\
&= \sup_{\substack{F \subseteq \Val(M) \\ F \text{ finite}}} \sum_{V \in F} \left( \sup_{i \in \mathbb N} P(M \probto{}_{\leq i} V) \right) \lrb V & \text{(Definition)} \\
&= \sup_{\substack{F \subseteq \Val(M) \\ F \text{ finite}}} \sup_{i \in \mathbb N} \sum_{V \in F}  P(M \probto{}_{\leq i} V) \lrb V & \left( \text{Scott-continuity of } \sum_i r_i a_i \text{ in each } r_i \right).
\end{align*}
Therefore, it suffices to show that
\begin{equation}
\label{eq:soundness-corollary}
\lrb M \geq \sum_{V \in F}  P(M \probto{}_{\leq i} V) \lrb V
\end{equation}
for any choice of finite $F \subseteq \Val(M)$ and $i \in \mathbb N.$ This can now be shown by induction on $i$. If $M \in F$ (which means $M$ is a value), then \eqref{eq:soundness-corollary} is a strict equality.
Assume $M \not \in F$.
If $i=0$, then the right-hand side of \eqref{eq:soundness-corollary} is 0 and so the inequality holds.
For the step case, if $M$ is a value, then RHS is 0 and the inequality holds. Otherwise:
\begin{align*}
\sum_{V \in F}  P(M \probto{}_{\leq i+1} V) \lrb V &= \sum_{V \in F}  \sum_{M \probto p M'} p \cdot P(M' \probto{}_{\leq i} V) \lrb V & \\
&= \sum_{M \probto p M'} p \cdot \sum_{V \in F}  P(M' \probto{}_{\leq i} V) \lrb V &\\
&\leq \sum_{M \probto p M'} p \cdot \lrb{M'}  & (\text{IH for } M') \\
&= \lrb M &(\text{Soundness})
\end{align*}
where we also implicitly used the fact that $\Val(M') \subseteq \Val(M).$
\end{proof}

The remainder of the appendix is dedicated to showing the converse inequality, which is considerably more difficult to prove.

\subsection{Overview of the Proof Strategy}

The proof of strong adequacy requires considerable effort. Our proof strategy consists in formulating logical relations that we use to prove our adequacy result.
These logical relations are described in Theorem \ref{thm:formal-relations} and the design of our logical relations follows that of Claire Jones in her thesis \cite{jones90}.
Once this theorem is proved, the proof of adequacy is fairly straightforward. We use the logical relations to establish some useful closure properties in Subsection \ref{sub:closure-properties} and this allows us to easily prove
Lemma \ref{lem:fundamental}, which is often called the Fundamental Lemma. This lemma easily implies Strong Adequacy as we show.

Most of the effort in proving our Strong Adequacy result lies in the proof of Theorem \ref{thm:formal-relations}. It is not possible to use the properties (A1) -- (A4)
as a definition of the relations, because then condition (A4) would be defined via non-well-founded induction.
The proof of the existence of this family of relations is not obvious. We use techniques from \cite{icfp19,lnl-fpc-lmcs} (which are in turn based on ideas from \cite{fiore-thesis}) to show the existence of these relations.
The main idea of the proof of existence is to define, for every type $A$, a category $\RR(A)$ of logical relations with a suitable notion of morphism. We then show that every such category has sufficient structure to construct
parameterised initial algebras (Proposition \ref{prop:logical-colimits}). We may then define functors on these categories (Definition \ref{def:logical-functors})
which construct logical relations in the same manner as they are needed in Theorem \ref{thm:formal-relations}. These functors are $\omega$-cocontinuous (Proposition \ref{prop:logical-functors}) which means that we may form
(parameterised) initial algebras using them. This allows us to define an \emph{augmented interpretation of types} on the categories $\RR(A)$ which satisfies some important coherence conditions with respect to the standard
interpretation of types (Corollary \ref{cor:cool-form}). These coherence conditions show that each augmented interpretation $\elrbs A$ of a type $A$ contains the standard interpretation $\lrb A$, together with the logical relation that we need,
as shown in Theorem \ref{thm:formal-relations}.

\subsection{Logical Relations}

\begin{assumption}
Throughout this appendix, we assume that all types are closed, unless otherwise noted.
\end{assumption}

\begin{definition}
For each type $A$, we write:
\begin{itemize}
  \item $\Val(A) \defeq \{ V \ |\ V \text{ is a value and } \cdot \vdash V : A\}.$
  \item $\Prog(A) \defeq \{ M \ |\ M \text{ is a term and } \cdot \vdash M : A\}.$
\end{itemize}
\end{definition}

Next, we define sets of relations that are parameterised by dcpo's $X$ from our semantic category, types $A$ from our language and partial deterministic embeddings $e_X : X \kto \lrb A$ which show how $X$ approximates $\lrb A$.
We shall write relation membership in infix notation, that is, for a binary relation $\tleq$, we write $v \tleq V$ to indicate $(v, V) \in \tleq.$

\begin{definition}
\label{def:logical-relations}
For any dcpo $X$, type $A$ and morphism $e \colon X \kto \lrb A$ in $\PDe$, let:
\begin{align*}
\ValRel(X, A, e)  = \{ \tleq_{X,A}^e \subseteq \TD(1, X) \times \Val(A) \ |\ &\forall V \in \Val(A).\ (-) \tleq_{X,A}^e V \text{ is a Scott closed subset of } \TD(1,X) \text{ and} \\
&\forall V \in \Val(A).\ v \tleq_{X,A}^e V \Rightarrow e \kcirc v \leq \lrb V \}. 
\end{align*}
\end{definition} 

\begin{remark}
In the above definition, relations $\tleq_{X,A}^e \in \ValRel(X,A,e)$ can be seen as ternary relations $\tleq_{X,A}^e \subseteq \TD(1,X) \times \Val(A) \times \{e\}$. However, since there is no choice for the third component, we prefer to see them as binary relations that are parameterised by the embeddings $e$. Indeed, this leads to a much nicer notation.
We shall also sometimes indicate the parameters $X, A$ and $e$ of the relation in order to avoid confusion as to which set $\ValRel(X,A,e)$ it belongs to.
\end{remark}

The relations we need for the adequacy proof inhabit the sets $\ValRel(\lrb A, A, \kid_{\lrb A})$. In the remainder of the appendix, we will show how to choose exactly one relation (the one we need) from each of those sets.

Before we may define the relation constructors we need, we have to introduce some auxiliary definitions.

\begin{definition}
\label{def:paths}
Let $M \colon A$ and $N \colon A$ be closed terms of the same type. We define
\[ \mathrm{Paths(M, N)} \defeq \left\{ \pi\ |\ \pi = \left( M = M_0  \probto{p_0} M_1 \probto{p_1} M_2 \probto{p_2} \cdots \probto{p_n} M_n = N \right)  \text{ is a reduction path} \right\} . \]
In other words, $\mathrm{Paths(M, N)}$ is the set of all reduction paths from $M$ to $N$. The \emph{probability weight} of a path $\pi \in \Paths(M,N)$ is $P(\pi) \eqdef \prod_{i=0}^n p_i,$ i.e., it is simply the product of all the probabilities of single-step reductions within the path.
The \emph{set of terminal reduction paths of $M$} is
\[ \mathrm{\TPaths(M)} \defeq \bigcup_{V \in \Val(A)} \Paths(M,V) . \]
Thus the endpoint of any path $\pi \in \TPaths(M)$ is a value. If $\pi \in \Paths(M,W)$, where $W$ is a value, then we shall write $V_\pi \eqdef W.$ That is, for a path $\pi \in \TPaths(M)$, the notation $V_\pi$ indicates the endpoint of the path $\pi$ which is indeed a value.
\end{definition}

\begin{remark}
We also note that for each closed term $M$, the set $\TPaths(M)$ is countable.
\end{remark}

The next definition we introduce is crucial for the proof of strong adequacy.

\begin{definition}\label{def:SM}
Given a relation $\tleqd \in \ValReld$ and a term $\cdot \vdash M : A$, let $\mathcal S(\tleqd; M)$ be the Scott-closure in $\KL(1,X)$ of the set
\begin{equation}
\label{eq:logical-sums}
\mathcal S_{0}(\tleqd; M) \defeq \left\{ \sum_{\pi \in F} P(\pi) v_\pi \ |\ F \subseteq \TPaths(M),\ F \text{ is finite and }  \text{$v_\pi \tleqd V_\pi$ for each $\pi \in F$} \right\} .
\end{equation}
In other words, $\mathcal S(\tleqd; M)$ is the smallest Scott-closed subset of $\KL(1,X)$ which contains all morphisms of the form in \eqref{eq:logical-sums}.
For a subset $U \subseteq \KL(1,X),$ we write $\overline U$ to indicate its Scott-closure in $\KL(1,X)$.
\end{definition}

\begin{lemma}
\label{lem:semantically-dense-value}
For any value $V$, we have $\mathcal S(\tleqd; V) = \ol{\{ v\ |\ v \tleqd V\}} \cup \{ 0\} = \ol{\{ v\ |\ v \tleqd V\} \cup \{ 0\}} .$
\end{lemma}
\begin{proof}
This is because all of the sums in \eqref{eq:logical-sums} are singleton sums or the empty sum.
\end{proof}

\begin{lemma}[{\cite[Lemma 8.4]{jones90}}]
\label{lem:topological-goodness}
Let $Y$ be a dcpo and let $\{X_i\}_{i \in F}$ be a finite collection of dcpo's.
Let $f \colon \prod_i X_i \to Y$ be a Scott-continuous function. Let $C_Y$ be a Scott-closed subset of $Y$. Let $U_i \subseteq X_i$ be arbitrary subsets, such that $f(\prod_i U_i) \subseteq C_Y$.
Then $f(\prod_i \overline{U_i}) \subseteq C_Y$, where $\overline{U_i}$ is the Scott-closure of $U_i$ in $X_i$.
\end{lemma}

\begin{lemma}
\label{lem:logical-composition}
Let $\tleq_{X_1, A}^{e_1}$ and $\tleq_{X_2,A}^{e_2}$ be two logical relations and $\cdot \vdash M : A$ a term.
Assume that $g: X_1 \kto X_2$ is a morphism, such that $v \tleq_{X_1, A}^{e_1} V$ implies $g \kcirc v \in \mathcal S(\tleq_{X_2, A}^{e_2} ; V),$ for any $V \in \Val(M).$
If $m \in \mathcal S(\tleq_{X_1, A}^{e_1} ; M)$, then $g \kcirc m \in \mathcal S(\tleq_{X_2, A}^{e_2} ; M).$
\end{lemma}
\begin{proof}
By Lemma \ref{lem:topological-goodness}, it suffices to show that
\[ \left( g \kcirc \sum_{\pi \in F} P(\pi) v_\pi \right) \in \mathcal S(\tleq_{X_2, A}^{e_2} ; M)\]
for any choice of finite $F \subseteq \TPaths(M)$ and morphisms $v_\pi$ with $v_\pi \tleq_{X_1,A}^{e_1} V_\pi.$
We have   
\[ g \kcirc \sum_{\pi \in F} P(\pi) v_\pi = \sum_{\pi \in F} P(\pi) (g \kcirc v_\pi), \]
where the equality follows by linearity of $(g \kcirc -)$. Next, for each $v_\pi$, by assumption $g \kcirc v_\pi \in \mathcal S(\tleq_{X_2, A}^{e_2} ; V_\pi).$
Therefore by applying Lemma \ref{lem:semantically-dense-value}, it follows $g \kcirc v_\pi \in \overline{\{ v'\ |\ v' \tleq_{X_2,A}^{e_2} V_\pi \} \cup \{  0 \}}.$
Now, consider the function 
\[ \sum_{\pi \in F} P(\pi) (-) \colon \prod_{|F|} \KL(1,X_2) \to \KL(1,X_2) . \]
This function is continuous, so by Lemma \ref{lem:topological-goodness} again, it suffices to show that 
\[ \sum_{\pi \in F} P(\pi) m'_\pi  = \sum_{\substack{ \pi \in F \\ m'_\pi \neq 0 } } P(\pi) m'_\pi \in \mathcal S(\tleq_{X_2, A}^{e_2} ; M), \]
where either $m'_\pi =  0$ or $m'_\pi \tleq_{X_2,A}^{e_2} V_\pi$ for each $\pi \in F$. Since the summands where $m'_\pi =  0$ do not affect the sum, it suffices to show that this is true under the assumption that
$m'_\pi \tleq_{X_2,A}^{e_2} V_\pi$.
But this is true by definition of $\mathcal S(\tleq_{X_2, A}^{e_2}; M)$.
\end{proof}

Next, we define important \emph{closure relations} which we use for terms.

\begin{definition}
\label{def:logical-closure}
If $\tleqd \in \ValRel(X, A, e)$, let ${\qtleqd} \subseteq \KL(1, X) \times \Prog(A)$ be the relation defined by
  \[ m \qtleqd M \text{ iff } m \in \mathcal S(\tleqd; M) . \]
\end{definition} 

\begin{lemma}
For any term  $\cdot \vdash M : A$ and $\tleqd \in \ValReld$, the set $(-) \qtleqd M$ is a Scott-closed subset of $\KL(1,X).$
\end{lemma}
\begin{proof}
This follows immediately by definition, because $\mathcal S(\tleqd; M)$ is Scott-closed.
\end{proof}

\begin{lemma}
\label{lem:scott-closure-embedding}
Let $C$ be a Scott-closed subset of a dcpo $X$. Let $W \eqdef \{ \delta_x\ |\ x \in C \} \subseteq \MM X$ and let $\overline W$ be the Scott-closure of $W$ in $\MM X.$ Then, $\delta_y \in \overline W$ iff $y \in C.$
\end{lemma}
\begin{proof}
The ``if'' direction is straightforward. The ``only if'' direction is trivial when $C = X$. We now prove the case that $C$ is a proper subset of $X$, and let $U$ be the complement of $C$. Hence $U$ is a nonempty Scott open subset of $X$. Let us assume that $\delta_y \in \overline W$ but $y\in U$, then we know that $[U>0]\defeq \{ \nu \in \MM X \ | \ \nu(U)>0 \} $ is a Scott open subset of $\MM X$ containing $\delta_{y}$, hence we would have that $[U>0] \cap W \not= \emptyset$ since by assumption $\delta_{y} \in \overline W$. However, this is impossible since for any $x\in C$, $\delta_{x}(U) = 0$.
\end{proof}

\begin{lemma}
\label{lem:id}
Let $X$ be a dcpo, let $v \in \TD(1,X)$ and let $V$ be a value. Then $v \tleqd V $ iff $v \qtleqd V.$
\end{lemma}
\begin{proof}
The left-to-right direction follows immediately by Lemma \ref{lem:semantically-dense-value}. For the other direction, we first observe that since $v \in \TD(1,X)$, then $v \neq 0$.
Therefore by Lemma \ref{lem:semantically-dense-value}, it follows $v \in \ol{\{ w\ |\ w \tleqd V\}} $ and then by Lemma \ref{lem:scott-closure-embedding} we complete the proof.
\end{proof}

\begin{lemma}
\label{lem:logical-inequality}
For any value $\cdot \vdash V : A$ and $\tleqd \in \ValReld$, if $m \qtleqd V$ then $ e \kcirc m \leq \lrb V$.
\end{lemma}
\begin{proof}
We know $m \in \mathcal S(\tleqd; V) = \ol{\{ v\ |\ v \tleqd V\}} \cup \{ 0 \} $ and clearly $e \kcirc m \leq \lrb V$ is equivalent to $(e \kcirc m ) \in \down \lrb V,$ which is a Scott-closed subset.
If $m = 0$, then the statement is obviously true. So, assume that $m \in \ol{\{ v\ |\ v \tleqd V\}}. $
Composition with $e$ is a Scott-continuous function and therefore using Lemma \ref{lem:topological-goodness}, to finish the proof it suffices to show $e \kcirc v \leq \lrb V$ for each choice
of $v \tleqd V$. But this is true by assumption on $\tleqd$.
\end{proof}

\subsection{Categories of Logical Relations}

\begin{definition}
\label{def:logical-categories}
For any type $A$, we define a category $\RR(A)$ where:
\begin{itemize}
\item Each object is a triple $(X, e_X, \tleq_X)$, where $X$ is a dcpo, $e_X \colon X \kto \lrb A$ is a morphism in $\PD_e$ and $\tleq_X \in \ValRel(X,A,e_X)$.
\item A morphism $f: (X, e_X, \tleq_{X}) \to (Y, e_Y, \tleq_{Y})$ is a morphism $f : X \kto Y$ in $\PD_e$, which satisfies the three additional conditions:
  \begin{itemize}
    \item If $v \tleq_{X} V,$ then $f \kcirc v\ \overline{\tleq_{Y}}\ V.$
    \item If $v \tleq_{Y} V,$ then $f^p \kcirc v\ \overline{\tleq_{X}}\ V.$
    \item $e_X = e_Y \kcirc f.$
  \end{itemize}
\item Composition and identities coincide with those in $\PD_e.$
\end{itemize}
\end{definition}

\begin{lemma}
For every type $A$, the category $\RR(A)$ is indeed well-defined.
\end{lemma}
\begin{proof}
We have to show that $\kid: (X, e_X, \tleq_X) \to (X, e_X, \tleq_X)$ is indeed a morphism in $\RR(A).$ This follows from Lemma \ref{lem:id}.
Next, we have to show that if $f \colon (X, e_X, \tleq_X) \to (Y, e_Y, \tleq_{Y})$ and $g \colon (Y, e_Y, \tleq_{Y}) \to (Z, e_Z, \tleq_{Z}),$ then we also have
$g \kcirc f \colon (X, e_X, \tleq_X) \to (Z, e_Z, \tleq_{Z}).$ But this follows by Lemma \ref{lem:logical-composition}.
\end{proof}

\begin{lemma}
\label{lem:logical-composition-fancy}
Let $\cdot \vdash M : A$ be a term and let $g \colon (X, e_X, \tleq_X) \to (Y, e_Y, \tleq_Y)$ be a morphism in $\RR(A)$. If $m \ol{\tleq_X} M$ then $g \kcirc m \ol{\tleq_Y} M$.
Moreover, if $n \ol{\tleq_Y} N,$ then $g^p \kcirc n \ol{\tleq_X}$ N.
\end{lemma}
\begin{proof}
This follows immediately by Lemma \ref{lem:logical-composition}.
\end{proof}

\begin{definition}
\label{def:logical-forgetful}
For every type $A$, we define the obvious forgetful functor $U^A \colon \RR(A) \to \PDe$ by
\begin{align*}
U^A(X,e, \tleq) &= X \\
U^A(f) &= f.
\end{align*}
\end{definition}

\begin{proposition}
\label{prop:logical-colimits}
For each type $A$, the category $\RR(A)$ has an initial object and all $\omega$-colimits. Furthermore, the forgetful functor $U^A \colon \RR(A) \to \PDe$ preserves and reflects $\omega$-colimits (and also the initial objects).
\end{proposition}
\begin{proof}
We begin with the initial object.

\paragraph*{\textbf{Initial object}}
For any dcpo's $X$ and $Y$, we write $0_{X,Y} : X \kto Y$ for the zero morphism in $\PD$. Notice that $0_{\varnothing,X}$ is an embedding with projection counterpart given by $0_{X,\varnothing}.$

The object $(\varnothing, 0_{\varnothing, \lrb{A}}, \emptyset)$ is initial in $\RR(A).$ 
Indeed, let $(X, e_X, \tleq_X)$ be any other object of $\RR(A).$ It suffices to show that $0_{\varnothing, X}: (\varnothing, 0_{\varnothing, \lrb A}, \emptyset) \to (X, e_X, \tleq_X)$ is a morphism in $\RR(A)$, because if it exists, then it is clearly unique.
The first and third conditions of Definition \ref{def:logical-categories} are trivially satisfied. The second condition is also satisfied, because $0_{\varnothing,X}^p \kcirc v = 0_{1, \varnothing}$, which is the least (and only) element in $\KL(1, \varnothing)$
and this element is contained in every relation $\overline{\tleq_Y}$, including $\overline{\emptyset}$.

\paragraph*{\textbf{The diagram}}
For the rest of the proof, let $D: \omega \to \RR(A)$ be an $\omega$-diagram in $\RR(A)$. Let $D(i) = (X_i, e_i, \tleq_i)$ and let $D(i \leq j) = f_{i,j}$.

\paragraph*{\textbf{Construction of the colimiting object}}
Consider the $\omega$-diagram $UD$ in $\PDe$.
This category has all $\omega$-colimits, so let $\tau \colon UD \naturalto X_\omega$ be its colimiting cocone. Next, consider the cocone $\epsilon \colon UD \naturalto \lrb A$ defined by $\epsilon_i \eqdef e_i \colon X_i \kto \lrb A.$
Let $e_\omega : X_\omega \kto \lrb A$ be the unique cocone morphism $e_\omega: \tau \to \epsilon$ induced by the colimit $\tau$ in $\PDe$. We now define a relation
\begin{align*}
&\tleq_\omega \in \ValRel(X_\omega, A, e_\omega) \text{ by:} \\
v &\tleq_\omega V \text{ iff } \forall k \in \mathbb N. \ \tau_k^p \kcirc v \ol{\tleq_k} V.
\end{align*}
We have to show that $\tleq_\omega \in \ValRel(X_\omega, A, e_\omega),$ as claimed above. We begin with downwards-closure. Assume $v \tleq_\omega V$ and that $v' \leq v$ in $\TD(1,X_\omega) $. Then,
$\forall k \in \mathbb N.\ \tau_k^p \kcirc v \ol{\tleq_k} V$ and therefore $\tau_k^p \kcirc v' \ol{\tleq_k} V$, because $(-\ol{\tleq_k} V)$ is downwards-closed and so by definition $v' \tleq_\omega V,$ as required.

Next, we show that $(- \tleq_\omega V)$ preserves directed suprema and is therefore Scott-closed in $\TD(1,X_\omega)$. Assume that $\{v_d\}_{d \in D}$ is a directed set, such that $v_d \tleq_\omega V$ for each $d \in D$.
Therefore, $\forall k \in \mathbb N.\  \forall d \in D.\ \tau_k^p \kcirc v_d \ol{\tleq_k} V.$ Scott-closure of $(- \ol{\tleq_k} V)$ implies that $\tau_k^p \kcirc (\sup_{d \in D} v_d) = \sup_{d \in D} \tau_k^p \kcirc  v_d \ol{\tleq_k} V$ holds for all $k \in \mathbb N$. Therefore,
by definition $\sup_{d \in D} v_d \tleq_\omega V.$

We also have to show that if $v \tleq_\omega V$, then $e_\omega \kcirc v \leq \lrb V.$ If $v \tleq_\omega V$, then $\forall k \in \mathbb N.\ \tau_k^p \kcirc v \ol{\tleq_k} V$ and so by Lemma \ref{lem:logical-inequality} we get
$ e_k \kcirc \tau_k^p \kcirc v \leq \lrb V$. But $e_k \kcirc \tau_k^p \kcirc v = e_\omega \kcirc \tau_k \kcirc \tau_k^p \kcirc v.$ The limit-colimit coincidence theorem in the category $\PD$, shows that this forms an increasing sequence and that
\[ \lrb V \geq \sup_{k \in \mathbb N} e_\omega \kcirc \tau_k \kcirc \tau_k^p \kcirc v = e_\omega \kcirc \left( \sup_{k \in \mathbb N} \tau_k \kcirc \tau_k^p \right) \kcirc v = e_\omega \kcirc \kid \kcirc v = e_\omega \kcirc v , \]
as required. We will show that the object $(X_\omega, e_\omega, \tleq_\omega)$ is the colimiting object of $D$ in $\RR(A)$. Before we can do this, we first have to construct the colimiting cocone in $\RR(A).$

\paragraph*{\textbf{Construction of the colimiting cocone}}
We show that $\tau : D \naturalto X_\omega$ is a cocone in $\RR(A)$. The commutativity requirements are clearly satisfied, so it suffices to show 
that each $\tau_i \colon X_i \kto X_\omega$ is a morphism $\tau_i \colon (X_i, e_i, \tleq_i) \to (X_\omega, e_\omega, \tleq_\omega)$ in $\RR(A)$.
Towards that end, assume that $v \tleq_i V$. We have to show that $\tau_i \kcirc v \ol{\tleq_\omega} V$, but by Lemma \ref{lem:id}, it
suffices to show that $\tau_i \kcirc v \tleq_\omega V$\footnote{Note that $\tau_i \kcirc v$ is a morphism of $\TD$, because $v$ is one and because $\tau_i \in \PD_e$ which is a subcategory of $\TD.$}.
Showing this is equivalent to showing that $\forall k \in \mathbb N.\ \tau_k^p \kcirc \tau_i \kcirc v \ol{\tleq_k} V.$
For any $k \geq i$, we get:
\[ \tau_k^p \kcirc \tau_i \kcirc v = \tau_k^p \kcirc \tau_k \kcirc f_{i,k} \kcirc v = f_{i,k} \kcirc v \ol{\tleq_k} V \]
because $f_{i,k}$ is a morphism $f_{i,k} \colon (X_i, e_i, \tleq_i) \to (X_k, e_k, \tleq_k)$ and $v \tleq_i V$ by assumption.
For any $k < i$, we get: 
\[ \tau_k^p \kcirc \tau_i \kcirc v = f_{k,i}^p \kcirc \tau_i^p \kcirc \tau_{i} \kcirc v = f_{k,i}^p \kcirc v \ol{\tleq_k} V \]
because $f_{k,i}$ is a morphism $f_{k,i} \colon (X_k, e_k, \tleq_k) \to (X_i, e_i, \tleq_i)$ and $v \tleq_i V$ by assumption (and Lemma \ref{lem:logical-composition-fancy}).

To show that $\tau_i \colon (X_i, e_i, \tleq_i) \to (X_\omega, e_\omega, \tleq_\omega)$ is a morphism, we have to show that if $v \tleq_\omega V,$ then also $\tau_i^p \kcirc v \ol{\tleq_i} V .$ But this is true by definition of $\tleq_\omega.$

Finally we have to show that $e_i = e_\omega \kcirc \tau_i.$ But this is true by construction of $e_\omega$.

Therefore, $\tau : D \naturalto (X_\omega, e_\omega, \tleq_\omega)$ is indeed a cocone of $D$ in $\RR(A).$

\paragraph*{\textbf{Couniversality of the cocone}}
For the rest of the proof, assume that $\alpha : D \naturalto (Y, e_y, \tleq_Y)$ is some other cocone of $D$ in $\RR(A).$ Next, consider the cocone $U\alpha$ in $\PDe$ and let $a: X_\omega \kto Y$ be the unique cocone morphism $a : U\tau \to U\alpha$ induced by the colimit in $\PDe.$
By the limit-colimit coincidence theorem in $\PD$, we get
\[ a = a \kcirc \kid = a \kcirc \sup_{i \in \mathbb N} \tau_i \kcirc \tau_i^p = \sup_{i \in \mathbb N} a \kcirc \tau_i \kcirc \tau_i^p = \sup_{i \in \mathbb N} \alpha_i \kcirc \tau_i^p \]
We will show that $a: (X_\omega, e_\omega, \tleq_\omega) \to (Y, e_Y, \tleq_Y)$ is a morphism in $\RR(A).$ Towards this end, assume that $v \tleq_\omega V$. Then $\forall k \in \mathbb N.\ \tau_k^p \kcirc v \ol{\tleq_k} V$ and therefore $\alpha_k \kcirc \tau_k^p \kcirc v \ol{\tleq_Y} V, $ because by assumption
$\alpha_k : (X_k, e_k, \tleq_k) \to (Y, e_y, \tleq_Y).$ Since $(- \ol{\tleq_Y} V)$ is closed under suprema, it follows
\[ \sup_{k \in \mathbb N} \alpha_k \kcirc \tau_k^p \kcirc v =  \left( \sup_{k \in \mathbb N} \alpha_k \kcirc \tau_k^p \right) \kcirc v = a \kcirc v \ol{\tleq_Y} V, \]
which shows that $a$ satisfies one of the requirements for being a morphism in $\RR(A).$

For the second requirement, assume that $v \tleq_Y V.$ Then, $\forall k \in \mathbb N.\ \alpha_k^p \kcirc v \ol{\tleq_k} V$, by assumption on $\alpha_k$. The same argument shows that $\forall k \in \mathbb N.\ \tau_k \kcirc \alpha_k^p \kcirc v \ol{\tleq_\omega} V$,
because $\tau_k$ is also a morphism in the category. Since $(- \ol{\tleq_\omega} V)$ is closed under suprema, we get:
\[ \sup_{k \in \mathbb N} \tau_k \kcirc \alpha_k^p \kcirc v = \sup_{k \in \mathbb N} \tau_k \kcirc \tau_k^p \kcirc a^p \kcirc v = \left( \sup_{k \in \mathbb N} \tau_k \kcirc \tau_k^p \right) \kcirc a^p \kcirc v =  a^p \kcirc v  \ol{\tleq_\omega} V \]
as required.

For the third requirement, we have to show that $e_\omega = e_Y \kcirc a.$ By assumption on the cone $\alpha : D \naturalto (Y, e_Y, \tleq_Y)$, we have that $\forall i \in \mathbb N.\ e_i = e_Y \kcirc \alpha_i$
and by construction of $a$, we know $\alpha_i = a \kcirc \tau_i.$ Therefore $\forall i\in \mathbb N.\ e_i = e_Y \kcirc a \kcirc \tau_i.$ However, $e_\omega$ is by construction the unique morphism in $\PD_e$, such that $\forall i.\ e_i = e_\omega \kcirc \tau_i,$ which shows that
$e_\omega = e_Y \kcirc a,$ as required. Therefore, we have shown that $a : (X_\omega, e_\omega, \tleq_\omega) \to (Y, e_Y, \tleq_Y) $ is indeed a morphism in $\RR(A).$

That $a : \tau \to \alpha$ is the unique cocone morphism is now obvious, because if $a': \tau \to \alpha$ is another one, then $Ua$ and $Ua'$ are both
cocone morphisms between $U\tau$ and $U\alpha$ in $\PDe$ and therefore $a = Ua = Ua' = a'.$
Therefore,  $\tau : D \naturalto (X_\omega, e_\omega, \tleq_\omega)$ is indeed the colimiting cocone of $D$ in $\RR(A)$, which shows that $\RR(A)$ has all $\omega$-colimits.

\paragraph*{\textbf{$U^A$ preserves $\omega$-colimits}}
Assume that the cocone $\alpha : D \naturalto (Y, e_y, \tleq_Y)$ from above is colimiting in $\RR(A).$ But, we know that $\tau : D \naturalto (X_\omega, e_\omega, \tleq_\omega)$ is also a colimiting cocone of $D$. Therefore, there exists a unique cocone isomorphism $i: \tau \to \alpha.$ Then, $Ui : U\tau \to U\alpha$ is
a cocone isomorphism in $\PDe$. However, by construction, $U\tau$ is a colimiting cocone of $UD$ in $\PDe$ and therefore so is $U\alpha$. 

\paragraph*{\textbf{$U^A$ reflects $\omega$-colimits}}
Assume that the cocone $\alpha : D \naturalto (Y, e_y, \tleq_Y)$ from above is such that $U\alpha : UD \naturalto Y$ is colimiting in $\PDe.$ Then the morphism $a: X_\omega \kto Y$ from above is an isomorphism in $\PDe$.
We have already shown that $a : (X_\omega, e_\omega, \tleq_\omega) \to (Y, e_Y, \tleq_Y)$ is a morphism in $\RR(A)$. Thus, to finish the proof, it suffices to show that $a^{-1}$ is a morphism in $\RR(A)$ in the opposite direction. But this is obviously true,
because $a^{-1} = a^p$ and $(a^{-1})^p = a$ and we have shown above that these morphisms satisfy the logical requirements and clearly $e_Y = e_\omega \kcirc a^{-1}.$

\end{proof}

Next, we introduce important relation constructors and some new notation.

\begin{notation}
Given morphisms $m_i : 1 \kto X_i$, for $i \in \{1, \ldots, n\},$ we define
\[ \llangle m_1, \ldots, m_n \rrangle \eqdef (m_1 \ktimes \cdots \ktimes m_n) \kcirc \JJ \langle \id_1, \ldots, \id_1 \rangle : 1 \kto X_1 \times \cdots \times X_n . \]
\end{notation}
\begin{notation}
\label{not:application}
Given morphisms $x : 1 \kto X$ and $f: 1 \kto [X \kto Y]$ in $\KL$, let $f[x] : 1 \kto Y$ be the morphism defined by
\[ f[x] \eqdef \epsilon \kcirc (f \ktimes x) \kcirc \JJ\langle \id_1, \id_1 \rangle . \]
\end{notation}

\begin{definition}[Relation Constructions]
\label{def:logical-constructions}
We define relation constructors:
\begin{itemize}
\item If $\tleqone \in \ValRel(X_1, A_1, e_1)$ and $\tleqtwo \in \ValRel(X_2, A_2, e_2)$, define
  \begin{align*}
  (\tleqone &+ \tleqtwo) \in \ValRel(X_1+ X_2, A_1 + A_2, e_1 \kplus e_2) \text{ by: } \\
  \JJ \emph{in}_i \kcirc v\ (\tleqone &+ \tleqtwo)\ \mathtt{in}_i V \text{ iff }  v \tleq_{X_i, A_i}^{e_i} V   \qquad\qquad (\text{for $i \in \{1,2\}$}). 
  \end{align*}
\item If $\tleqone \in \ValRel(X_1, A_1, e_1)$ and $\tleqtwo \in \ValRel(X_2, A_2, e_2)$, define
  \begin{align*}
  (\tleqone &\times \tleqtwo) \in \ValRel(X_1 \times X_2, A_1 \times A_2, e_1 \ktimes e_2) \text{ by: } \\
  \llangle v_1, v_2 \rrangle \ (\tleqone &\times \tleqtwo)\ (V_1,V_2) \text{ iff }  v_1 \tleqone V_1 \text{ and } v_2 \tleqtwo V_2.
  \end{align*}
\item If $\tleqone \in \ValRel(X_1, A_1, e_1)$ and $\tleqtwo \in \ValRel(X_2, A_2, e_2)$, define
  \begin{align*}
  (\tleqone &\to\ \tleqtwo) \in \ValRel([X_1 \kto X_2], A_1 \to A_2, \JJ [e_1^p \kto e_2]) \text{ by: } \\
  f\ (\tleqone &\to\ \tleqtwo)\ \lambda x. M \text{ iff }  \JJ [e_1^p \kto e_2] \kcirc f \leq \lrb{\lambda x.M} \text{ and } \forall (v \tleqone V).\  f[v] \qtleqtwo (\lambda x.M)V. 
  \end{align*}
\end{itemize}
\end{definition}

\begin{lemma}
The assignments in Definition \ref{def:logical-constructions} are indeed well-defined.
\end{lemma}
\begin{proof}
Straightforward verification.
\end{proof}

Next, a simple lemma that we use later.

\begin{lemma}
\label{lem:parameterised-adjunction}
Assume we are given morphisms $f : 1 \kto [C \kto D], h : A \kto C, g : D \kto B$ and $v: 1 \kto A.$ Then
\[ ( \JJ[h \kto g] \kcirc f)[v] = g \kcirc f[h \kcirc v] . \]
\end{lemma}
\begin{proof}
\begin{align*}
( \JJ[h \kto g] \kcirc f)[v] &= \epsilon \kcirc ( (\JJ[h \kto g] \kcirc f) \ktimes v) \kcirc \JJ \langle \id, \id \rangle & (\text{Definition}) \\
&= \epsilon \kcirc (\JJ[h \kto g] \ktimes \kid) \kcirc ( f \ktimes v) \kcirc \JJ \langle \id, \id \rangle &  \\
&= \epsilon \kcirc (\JJ[\kid \kto g] \ktimes \kid) \kcirc (\JJ[h \kto \kid] \ktimes \kid) \kcirc ( f \ktimes v) \kcirc \JJ \langle \id, \id \rangle &  \\
&= g \kcirc \epsilon \kcirc (\JJ[h \kto \kid] \ktimes \kid) \kcirc ( f \ktimes v) \kcirc \JJ \langle \id, \id \rangle & (\text{Naturality of } \epsilon) \\
&= g \kcirc \epsilon \kcirc (\kid \ktimes h) \kcirc ( f \ktimes v) \kcirc \JJ \langle \id, \id \rangle & (\text{Parameterised adjunction \cite[pp.102]{maclane}}) \\
&= g \kcirc f[h \kcirc v]& (\text{Definition})
\end{align*}
\end{proof}

\begin{notation}
Throughout the rest of the paper we shall write $(- \ktwo_e -) \eqdef [- \kto -]_e^\JJ : \PDe \times \PDe \to \PDe$. That is, we just introduce a more concise notation for the functor $[- \kto - ]_e^\JJ$ from Proposition \ref{app:omega-functors}.
\end{notation}

The next definition is crucial. Given two logical relations, it is used to define the product, coproduct and function space logical relations. Moreover, this is done in a functorial sense on the categories $\RR(A)$.

\begin{definition}
\label{def:logical-functors}
Let $A$ and $B$ be types. We define \emph{covariant} functors in the following way (recall Definition \ref{def:logical-constructions}):
\begin{enumerate}
\item $\times^{A,B} \colon \RR(A) \times \RR(B) \to \RR(A \times B)$ by
  \begin{align*}
  (X, e_X, \tleq_X) \times^{A,B} (Y, e_Y, \tleq_Y) &\eqdef (X \times Y, e_X \ktimes_e e_Y, \tleq_X \times \tleq_Y) \\
  f \times^{A,B} g &\eqdef f \ktimes_e g
  \end{align*}
\item $+^{A,B} \colon \RR(A) \times \RR(B) \to \RR(A + B)$ by
  \begin{align*}
  (X, e_X, \tleq_X) +^{A,B} (Y, e_Y, \tleq_Y) &\eqdef (X + Y, e_X \kplus_e e_Y, \tleq_X + \tleq_Y) \\
  f +^{A,B} g &\eqdef f \kplus_e g
  \end{align*}
\item $\to^{A,B} \colon \RR(A) \times \RR(B) \to \RR(A \to B)$ by
  \begin{align*}
  (X, e_X, \tleq_X) \to^{A,B} (Y, e_Y, \tleq_Y) &\eqdef ([X \kto Y], e_X \ktwo_e e_Y, \tleq_X \to \tleq_Y) \\
  f \to^{A,B} g &\eqdef f \ktwo_e g
  \end{align*}
\end{enumerate}
\end{definition}

\begin{proposition}
Each of the functors from Definition \ref{def:logical-functors} is well-defined.
\end{proposition}
\begin{proof}
We will show the case for function types which is the most complicated. The other cases follow by a straightforward verification using similar arguments.

\paragraph*{\textbf{Function types}}
Let
\begin{align*}
f_1 &: (X_1, e_1^X, \tleq_1^X) \to (Y_1, e_1^Y, \tleq_1^Y) \\
f_2 &: (X_2, e_2^X, \tleq_2^X) \to (Y_2, e_2^Y, \tleq_2^Y)
\end{align*}
We have to show
\[ f_1 \ktwo_e f_2 : (X_1 \ktwo_e X_2, e_1^X \ktwo_e e_2^X, \tleq_1^X \to \tleq_2^X) \to (Y_1 \ktwo_e Y_2, e_1^Y \ktwo_e e_2^Y, \tleq_1^Y \to \tleq_2^Y)  \]
is a morphism in $\RR(A \to B)$.

First, we show that $f_1 \ktwo_e f_2$ respects the embedding component. Indeed:
\begin{align*}
e_1^X \ktwo_e e_2^X = (e_1^Y \kcirc f_1) \ktwo_e (e_2^Y \kcirc f_2) = (e_1^Y \ktwo_e e_2^Y) \kcirc (f_1 \ktwo_e f_2).
\end{align*}
Next, assume that $v\ (\tleq_1^X \to \tleq_2^X)\ V$ . Assume further that $v' \tleq_1^Y V'$. Then, clearly
$f_1^p \kcirc v' \ol{\tleq_1^X} V'$. If $f_1^p \kcirc v' = 0$, then it trivially follows that $v[f_1^p \kcirc v'] = 0 \ol{\tleq_2^X} VV'.$
Otherwise, $f_1^p \kcirc v' \in \TD$ and so $f_1^p \kcirc v' \tleq_1^X V'$ and therefore $v[f_1^p \kcirc v'] \ol{\tleq_2^X} VV'.$
In all cases, $v[f_1^p \kcirc v'] \ol{\tleq_2^X} VV'$ and therefore $f_2 \kcirc v[f_1^p \kcirc v'] \ol{\tleq_2^Y} VV'$.
But then, by Lemma \ref{lem:parameterised-adjunction} we have:
\[ f_2 \kcirc v[f_1^p \kcirc v'] = (\JJ[f_1^p \kto f_2] \kcirc v)[v'] = ((f_1 \ktwo_e f_2) \kcirc v)[v'] \ol{\tleq_2^Y} VV' . \]
Furthemore 
\begin{align*}
(e_1^Y \ktwo_e e_2^Y) \kcirc (f_1 \ktwo_e f_2) \kcirc v = (e_1^X \ktwo_e e_2^X) \kcirc v \leq \lrb V
\end{align*}
and therefore by definition $(f_1 \ktwo_e f_2) \kcirc v\ (\tleq_1^Y \to \tleq_2^Y)\ V$ and therefore also $(f_1 \ktwo_e f_2) \kcirc v \ol{(\tleq_1^Y \to \tleq_2^Y)} V$, as required.

For the other direction, assume that $v\ (\tleq_1^Y \to \tleq_2^Y)\ V$ . Assume further that $v' \tleq_1^X V'$. Then, clearly
$f_1 \kcirc v' \ol{\tleq_1^Y} V'$. If $f_1 \kcirc v' = 0$, then it trivially follows that $v[f_1 \kcirc v'] = 0 \ol{\tleq_2^Y} VV'.$
Otherwise, $f_1 \kcirc v' \in \TD$ and so $f_1 \kcirc v' \tleq_1^Y V'$ and therefore $v[f_1 \kcirc v'] \ol{\tleq_2^Y} VV'.$
In all cases, $v[f_1 \kcirc v'] \ol{\tleq_2^Y} VV'$ and therefore $f_2^p \kcirc v[f_1 \kcirc v'] \ol{\tleq_2^X} VV'$.
But then, by Lemma \ref{lem:parameterised-adjunction} we have:
\[ f_2^p \kcirc v[f_1 \kcirc v'] = (\JJ[f_1 \kto f_2^p] \kcirc v)[v'] = ((f_1 \ktwo_e f_2)^p \kcirc v)[v'] \ol{\tleq_2^X} VV' . \]
Furthemore
\begin{align*}
(e_1^X \ktwo_e e_2^X) \kcirc (f_1 \ktwo_e f_2)^p \kcirc v &= \JJ[(e_1^X)^p \kto e_2^X] \kcirc \JJ[ f_1 \kto f_2^p ] \kcirc v \\
&= \JJ[ (f_1 \kcirc (e_1^X)^p) \kto (e_2^X \kcirc f_2^p) ] \kcirc v \\
&\leq \JJ[ (f_1 \kcirc (e_1^X)^p) \kto e_2^Y ] \kcirc v \\
&\leq \JJ[ (e_1^Y)^p \kto e_2^Y ] \kcirc v \\
&\leq \lrb V .
\end{align*}
If $(f_1 \ktwo_e f_2)^p \kcirc v \in \TD$, then $(f_1 \ktwo_e f_2)^p \kcirc v\ (\tleq_1^X \to \tleq_2^X)\ V$ by definition.
Otherwise, $(f_1 \ktwo_e f_2)^p \kcirc v = 0 $ and then trivially $(f_1 \ktwo_e f_2)^p \kcirc v = 0 \ol{(\tleq_1^X \to \tleq_2^X)} V$.
Therefore, in all cases $(f_1 \ktwo_e f_2)^p \kcirc v \ol{(\tleq_1^X \to \tleq_2^X)} V$, as required.

Therefore, the functor $\to^{A,B}$ is indeed well-defined.
\end{proof}

Observe that Definition \ref{def:logical-functors} lifts the functors that we use to interpret our types in the category $\KL$ to the categories $\RR(A)$. Next, we show that the functors we just defined are also suitable for forming (parameterised) initial algebras.

\begin{proposition}
\label{prop:logical-functors}
For $\star \in \{\times, +, \to\},$ for all types $A$ and $B$, the functor $\star^{A,B} : \RR(A) \times \RR(B) \to \RR(A \star B)$ is $\omega$-cocontinuous and the following diagram:
\[ \stikz{logical-functors.tikz} \]
commutes.
\end{proposition}
\begin{proof}
Commutativity of the diagram is immediate from the definitions. To see $\omega$-cocontinuity, let $D$ be an $\omega$-diagram in $\RR(A) \times \RR(B)$ and let $\tau$ be its colimiting cocone.
Because the functors $U^A, U^B$ and $\kstar_e$ are $\omega$-cocontinuous, it follows that :
\begin{align*}
& \qquad (\kstar_e \circ U^A \times U^B)\tau \text{ is colimiting in } \PDe & & \\
& \Longrightarrow (U^{A \star B} \circ \star^{A,B})\tau \text{ is colimiting in } \PDe & &\text{(Commutativity of the above diagram)} \\
& \Longrightarrow \star^{A,B}\tau \text{ is colimiting in } \RR(A \star B) & &\text{($U$ reflects $\omega$-colimits)}
\end{align*}
which shows that $\star^{A,B}$ is $\omega$-cocontinuous.
\end{proof}

Next, we establish an isomorphism between the categories $\RR(\mu X. A)$ and $\RR(A[\mu X.A / X]).$

\begin{definition}
\label{def:logical-fold-unfold}
We define constructors for folding and unfolding logical relations as follows:
\begin{itemize}
\item If $\tleq_{X, A[\mu Y. A / Y]}^e \in \ValRel(X, A[\mu Y.A / Y], e)$, define
  \begin{align*}
  (\mathbb I^{\mu Y. A}    &\tleq_{X, A[\mu Y. A / Y]}^e ) \in \ValRel(X, \mu Y. A, \sfold \kcirc e) \text{ by:} \\
  v\ (\mathbb I^{\mu Y. A} &\tleq_{X, A[\mu Y. A / Y]}^e )\ \mathtt{fold}\ V \text{ iff } v \tleq_{X, A[\mu Y. A / Y]}^e V .
  \end{align*}
\item If $\tleq_{X, \mu Y.A}^e \in \ValRel(X, \mu Y.A, e )$, define
  \begin{align*}
  (\mathbb E^{\mu Y. A}    &\tleq_{X, \mu Y.A}^e) \in \ValRel(X, A[\mu Y. A / Y], \sunfold \kcirc e) \text{ by:}\\
  v\ (\mathbb E^{\mu Y. A} &\tleq_{X, \mu Y.A}^e) )\ V \text{ iff } v \tleq_{X, \mu Y.A}^e \mathtt{fold}\ V .
  \end{align*}
\end{itemize}
\end{definition}

\begin{proposition}
The above assignments are indeed well-defined.
\end{proposition}
\begin{proof}
Straightforward verification.
\end{proof}

\begin{proposition}
\label{prop:logical-folding-isomorphism}
For every type $\cdot \vdash \mu X.A,$ we have an isomorphism of categories
\[\FOLD{\mu X.A} : \RR(A[\mu X.A / X]) \cong \RR(\mu X. A) : \UNFOLD{\mu X. A} , \]
where the functors are defined by 
\begin{align*}
&\FOLD{\mu X.A} :  \RR(A[\mu X.A / X]) \to \RR(\mu X. A)   & & \UNFOLD{\mu X.A} : \RR(\mu X. A) \to \RR(A[\mu X.A / X]) \\
&\FOLD{\mu X.A}(Y, e, \tleq) = (Y, \sfold \kcirc e , \FOLD{\mu X.A} \tleq) & & \UNFOLD{\mu X.A}(Y, e, \tleq) = (Y, \sunfold \kcirc e, \UNFOLD{\mu X.A} \tleq)\\
&\FOLD{\mu X.A}(f) = f & & \UNFOLD{\mu X.A}(f) = f
\end{align*}
\end{proposition}
\begin{proof}
The proof is essentially the same as \cite[Lemma 7.23]{lnl-fpc-lmcs}, with one extra proof obligation, namely we have to show that our functorial assignments respect the embedding components. But this is obviously true.
\end{proof}

This finishes the categorical development of the categories $\RR(A)$.

\subsection{Augmented Interpretation of Types}

We have now established sufficient categorical structure in order to construct parameterised initial algebras in the categories $\RR(A).$ Furthermore, we have sufficient structure to also define an \emph{augmented} interpretation of types in these categories.
The main idea behind providing the augmented interpretation is to show how to pick out the logical relations we need from all those that exist in the categories $\RR(A)$.

\begin{notation}
Given any type context $\Theta = X_1, \ldots , X_n$ and closed types $\cdot \vdash C_i$ with $i \in \{1, \ldots, n \}$,
we shall write $\vec C$ for $C_1, \ldots , C_n$ and we also write $[\vec C / \Theta]$ for $[C_1 / X_1, \ldots , C_n / X_n]$.
\end{notation}

\begin{definition}
For any type $\Theta \vdash A$ and closed types $\vec C$, we define their \emph{augmented interpretation} to be the functor
\[ \elrbc{\Theta \vdash A} : \RR(C_1) \times \cdots \times \RR(C_n) \to \RR(A[ \vec C / \Theta ]) \]
defined by induction on the derivation of $\Theta \vdash A$:
\begin{align*}
\elrbc{\Theta \vdash \Theta_i} &:= \Pi_i &&\\
\elrbc{\Theta \vdash A \star B } &:= \star^{A[\vec C / \Theta], B[\vec C / \Theta]} \circ \langle \elrbc{\Theta \vdash A}, \elrbc{\Theta \vdash B} \rangle &(\text{for }\star \in \{ +, \times, \to \})&\\
\elrbc{\Theta \vdash \mu X.A} &:= \left(\FOLD{\mu X. A[\vec C / \Theta]} \circ \elrb{\Theta, X \vdash A}{\vec C, \mu X. A[\vec C / \Theta]} \right)^\sharp , &&
\end{align*}
where the $(-)^\sharp$ operation is from Definition \ref{def:initial-algebra}.
\end{definition} 
 
\begin{proposition}
\label{prop:augmented-interpretation}
Each functor $\elrbc{\Theta \vdash A}$ is well-defined and $\omega$-cocontinuous. Moreover, the following diagram:
\cstikz{augmented-diagram.tikz}
commutes.
\end{proposition}
\begin{proof}
The proof is essentially the same as \cite[Proposition 7.26]{lnl-fpc-lmcs}.
\end{proof}

Next, a corollary which shows that parameterised initial algebras for our type expressions are constructed in the same way in both categories.

\begin{corollary}
\label{cor:augmented-algebras}
The 2-categorical diagram:
\cstikz{parameterised-initial-algebra-augmented.tikz}
commutes, where $\iota$ is the parameterised initial algebra isomorphism (see Definition \ref{def:initial-algebra}).
\end{corollary}
\begin{proof}
The proof is the same as \cite[Corollary 7.27]{lnl-fpc-lmcs}.
\end{proof}

Proposition \ref{prop:augmented-interpretation} shows that the first component of the augmented interpretation coincides with the standard interpretation. This is true for all types, including open ones.
In the special case for closed types, let $\elrbs{A} \eqdef \elrb{\cdot \vdash A}{\cdot}(*)$, where $*$ is the unique object of the terminal category $\mathbf 1 = \RR(A)^0$.
Proposition \ref{prop:augmented-interpretation} therefore shows that $U \elrbs A = \lrb A$, which means that $\elrbs A$ has the form $\elrbs A = (\lrb A, e, \tleq),$ where $e : \lrb A \kto \lrb A$ is some
embedding. Next, we show that $e = \kid.$ In order to do this, we prove a stronger proposition first. We show that the action of the functor $\elrbc{\Theta \vdash A}$ on the embedding component is also completely determined by the action of $\lrb{\Theta \vdash A}$ on embeddings.

\begin{proposition}
\label{prop:augmented-embedding}
For every functor $\elrbc{\Theta \vdash A}$ and objects $(X_i, e_i, \tleq_i)$ with $i \in \{1, \ldots, n\}$, we have:
\[ \pi_e \left( \elrbc{\Theta \vdash A}\left( (X_1, e_1, \tleq_1), \ldots, (X_n, e_n, \tleq_n) \right) \right) =  \lrb{\Theta \vdash A}(e_1, \ldots, e_n) , \]
where for an object $(Z, e_Z, \tleq_Z)$ in any category $\RR(B)$, we define $\pi_e(Z, e_Z, \tleq_Z) = e_Z.$ 
\end{proposition}
\begin{proof}
By induction on the derivation of $\Theta \vdash A.$

\paragraph*{\textbf{Case} $\Theta_i$ } This is obviously true.

\paragraph*{\textbf{Case $A = A_1 \star A_2$, for $\star \in \{\times, +, \to \}$}}
The statement follows easily by induction and the fact that for every pair of objects $(Y, e_Y, \tleq_Y)$ and $(Z, e_Z, \tleq_Z)$ we have
\[ \pi_e\left(  (Y, e_Y, \tleq_Y)  \star^{A_1,A_2}  (Z, e_Z, \tleq_Z) \right) = e_Y \kstar_e e_Z \]
which follows by definition of the relevant functors.

\paragraph*{\textbf{Case $\mu X. A$}}
First we introduce some abbreviations to simplify notation. We define:
\begin{itemize}
\setlength\itemsep{0.5em}
\item $T \eqdef  \elrb{\Theta, X \vdash A}{\vec C, \mu X. A[\vec C / \Theta]}.$
\item $H \eqdef \lrb{\Theta, X \vdash A}.$
\item $\FOLD \eqdef \FOLD{\mu X. A[\vec C / \Theta]} .$
\item $\vv{(X, e, \tleq)} \eqdef ((X_1, e_1, \tleq_1), \ldots, (X_n, e_n, \tleq_n) )$.
\item $\vv{X} \eqdef (X_1, \ldots, X_n)$.
\item $\vv{e} \eqdef (e_1, \ldots, e_n)$.
\end{itemize}

Now, let $ (Y, e_Y, \tleq_Y) \eqdef (\FOLD{} \circ T)^\sharp \vv{(X, e, \tleq)}.$ To finish the proof, we have to show that $H^\sharp(\vv e) = e_Y.$
From Proposition \ref{prop:augmented-interpretation} we know that $Y = H^\sharp(\vv X).$ From Corollary \ref{cor:augmented-algebras}, we have a parameterised initial algebra isomorphism
\begin{equation}
\label{eq:algebra-lifted}
\iota \colon \FOLD{}T \left( \vv{(X,e,\tleq)}, (H^\sharp \vv{X}, e_Y, \tleq_Y) \right) \to (H^\sharp \vv{X}, e_Y, \tleq_Y) 
\end{equation}
which is also a parameterised initial algebra isomorphism
\begin{equation}
\label{eq:algebra-lower}
\iota \colon H \left( \vv{X}, H^\sharp \vv{X} \right) \to H^\sharp \vv{X} 
\end{equation}
in $\PDe$. By the induction hypothesis for $T$ and $H$ and Proposition \ref{prop:augmented-interpretation}, we get
\[ T \left( \vv{(X,e,\tleq)}, (H^\sharp \vv{X}, e_Y, \tleq_Y) \right) = \left( H(\vv X, H^\sharp \vv X) , H(\vv e, e_Y) , \btleq \right) , \]
where $\btleq$ is some (unimportant) logical relation. Therefore by \eqref{eq:algebra-lifted} and definition of $\FOLD{}$, we get that 
\begin{equation}
\label{eq:algebra-last}
\iota \colon \left( H(\vv X, H^\sharp \vv X) , \mathrm{fold} \kcirc H(\vv e, e_Y) , \mathbb I \btleq \right) \to (H^\sharp \vv{X}, e_Y, \tleq_Y) 
\end{equation}
is an isomorphism with the indicated type. This means that in the category $\PDe$, we have:
\begin{equation}
\label{eq:initial-universal}
\mathrm{fold} \kcirc H(\vv e, e_Y) = e_Y \kcirc \iota
\end{equation}
where we already know that $\iota = \iota_{X_1, \ldots, X_n}$ is the parameterised initial algebra in $\PDe$ of $H$. But, by definition, so is $\mathrm{fold}$ and in fact
$\mathrm{fold} = \iota_{\lrb{C_1}, \ldots, \lrb{C_n}}.$ However, $H^\sharp \vv e$ is the unique morphism, such that
\[ \iota_{\lrb{C_1}, \ldots, \lrb{C_n}} \kcirc H(\vv e, H^\sharp \vv e) = H^\sharp \vv e \kcirc \iota_{X_1, \ldots X_n} \]
which is the universal property of a parameterised initial algebra (see \cite[Remark 4.6]{lnl-fpc-lmcs}) and therefore by equation \eqref{eq:initial-universal} it follows that $e_Y = H^\sharp \vv e,$ as required.
\end{proof}

\begin{corollary}
\label{cor:cool-form}
For every closed type $A$, we have $\elrbs A = (\lrb A, \kid_{\lrb A}, \tleq_A)$ for some logical relation $\tleq_A.$ 
\end{corollary}
\begin{proof}
We already know that the first component is $\lrb A$. For the second component, the previous proposition shows that $\pi_e \elrbs A = \pi_e \elrb{\cdot \vdash A}{\cdot}(*) = \lrb{\cdot \vdash A}(\id_*) = \kid_{\lrb A},$ where $*$ denotes the empty tuple of objects and $\id_*$ the empty tuple of embeddings.
\end{proof}

Finally, we want to show that the third component of $\elrbs A$ is the logical relation that we need to carry out the adequacy proof. For this, we have to prove a substitution lemma first.

\begin{lemma}[Substitution]
\label{lem:substitution-augmented}
For any types $\Theta, X \vdash A$ and $\Theta \vdash B$ and closed types $C_1, \ldots, C_{n}$, we have:
\[ \elrbc{\Theta \vdash A[B/X]} = \elrb{\Theta, X \vdash A}{\vec C, B[\vec C / \Theta]} \circ \langle \Id, \elrbc{\Theta \vdash B} \rangle . \]
\end{lemma}
\begin{proof}
The proof is the same as \cite[Lemma 7.30]{lnl-fpc-lmcs}.
\end{proof}

For each type $A$, we have now provided an augmented interpretation $\elrbs A$ of $A$ in the category $\RR(A).$
The interpretation $\elrbs{-}$ satisfies all the fundamental properties of $\lrb{-},$ as we have now shown. It should now be clear that this augmented interpretation is true to its name, because it carries strictly more information compared to the standard interpretation of types.
The additional information that $\elrbs{A}$ carries is precisely the logical relation that we need at type $A$, as we show in the next subsection.

\subsection{Existence of the Logical Relations}

We can now show that the logical relations we need for the adequacy proof exist.

\begin{theorem}
\label{thm:formal-relations}
For each closed type $A$, there exist \emph{formal approximation relations:}
\begin{align*}
\tleq_A &\subseteq \TD(1, \lrb A) \times \mathrm{Val}(A) \\
\ol{\tleq_A} &\subseteq \KL(1, \lrb A) \times \mathrm{Prog}(A)
\end{align*}
which satisfy the following properties:
\begin{enumerate}
\setlength\itemsep{0.4em}
\item [(A1)] $ \JJ \emph{in}_i \kcirc v \tleq_{A_1 + A_2} \mathtt{in}_i V \text{ iff } v \tleq_{A_i} V$, where $i \in \{1,2\}$.
\item [(A2)] $\llangle v_1, v_2 \rrangle \tleq_{A_1 \times A_2} (V_1,V_2) \text{ iff }  v_1 \tleq_{A_1} V_1 \text{ and } v_2 \tleq_{A_2} V_2 .$
\item [(A3)] $ f \tleq_{A \to B} \lambda x. M \text{ iff }  f \leq \lrb{\lambda x.M} \text{ and } \forall (v \tleq_A V).\  f[v] \ol{\tleq_B} (\lambda x.M)V. $
\item [(A4)] $ v \tleq_{\mu X.A} \mathtt{fold}\ V \text{ iff } \sunfold \kcirc v \tleq_{A[\mu X. A / X]} V$.
\item [(B)] $ m \ol{\tleq_A} M \text{ iff } m \in \mathcal S(\tleq_A; M), $ where $\mathcal S(\tleq_A; M)$ is the Scott-closure in $\KL(1, \lrb A)$ of the set
\[ \mathcal S_{0}(\tleq_A; M) \defeq \left\{ \sum_{\pi \in F} P(\pi) v_\pi \ |\ F \subseteq \TPaths(M),\ F \text{ is finite and }  \text{$v_\pi \tleq_A V_\pi$ for each $\pi \in F$} \right\} \ (\text{see Definition \ref{def:paths}}) . \]
\item [(C1)]\label{item:below} If $v \tleq_A V$, then $v \leq \lrb V$.
\item [(C2)] $(- \tleq_A V)$ is a Scott-closed subset of $\TD(1,\lrb A).$
\item [(C3)] If $m \ol{\tleq_A} M$, then $m \leq \lrb M$.
\item [(C4)] $(- \ol{\tleq_A} M)$ is a Scott-closed subset of $\KL(1,\lrb A).$
\item [(C5)] If $v \in \TD(1, \lrb A)$ and $V$ is a value, then $v \tleq_A V$ iff $v \ol{\tleq_A} V.$
\end{enumerate}
\end{theorem}
\begin{proof}
Consider the object $\elrbs A \in \RR(A).$ We have already shown that $\elrbs A = (\lrb A, \kid_{\lrb A}, \tleq_A)$ for some logical relation $\tleq_A \in \ValRel(\lrb A, A, \kid_{\lrb A}).$ We now show that $\tleq_A$ satisfies the required properties.
Notice that the embedding components are just identities.

Property (B) is satisfied by construction (Definition \ref{def:logical-closure}).
Properties (C1) and (C2) are also satisfied by construction (Definition \ref{def:logical-relations}). Property (C4) is satisfied by construction and property (B).
Property (C3) is satisfied, because if $m \ol{\tleq_A} M$, then by Corollary \ref{cor:soundness-inequality} and property (C1) it follows that $\mathcal S_0(\tleq_A; M) \subseteq \down \lrb M $. The latter set is Scott-closed and therefore
$m \in \mathcal S(\tleq_A; M) \subseteq \down \lrb M ,$ as required. Property (C5) is satisfied by Lemma \ref{lem:id}.

Properties (A1), (A2) and (A3) are satisfied, because for $\star \in \{+, \times, \to\} $, we have that $\tleq_{A \star B} = \tleq_A \star \tleq_B$ and then by Definition \ref{def:logical-constructions}.

To show that property (A4) is also satisfied, we reason as follows. Consider the isomorphism
\[  \sunfold_{\mu X. A} : \lrb{\mu X. A}   \cong \lrb{X \vdash A} \lrb{\mu X. A}   =    \lrb{A[\mu X. A/ X]} : \sfold_{\mu X. A} \]
from Definition \ref{def:fold/unfold}. By Corollary \ref{cor:augmented-algebras} and Lemma \ref{lem:substitution-augmented} (when $\Theta = \cdot$) it follows that this isomorphism lifts to an isomorphism
\[  \sunfold_{\mu X. A} : \elrbs{\mu X.A} \cong \FOLD{\mu X.A} \left( \elrb{X \vdash A}{\mu X.A} \left( \elrbs{\mu X. A} \right) \right) = \FOLD{\mu X.A} \left( \elrbs{A[ \mu X. A / X]} \right) : \sfold_{\mu X. A} \]
in the category $\RR(\mu X.A).$ Expanding definitions, this means we have an isomorphism
\begin{equation}
\label{eq:fold-unfold-adequacy}
\begin{split}
\sunfold_{\mu X.A } : (\lrb{\mu X.A}, \kid, \tleq_{\mu X.A})
& = \elrbs{\mu X.A}\\
& \cong \FOLD{\mu X.A} \left( \elrbs{A[ \mu X. A / X]} \right)\\
& = (\lrb{A[ \mu X. A / X]}, \sfold_{\mu X.A}, \FOLD{\mu X.A} \tleq_{A[ \mu X. A / X]}) : \sfold_{\mu X. A}
\end{split}
\end{equation}
in the category $\RR(\mu X. A).$ The notion of morphism in this category (Definition \ref{def:logical-categories}), construction of $\mathbb I$ (Definition \ref{def:logical-fold-unfold}) and property (C5) allow us to conclude that property (A4) is satisfied. Indeed:
\begin{align*}
&v \tleq_{\mu X. A} \mathtt{fold}\ V &\\
\Longrightarrow \quad & \sunfold_{\mu X.A} \kcirc v\ (\FOLD{\mu X.A} \tleq_{A[ \mu X. A / X]} )\ \mathtt{fold}\ V \\
\Longrightarrow \quad & \sunfold_{\mu X.A} \kcirc v\ \tleq_{A[ \mu X. A / X]} V
\end{align*}
and for the other direction of (A4):
\begin{align*}
& \sunfold_{\mu X.A} \kcirc v \tleq_{A[ \mu X. A / X]} V & \\
\Longrightarrow \quad & \sunfold_{\mu X.A} \kcirc v\ (\FOLD{\mu X.A} \tleq_{A[ \mu X. A / X]} )\ \mathtt{fold}\ V \\
\Longrightarrow \quad & v = \sfold_{\mu X. A} \kcirc \sunfold_{\mu X.A} \kcirc v \tleq_{\mu X. A} \mathtt{fold}\ V .
\qedhere
\end{align*}
\end{proof}

\subsection{Closure Properties of the Logical Relations}
\label{sub:closure-properties}

Here we establish some important closure properties of the relations $\overline{\tleq_A}$ from Theorem \ref{thm:formal-relations}.

\begin{lemma}
\label{lem:logical-convex-sum}
Let $\cdot \vdash M : A$ be a term and let $F$ be some finite index set. Assume that we are given morphisms $m_i$ and terms $M_i$ such that $m_i \ol{\tleq_A} M_i$ for $i \in F$.
Assume further that for each $i \in F$, we are given a reduction path $\pi_i \in \Paths(M, M_i)$, such that all paths $\pi_i$ are distinct.
Then
\[ \sum_{i \in F} P(\pi_i) m_i \ol{\tleq_A} M. \]
\end{lemma}
\begin{proof}
By assumption, for every $i \in F$, we know that $m_i \in \mathcal S(\tleq_A; M_i)$. Next, consider the function
\[ g \eqdef \sum_{i \in F} P(\pi_i)(-) \colon \prod_{|F|} \KL(1, \lrb A) \to \KL(1, \lrb A) . \]
This function is Scott continuous and therefore by Lemma \ref{lem:topological-goodness}, it suffices to show that
$g(\prod_i s_i) \in \mathcal S(\tleq_A ; M)$ for any choice of $s_i \in \mathcal S_0(\tleq_A; M_i).$
Next, for every $i \in F$, let
\[ s_i = \left( \sum_{\pi \in F_i} P(\pi) v_{\pi} \right) \in \mathcal S_0(\tleq_A; M_i) \]
where $F_i \subseteq \TPaths(M_i)$ is a finite subset and such that $v_{\pi} \tleq_{A} V_\pi,$ for each $\pi \in F_i.$ Then, we have
\begin{align*}
g \left( \prod_i s_i \right) &= \sum_{i \in F} P(\pi_i) \left( \sum_{\pi \in F_i} P(\pi)v_{\pi} \right)  \\
&= \sum_{i \in F} \sum_{\pi \in F_i} \left(P(\pi_i) \cdot P(\pi) \right) v_{\pi} \\
&= \sum_{i \in F} \sum_{\pi \in F_i} P(\pi_i \pi)  v_{\pi} \\
&\in \mathcal S_0(\tleq_A; M) , 
\end{align*}
where $\pi_i \pi \in \Paths(M, V_\pi)$ is the path constructed by concatenating the path $\pi_i$ to $\pi$.
\end{proof}

\begin{lemma}
\label{lem:logical-probabilistic-choice}
If $m \ol{\tleq_A} M$ and $n \ol{\tleq_A} N,$ then $p \cdot m + (1-p) \cdot n \ol{\tleq_A} M\ \mathtt{or}_p\ N.$
\end{lemma}
\begin{proof}
This is just a special case of Lemma \ref{lem:logical-convex-sum}.
\end{proof}

\begin{lemma}
\label{lem:logical-injections}
For $i \in \{1,2\}:$ if $m \ol{\tleq_{A_i}} M$, then $\JJ \emph{\emph{in}}_i \kcirc m \ol{\tleq_{A_1 + A_2}} \mathtt{in}_i M .$
\end{lemma}
\begin{proof}
Assume, without loss of generality, that $i=1$. By definition we know that $m \in \mathcal S(\tleq_{A_1}; M) = \ol{\mathcal S_{0}(\tleq_{A_1}; M)}$.
By Lemma \ref{lem:topological-goodness}, it suffices to show
\[ \JJ\emph{in}_1 \kcirc \sum_{\pi \in F} P(\pi) v_\pi \in \mathcal S(\tleq_{A_1 + A_2}; \mathtt{in}_1 M) \]
for any  $\sum_{\pi \in F} P(\pi) v_\pi \in \mathcal S_0(\tleq_{A_1}; M) . $
Since $(\JJ\emph{in}_1 \kcirc -)$ is linear, we see
\[ \JJ\emph{in}_1 \kcirc \sum_{\pi \in F} P(\pi) v_\pi = \sum_{\pi \in F} P(\pi) ( \JJ\emph{in}_1 \kcirc v_\pi ) = \sum_{\pi \in F} P(\mathrm{in}_1(\pi)) ( \JJ\emph{in}_1 \kcirc v_\pi )  \in \mathcal S(\tleq_{A_1 + A_2}; \mathtt{in}_1 M)  , \]
where $\mathrm{in}_1(\pi) \in \Paths(\mathtt{in_1} M, \mathtt{in}_1 V_\pi)$ is the path constructed by reducing $\mathtt{in}_1 M$ to $\mathtt{in}_1 V_\pi$, as specified by $\pi.$
The membership relation is satisfied because by assumption $v_\pi \tleq_{A_1} V_\pi$ and then by Theorem~\ref{thm:formal-relations}~(A1).
\end{proof}

\begin{lemma}
\label{lem:logica-case}
Let $m \ol{\tleq_{A_1 + A_2}} M$. Next, assume that for $k \in \{1,2\}$ we have terms $x_k : A_k \vdash N_k : B$ and morphisms $n_k \colon \lrb{A_k} \kto \lrb{B}$, such that for every $v_k \tleq_{A_k} V_k,$ it is the case that $n_k \kcirc v_k \ol{\tleq_B} N_k[V_k / x_k].$
Then 
\[ [n_1, n_2] \kcirc m \ol{\tleq_B} \mathtt{case}\ M\ \mathtt{of}\ \mathtt{in}_1 x_1 \Rightarrow N_1\ |\ \mathtt{in}_2 x_2 \Rightarrow N_2 . \]
\end{lemma}
\begin{proof}
For brevity, let $C$ be the term $C \defeq ( \mathtt{case}\ M\ \mathtt{of}\ \mathtt{in}_1 x_1 \Rightarrow N_1\ |\ \mathtt{in}_2 x_2 \Rightarrow N_2 ).$
Next, consider the function
\[ ([n_1,n_2] \kcirc -) \colon \KL(1, \lrb{A_1 + A_2}) \to \KL(1,\lrb B) . \]
This function is Scott continuous.
By Lemma \ref{lem:topological-goodness}, to complete the proof it suffices to show that $[n_1,n_2] \kcirc m' \ol{\tleq_B} C$ for any $m' \in \mathcal S_0(\tleq_{A_1 + A_2}; M)$. Towards that end, let
\begin{align*}
m' = \sum_{\pi \in F} P(\pi) v_\pi ,
\end{align*}
where $F$ is finite and where $v_\pi \tleq_{A_1 + A_2} V_\pi$, for each $\pi \in F.$
Let $F_1 \subseteq F$ be the set of paths $\pi$ such that $V_\pi = \mathtt{in}_1 V_\pi'$ for some $V_\pi'$ and let $F_2 = F - F_1.$
Then  by Theorem~\ref{thm:formal-relations}~(A1), for each $\pi \in F_1$, it follows that $V_\pi = \mathtt{in}_{1} V_\pi'$ and $v_\pi = \JJ \emph{in}_{1} \kcirc v_\pi'$ and $v_\pi' \tleq_{A_{1}} V_\pi'$.
Similarly, for each $\pi \in F_2$, it follows that  $V_\pi = \mathtt{in}_{2} V_\pi'$ and $v_\pi = \JJ \emph{in}_{2} \kcirc v_\pi'$ and $v_\pi' \tleq_{A_{2}} V_\pi'$.
Therefore, we get:
\begin{align*}
[n_1,n_2] \kcirc m' &= [n_1,n_2] \kcirc \left( \left( \sum_{\pi \in F_1} P(\pi) (\JJ \emph{in}_1 \kcirc v_\pi') \right) + \left(  \sum_{\pi \in F_2} P(\pi) (\JJ \emph{in}_2 \kcirc v_\pi')  \right) \right) \\ 
&= \left( \sum_{\pi \in F_1} P(\pi) (n_1 \kcirc v_\pi') \right) + \left(  \sum_{\pi \in F_2} P(\pi) (n_2 \kcirc v_\pi')  \right) \\ 
\end{align*}
In the above sums, by assumption, we know that $n_1 \kcirc v_\pi' \ol{\tleq_B} N_1[V_\pi'/ x_1]$, for each $\pi \in F_1$ and similarly 
$n_2 \kcirc v_\pi' \ol{\tleq_B} N_2[V_\pi' / x_2]$, for each $\pi \in F_2.$ Next, consider the function
\begin{align*}
\left( \left( \sum_{\pi \in F_1} P(\pi)( - ) \right) + \left( \sum_{\pi \in F_2} P(\pi)( - ) \right) \right) : \KL(1, \lrb{B})^{|F_1|} \times \KL(1, \lrb{B})^{|F_2|}  \to \KL(1, \lrb B) .
\end{align*}
This function is Scott-continuous and by Lemma \ref{lem:topological-goodness}, to complete the proof it suffices to show that
\[   \left( \sum_{\pi \in F_1} P(\pi)( n_{1}^\pi ) \right) + \left( \sum_{\pi \in F_2} P(\pi)( n_{2}^\pi ) \right)  \ol{\tleq_B} C , \]
where $n_{1}^\pi \in \mathcal S_0(\tleq_B ; N_1[V_\pi' / x_1])$ for $\pi \in F_1$ and $n_{2}^\pi \in \mathcal S_0(\tleq_B ; N_2[V_\pi' / x_2])$ for $\pi \in F_2$ are taken to be arbitrary.
Towards this end, let
\begin{align*}
n_{1}^\pi &= \sum_{\pi' \in F_{1}^\pi} P(\pi') v_{\pi'} \in \mathcal S_0(\tleq_B ; N_1[V_\pi' / x_1]) \\
n_{2}^\pi &= \sum_{\pi' \in F_{2}^\pi} P(\pi') v_{\pi'} \in \mathcal S_0(\tleq_B ; N_2[V_\pi' / x_2])
\end{align*}
where $F_{k}^\pi$ is finite and where $v_{\pi'} \tleq_{B} V_{\pi'}$, for every $\pi' \in F_{k}^\pi$ and where $k \in \{1,2\}$. Then, we get
\begin{align*}
& \left( \sum_{\pi \in F_1} P(\pi)( n_{1}^\pi ) \right) + \left( \sum_{\pi \in F_2} P(\pi)( n_{2}^\pi ) \right) = \\
&= \left( \sum_{\pi \in F_1} \sum_{\pi' \in F_{1}^\pi} P(\pi) P(\pi') v_{\pi'}  \right) + \left( \sum_{\pi \in F_2} \sum_{\pi' \in F_{2}^\pi} P(\pi) P(\pi') v_{\pi'}  \right)  \\
&= \left( \sum_{\pi \in F_1} \sum_{\pi' \in F_{1}^\pi} P( \mathrm{case}_1(\pi,\pi') ) v_{\pi'}  \right) + \left( \sum_{\pi \in F_2} \sum_{\pi' \in F_{2}^\pi} P( \mathrm{case}_2(\pi,\pi') ) v_{\pi'}  \right)  \\
& \in \mathcal S_0(\tleq_B; C) \subseteq \mathcal S(\tleq_B; C) ,
\end{align*}
where $\mathrm{case}_1(\pi, \pi') \in \Paths(C, V_{\pi'})$ is the path obtained by reducing $C$ to $C_\pi \eqdef ( \mathtt{case}\ \mathtt{in}_1 V_\pi'\ \mathtt{of}\ \mathtt{in}_1 x_1 \Rightarrow N_1\ |\ \mathtt{in}_2 x_2 \Rightarrow N_2 )$ as specified by $\pi$, then performing
the beta reduction $C_\pi \probto 1 N_1[V_\pi'/x_1]$ and then reducing $N_1[V_\pi'/x_1]$ to $V_{\pi'}$ as specified by $\pi'$. Similarly for $\mathrm{case}_2(\pi, \pi').$
The last sum is now by definition in $\mathcal S_0(\tleq_B; C)$.
\end{proof}

\begin{lemma}
\label{lem:logical-pairs}
If $m_1 \ol{\tleq_{A_1}} M_1$ and $m_2 \ol{\tleq_{A_2}} M_2$ then $\llangle m_1, m_2 \rrangle \ol{\tleq_{A_1 \times A_2}} (M_1, M_2).$
\end{lemma}
\begin{proof}
The map
$\llangle -, - \rrangle \colon \KL(1, \lrb{A_{1}}) \times \KL(1, \lrb{A_{2}}) \to \KL(1, \lrb{A_{1}\times A_{2}}) $
is Scott-continuous in both arguments and therefore by Lemma~\ref{lem:topological-goodness}, to complete the proof it suffices to show that $\llangle m_1', m_2' \rrangle \ol{\tleq_{A_1 \times A_2}} (M_1, M_2) $ for any $m_1' \in \mathcal S_0(\tleq_{A_1} ; M_1)$ and $m_2' \in \mathcal S_0(\tleq_{A_2}; M_2)$.

Now, take $m_{1}' =\sum_{\pi_1\in F_1} P(\pi_1) v_{\pi_1} \in \mathcal S_{0}({\tleq_{A_1}}; M_1)$ and $m_{2}' = \sum_{\pi_2 \in F_2} P(\pi_2) v_{\pi_2} \in \mathcal S_{0}({\tleq_{A_2}}; M_2)$, where $F_1$ and $F_2$ are finite sets, and
where $v_{\pi_1} \tleq_{A_1} V_{\pi_1}$ for each $\pi_1 \in F_1$ and where $v_{\pi_2} \tleq_{A_2} V_{\pi_2}$ for each $\pi_2 \in F_2.$ We then have:
\begin{align}
\llangle m_1', m_2' \rrangle &=  \llangle \sum_{\pi_1 \in F_1} P(\pi_1) v_{\pi_1},  \sum_{\pi_2\in F_2} P(\pi_2) v_{\pi_2}  \rrangle  \label{eq:times1}  \\ 
&= \sum_{\pi_1 \in F_1} \sum_{\pi_2\in F_2} P(\pi_1) P(\pi_2)  \llangle v_{\pi_1}, v_{\pi_2} \rrangle \label{eq:times2}  \\
&= \sum_{\pi_1 \in F_1} \sum_{\pi_2\in F_2} P(\mathrm{pair}(\pi_1, \pi_2))  \llangle v_{\pi_1}, v_{\pi_2} \rrangle   \label{eq:times3} \\
& \ol{\tleq_{A_1 \times A_2}} (M_1, M_2) .  \label{eq:times4}
\end{align}
Equation~\ref{eq:times1} holds by definition. Equation~\ref{eq:times2} is true since the function $\llangle -, - \rrangle $ defined above is linear in each component by  
Lemma~\ref{lemma:respectkegel}~Item~3. In Equation~\ref{eq:times3} $\mathrm{pair}(\pi_1,\pi_2) \in \Paths( (M_1,M_2), (V_{\pi_1}, V_{\pi_2})   )$ is the path which first reduces $(M_1, M_2)$ to $(V_{\pi_1}, M_2)$ as specified by $\pi_1$ and then reduces $(V_{\pi_1}, M_2)$ to $(V_{\pi_1}, V_{\pi_2})$ as specified by $\pi_2$ and it is easy to see
that Equation \ref{eq:times3} holds. Finally \ref{eq:times4} holds, because $v_{\pi_1} \tleq_{A_1} V_{\pi_1}$ and $v_{\pi_2} \tleq_{A_2} V_{\pi_2}$ by assumption and then by Theorem~\ref{thm:formal-relations}~(A2) we have that $\llangle v_{\pi_1}, v_{\pi_2} \rrangle \tleq_{(A_1, A_{2})} (V_{\pi_1}, V_{\pi_2})$.
\end{proof}

\begin{lemma}
\label{lem:logical-projections}
If $m \ol{\tleq_{A_1 \times A_2}} M$ then $\JJ \pi_i \kcirc m \ol{\tleq_{A_i}} \pi_i M$, for $i \in \{1,2\}.$
\end{lemma}
\begin{proof}
Without loss of generality, we will show the statement for the first projection.
In order to avoid notational confusion, we will write $\mathrm{pr}_1$ for $\pi_1$ for the projection on the first component in this lemma.
We shall use $\pi$ to range over paths, as in the other lemmas.

Using Lemma \ref{lem:topological-goodness}, to complete the proof it suffices to show that 
\[ \JJ \mathrm{pr}_1 \kcirc m' \ol{\tleq_{A_1}} \mathtt{pr}_1 M \]
for any $m' \in \mathcal S_0(\tleq_{A_1 \times A_2}; M). $ Towards this end, let
\[ m' = \sum_{\pi \in F} P(\pi) v_\pi \in \mathcal S_0(\tleq_{A_1 \times A_2}; M) , \]
where $F \subseteq \TPaths(M)$ is finite and where $v_\pi \tleq_{A_1 \times A_2 } V_\pi$ for every $\pi \in F$. Using Theorem \ref{thm:formal-relations} (A2), we see that it must be the case
\[ v_\pi = \llangle v_\pi^1, v_\pi^2 \rrangle \text{ and } V_\pi = (V_\pi^1, V_\pi^2) \text{ and } v_\pi^1 \tleq_{A_1} V_\pi^1 \text{ and } v_\pi^2 \tleq_{A_2} V_\pi^2 . \]
Therefore, we have
\begin{align*}
\JJ \mathrm{pr}_1 \kcirc m' &= \JJ \mathrm{pr}_1 \kcirc \sum_{\pi \in F} P(\pi) v_\pi  \\
&= \JJ \mathrm{pr}_1 \kcirc \sum_{\pi \in F} P(\pi) \llangle v_\pi^1, v_\pi^2 \rrangle \\
&= \sum_{\pi \in F} P(\pi) ( \JJ \mathrm{pr}_1 \kcirc  \llangle v_\pi^1, v_\pi^2 \rrangle ) \\
&= \sum_{\pi \in F} P(\pi) v_\pi^1 \\
&= \sum_{\pi \in F} P(\mathrm{pr}_1(\pi)) v_\pi^1 \\
&\ol{\tleq_{A_1}} \mathtt{pr}_1 M ,
\end{align*}
where $\mathrm{pr}_1(\pi) \in \Paths( \mathrm{pr}_1 M, V_\pi^1)$ is the path that reduces $\mathtt{pr}_1 M$ to $\mathtt{pr}_1 (V_\pi^1, V_\pi^2)$ as specified by $\pi$ and then finally performs the
reduction $\mathrm{pr}_1 (V_\pi^1, V_\pi^2) \probto 1 V_{\pi}^1$.
\end{proof}

\begin{lemma}
\label{lem:logical-unfold}
If $m \ol{\tleq_{\mu X. A}} M$ then $\sunfold{} \kcirc m \ol{\tleq_{ A[\mu X. A/X] }} \mathtt{unfold}\ M.$
\end{lemma}
\begin{proof}
By Lemma \ref{lem:topological-goodness}, to complete the proof it suffices to show that
\[ \sunfold \kcirc m' \in \mathcal S(\tleq_{ A[\mu X. A/X]};  \mathtt{unfold}\ M) \]
for any $m'\in \mathcal S_{0}(\tleq_{\mu X. A}; M)$.
Towards this end, let
\[ m' = \sum_{\pi\in F} P(\pi)v_{\pi} \in \mathcal S_{0}(\tleq_{\mu X. A}; M) \]
for some finite $F \subseteq \TPaths(M)$ and where $v_\pi \tleq_{\mu X. A} V_\pi = \mathtt{fold}\ V_\pi'$ for each $\pi \in F$.
Then we have
\begin{align*}
\sunfold\kcirc m' &= \sum_{\pi\in F} P(\pi) ( \sunfold \kcirc v_{\pi} ) \\
&=  \sum_{\pi\in F} P(\mathrm{unfold}(\pi)) ( \sunfold \kcirc v_{\pi} ) \\
& \in \mathcal S_0(\tleq_{ A[\mu X. A/X]};  \mathtt{unfold}\ M) ,
\end{align*}
where $\mathrm{unfold}(\pi) \in \Paths(\mathtt{unfold}\ M, V_\pi')$ is the path that reduces $\mathtt{unfold}\ M$ to $\mathtt{unfold}\ \mathtt{fold}\ V_\pi'$ as specified by $\pi$ and then finally performs the
reduction $\mathtt{unfold}\ \mathtt{fold}\ V_\pi' \probto{1} V_\pi'.$ This last sum satisfies the membership relation, because we know that $v_\pi \tleq_{\mu X. A} V_\pi = \mathtt{fold}\ V_\pi'$ and then
by Theorem \ref{thm:formal-relations} (A4) we see that $\sunfold \kcirc v_\pi \tleq_{ A[\mu X. A/X] } V_\pi',$ as required.
\end{proof}

\begin{lemma}
\label{lem:logical-fold}
If $m \ol{\tleq_{A[\mu X. A / X]}} M$ then $\sfold{} \kcirc m \ol{\tleq_{\mu X. A}} \mathtt{fold}\ M.$
\end{lemma}
\begin{proof}
The function
\[ (\sfold{} \kcirc -) \colon \KL(1, \lrb{A[\mu X. A / X]}) \to \KL(1, \lrb{\mu X. A}) \]
is Scott-continuous and therefore by Lemma \ref{lem:topological-goodness}, to complete the proof it suffices to show that
\[ \sfold{} \kcirc m' \in \mathcal S(\tleq_{\mu X. A}; \mathtt{fold}\ M) \]
for each $m'\in \mathcal S_{0}(\tleq_{A[\mu X. A / X]}; M)$.
Towards this end, assume that
\[ m' = \sum_{\pi\in F}P(\pi) v_{\pi} \in \mathcal S_{0}(\tleq_{A[\mu X. A / X]}; M), \] 
where $F \subseteq \TPaths(M)$ is finite and for each $\pi\in F$ we have $v_{\pi} \tleq_{A[\mu X. A / X]} V_{\pi}$.
Therefore, by Theorem \ref{thm:formal-relations} (A4) we conclude that $\sfold{} \kcirc v_\pi \tleq_{\mu X. A} \mathtt{fold}\ V_\pi,$ for each $\pi \in F.$
Now we finish the proof with the following derivation:
\begin{align*}
\sfold{} \kcirc m' &= \sfold{} \kcirc   \sum_{\pi\in F}P(\pi) v_{\pi}  \\
&= \sum_{\pi\in F} P(\pi) ( \sfold{}\kcirc v_{\pi} ) \\
&= \sum_{\pi\in F} P(\mathrm{fold}(\pi)) ( \sfold{}\kcirc v_{\pi} ) \\
&\in \mathcal S_0(\tleq_{\mu X. A}; \mathtt{fold}\ M) \subseteq \mathcal S(\tleq_{\mu X. A}; \mathtt{fold}\ M) ,
\end{align*}
where $\mathrm{fold}(\pi) \in \Paths(\mathtt{fold}\ M, \mathtt{fold}\ V_\pi)$ is the path that reduces $\mathtt{fold}\ M$ to $\mathtt{fold}\ V_\pi$ as specified by $\pi.$
\end{proof}

\begin{lemma}
\label{lem:logical-aplication}
If $m \ol{\tleq_{A \to B}} M$ and $n \ol{\tleq_A} N,$ then $m[n] \ol{\tleq_B} MN.$
\end{lemma}
\begin{proof}
Consider the function $g \colon \KL(1, \lrb{A \to B}) \times \KL(1, \lrb{A}) \to \KL(1,\lrb B)$ defined by $g(x,y) = x[y]$ (see Notation \ref{not:application}). This function is Scott continuous and linear in both arguments.
By Lemma \ref{lem:topological-goodness}, to complete the proof it suffices to show that $m'[n'] \ol{\tleq_B} MN$ for any $m' \in \mathcal S_0(\tleq_{A \to B}; M)$ and $n' \in \mathcal S_0(\tleq_A; N).$ Towards that end, let
\begin{align*}
m' &= \sum_{\pi \in F} P(\pi) v_{\pi} \in \mathcal S_0(\tleq_{A \to B}; M) \\
n' &= \sum_{\pi' \in F'} P(\pi') v_{\pi'} \in \mathcal S_0(\tleq_A ; N)
\end{align*}
with $v_\pi \tleq_{A \to B} V_\pi$ and $v_{\pi'} \tleq_A V_{\pi'}$. Then by Theorem \ref{thm:formal-relations} (A3) we have that  $v_\pi[v_{\pi'}] \ol{\tleq_B} V_\pi V_{\pi'}$ and
\begin{align*}
m'[n'] &= \sum_{\pi \in F} \sum_{\pi' \in F'} \left( P(\pi) \cdot P(\pi') \right) v_\pi [v_{\pi'}] & & \\
&= \sum_{(\pi,\pi') \in F \times F' } P(\mathrm{app}(\pi,\pi')) v_\pi [v_{\pi'}] & & \\
&\ol{\tleq_B} MN && \text{(Lemma \ref{lem:logical-convex-sum})}
\end{align*}
where $\mathrm{app}(\pi, \pi') \in \Paths(MN, V_\pi V_{\pi'})$ is the path where we first reduce $MN$ to $V_\pi N$ in the same way as in $\pi$
and then we reduce $V_\pi N$ to $V_\pi V_{\pi'}$ in the same way as in $\pi'.$
Note: in the above sum $V_\pi V_{\pi'}$ is not a value, so Lemma \ref{lem:logical-convex-sum} is crucial.
\end{proof}

\subsection{Fundamental Lemma and Strong Adequacy}

We may now prove the Fundamental Lemma which then easily implies our adequacy result.

\begin{lemma}[Fundamental]
\label{lem:fundamental}
Let $x_1 : A_1, \ldots, x_n : A_n \vdash M : B$ be a term. Assume further we are given a collection of morphisms $v_i$ and values $V_i$, such that $v_i \tleq_{A_i} V_i$ for $i \in \{1,\ldots,n\}.$
Then:
\[ \lrb M \kcirc \llangle\ \vec v\ \rrangle \ol{\tleq_B} M[ \vec V / \vec x] . \]
\end{lemma}
\begin{proof}
By induction on the derivation of the term $M$.

For the case of lambda abstractions, we reason as follows. Let us assume that the term of the induction hypothesis is
\[ x_1 : A_1, \ldots, x_n : A_n, y : A \vdash M : B. \]
Let us write $l \eqdef \lrb{\lambda y. M} \kcirc \llangle\ \vec v\ \rrangle$ and $R \eqdef \lambda y. M[ \vec V / \vec x]$.
Observe that $l \in \TD$ and therefore by Theorem \ref{thm:formal-relations} (C5), we may equivalently show that
\[ l  \tleq_{A \to B} R . \]
By Theorem \ref{thm:formal-relations} (A3), this is in turn equivalent to showing that
\[ l \leq \lrb{R} \text{ and } \forall (w \tleq_{A} W).\  l[w] \ol{\tleq_B} RW.  \]
The inequality is satisfied, because
\begin{align*}
l &= \lrb{\lambda y. M} \kcirc \llangle\ \vec v\ \rrangle & & \\
& \leq \lrb{\lambda y. M} \kcirc \llangle\ \vec{\lrb V}\ \rrangle & & (\text{Theorem \ref{thm:formal-relations} (C1)} ) \\
& = \lrb R . & & (\text{Lemma \ref{lem:term-substitution}} )
\end{align*}
For the other requirement, assuming that $w \tleq_{A} W$, we reason as follows
\begin{align*}
l[w]  &= (\lrb{\lambda y. M} \kcirc \llangle\ \vec v\ \rrangle)[w] & & (\text{Definition})\\
&= \epsilon \kcirc (\lrb{\lambda y. M} \ktimes \kid) \kcirc \llangle \vec v, w \rrangle & & \\
&= \epsilon \kcirc (\JJ \lambda(\lrb M) \ktimes \kid) \kcirc \llangle \vec v, w \rrangle & & (\text{Definition}) \\
&= \lambda^{-1}(\lambda(\lrb M)) \kcirc \llangle \vec v, w \rrangle & & \text{(Property of adjunction \eqref{eq:currying-simple})}\\
&= \lrb M \kcirc \llangle \vec v, w \rrangle & & \\
& \ol{\tleq_B} M[\vec V/ \vec x, W / y]  . & & \text{(Induction Hypothesis)}
\end{align*}
Finally, observe that $RW = (\lambda y. M[ \vec V / \vec x])W \probto{1} M[\vec V/ \vec x, W / y] ,$ i.e. $RW$ beta-reduces to $M[\vec V/ \vec x, W / y]$. Therefore by Lemma \ref{lem:logical-convex-sum} it follows that 
\[ l[w] \ol{\tleq_B} RW , \]
as required.

The case for variables follows immediately by expanding definitions and Theorem \ref{thm:formal-relations} (C5).

All other cases follow by straightforward induction using closure Lemmas \ref{lem:logical-probabilistic-choice} -- \ref{lem:logical-aplication}.
\end{proof}

Adequacy now follows as a corollary of this lemma.

\begin{theorem}[Strong Adequacy]
For any closed term $\cdot \vdash M : A$, we have
\[ \lrb M = \sum_{V \in \Val(M)} P(M \probto{}_* V) \lrb V . \]
\end{theorem}
\begin{proof}
Let 
\[ u \eqdef \sum_{V \in \Val(M)} P(M \probto{}_* V) \lrb V . \]
From Corollary \ref{cor:soundness-inequality}, we know that $\lrb M \geq u.$ To finish the proof, we have to show the converse inequality.
Next, observe that $\mathcal S_0(\tleq_A ; M) \subseteq \down u$, which follows from Theorem \ref{thm:formal-relations} (C1).
To see this, we reason as follows. Taking an arbitrary element of $\mathcal S_0(\tleq_A ; M)$ as in Theorem \ref{thm:formal-relations} (B):
\begin{align*}
\sum_{\pi \in F} P(\pi) v_\pi &\leq \sum_{\pi \in F} P(\pi) \lrb{V_\pi} & & (\text{Theorem \ref{thm:formal-relations} (C1)})  \\
&= \sum_{V \in \cup\{V_\pi | \pi \in F\}} \left( \sum_{\substack{\pi \in F \\ V_\pi = V}} P(\pi) \right) \lrb{V} & & \\
&\leq \sum_{V \in \cup\{V_\pi | \pi \in F\}} \left( \sum_{\pi \in \Paths(M,V) } P(\pi) \right) \lrb{V} & & \\
&= \sum_{V \in \cup\{V_\pi | \pi \in F\}} P(M \probto{}_* V) \lrb{V} & & \\
&\leq \sum_{V \in \Val(M)} P(M \probto{}_* V) \lrb{V} . & & \\
\end{align*}

The set $\down u $ is Scott-closed and therefore $\mathcal S(\tleq_A ; M) \subseteq \down u $. By Lemma \ref{lem:fundamental}, we know
that $\lrb M \ol{\tleq_A} M.$ By definition of $\ol{\tleq_A}$ it follows $\lrb M \in \mathcal S(\tleq_A ; M)$ and therefore $\lrb M \leq u,$ thus finishing the proof.
\end{proof}

\end{document}